\documentclass[10pt]{article}
\usepackage{amssymb}
\usepackage{amsmath}
\usepackage{amsthm}
\usepackage{latexsym}
\usepackage[dvips]{epsfig}
\usepackage{wasysym}
\usepackage{mathrsfs}
\usepackage{eufrak}
\usepackage{bm}
\usepackage{authblk}
\usepackage{slashed}
\usepackage{yhmath} 
\usepackage{stmaryrd}
\usepackage{tikz}

\theoremstyle{plain}
\newtheorem{proposition}{Proposition}

\newtheorem{theorem}{Theorem}

\newtheorem*{main}{Theorem}

\def\bma{{\bm a}}
\def\bmb{{\bm b}}
\def\bmc{{\bm c}}
\def\bmd{{\bm d}}
\def\bme{{\bm e}}
\def\bmf{{\bm f}}
\def\bmg{{\bm g}}
\def\bmh{{\bm h}}
\def\bmi{{\bm i}}
\def\bmj{{\bm j}}
\def\bmk{{\bm k}}
\def\bml{{\bm l}}
\def\bmn{{\bm n}}
\def\bmm{{\bm m}}

\def\bmu{{\bm u}}

\def\bmx{{\bm x}}

\def\bmzero{{\bm 0}}
\def\bmone{{\bm 1}}
\def\bmtwo{{\bm 2}}
\def\bmthree{{\bm 3}}

\def\bmA{{\bm A}}
\def\bmB{{\bm B}}
\def\bmC{{\bm C}}
\def\bmD{{\bm D}}
\def\bmE{{\bm E}}
\def\bmF{{\bm F}}
\def\bmG{{\bm G}}
\def\bmH{{\bm H}}
\def\bmK{{\bm K}}

\def\bmP{{\bm P}}
\def\bmQ{{\bm Q}}

\def\bmS{{\bm S}}

\def\bmalpha{{\bm \alpha}}

\def\bmomega{{\bm \omega}}

\def\bmphi{{\bm \phi}}

\def\bmpartial{{\bm \partial}}
\def\bmnabla{{\bm \nabla}}
\def\bmhbar{{\bm \hbar}}

\newcounter{mnote}

\begin{document}

\title{\textbf{Spinorial Wave Equations and Stability of the Milne Spacetime }}

\author[,1]{Edgar Gasper\'in\footnote{E-mail address:{\tt e.gasperingarcia@qmul.ac.uk}}}
\author[,1]{Juan Antonio Valiente Kroon \footnote{E-mail address:{\tt j.a.valiente-kroon@qmul.ac.uk}}}
\affil[1]{School of Mathematical Sciences, Queen Mary, University of London,
Mile End Road, London E1 4NS, United Kingdom.}

\maketitle

\begin{abstract}
The spinorial version of the conformal vacuum Einstein field equations
are used to construct a system of quasilinear wave equations for the
various conformal fields. As a part of analysis we also show how to
construct a subsidiary system of wave equations for the zero
quantities associated to the various conformal field equations. This
subsidiary system is used, in turn, to show that under suitable
assumptions on the initial data a solution to the wave equations for
the conformal fields implies a solution to the actual conformal
Einstein field equations. The use of spinors allows for a more unified
deduction of the required wave equations and the analysis of the
subsidiary equations than similar approaches based on the metric
conformal field equations. As an application of our construction we
study the non-linear stability of the Milne Universe. It is shown that sufficiently
small perturbations of initial hyperboloidal data for the Milne
Universe gives rise to a solution to the Einstein field equations
which exist towards the future and has an asymptotic structure similar
to that of the Milne Universe.
\end{abstract}

\textbf{Keywords:} Conformal methods, spinors, wave equations, Milne
Universe, global existence

\medskip
\textbf{PACS:} 04.20.Ex, 04.20.Ha, 04.20.Gz

\section{Introduction}
\label{Introduction}

The conformal Einstein field equations (CEFE) constitute a powerful
tool for the global analysis of spacetimes ---see
e.g. \cite{Fri81a,Fri81b,Fri82,Fri86a,Fri86b,Fri91}. The CEFE
provide a system of field equations for geometric objects defined on a
Lorentzian manifold $(\mathcal{M},\bmg)$ (the so-called
\emph{unphysical spacetime}) which is conformally related to a to a
spacetime $(\tilde{\mathcal{M}},\tilde{\bmg})$ (the so-called
\emph{physical spacetime}) satisfying the (vacuum) Einstein field
equations. The metrics $\bmg$ and $\tilde{\bmg}$ are related to each
other via a rescaling of the form $\bmg = \Xi^2 \tilde{\bmg}$ where
$\Xi$ is the so-called \emph{conformal factor}. The CEFE have the property of
being regular at the points where $\Xi=0$ (the so-called
\emph{conformal boundary}) and a solution thereof implies, wherever
$\Xi\neq 0$, a solution to the Einstein field equations. The great
advantage of the conformal point of view provided by the CEFE is that
it allows to recast global problems in the physical spacetime
$(\tilde{\mathcal{M}},\tilde{\bmg})$ as local problems in the
unphysical one $(\mathcal{M},\bmg)$. The CEFE have been extended to
include matter sources consisting of suitable trace-free matter ---see
e.g. \cite{Fri91,LueVal12,LueVal13a}. The CEFE can be expressed in terms of a \emph{Weyl
connection} (i.e. a connection which is not metric but nevertheless
preserves the conformal structure) to obtain a more general system of
field equations ---the so-called \emph{extended conformal Einstein
field equations}, see \cite{Fri95}. In what follows, the conformal field
equations expressed in terms of the Levi-Civita connection of the
metric $\bmg$ will be known as the \emph{standard CEFE}. The analysis
of the present article is restricted to this version of the CEFE. The standard
CEFE can be read as differential conditions on the conformal factor
and some concomitants thereof: the Schouten tensor, the rescaled Weyl
tensor and the components of the unphysical metric $\bmg$ ---this
version of the equations is known as the \emph{metric
CEFE}. Alternatively, by supplementing the field equations with the
Cartan structure equations, one can replace the metric components by
the coefficients of a frame and the associated connection coefficients
as unknowns. This \emph{frame version} of the equations allows a
direct translation of the CEFE into a spinorial formalism ---the
so-called \emph{spinorial CEFE}.

\medskip
In view of the tensorial nature of the CEFE, in order to make
assertions about the existence and properties of their solutions, it is
necessary to derive from them a suitable evolution system to which
the theory of hyperbolic partial differential equations can be
applied. This procedure is known as a \emph{hyperbolic reduction}. Part
of the hyperbolic reduction procedure consists of a specification of
the gauge inherent to the equations. A systematic way of proceeding to
the specification of the gauge is through so-called \emph{gauge source
  functions}. These functions are  associated to derivatives
of the field unknowns which are not determined by the field
equations. This idea can be used to extract a first order symmetric
hyperbolic system of equations for the field unknowns for the metric,
frame and spinorial versions of the standard CEFE. More recently, it
has been shown that gauge source functions can be used to obtain, out
of the metric conformal field equations, a system of quasilinear wave
equations ---see \cite{Pae13}. This particular construction requires
the specification of a \emph{coordinate gauge source function} and a
\emph{conformal gauge source function} and is close, in spirit, to the
classical treatment of the Cauchy problem in General Relativity in
\cite{Fou52} ---see also \cite{Cho08}.

\medskip
In the present article we show how to deduce a system of quasilinear
wave equations for the unknowns of the spinorial CEFE and analyse its
relation to the original set of field equations. The use of the
spinorial CEFE (or, in fact, the frame CEFE) gives access to a wider
set of gauge source functions consisting of \emph{coordinate, frame
  and conformal gauge source functions}. Another advantage of the
spinorial version of the CEFE is that they have a much simpler
algebraic structure than the metric equations.  In fact, one of the
features of the spinorial formalism simplifying our analysis is the
use of the symmetric operator $\square_{AB} \equiv \nabla_{Q(A}
\nabla_{B)}{}^{Q}$ instead of the usual commutator of covariant
derivatives $[ \nabla_a,\nabla_b]$.  As shown in this article, the use
of spinors allows a more unified and systematic discussion of the
construction of the wave equations and the so-called \emph{subsidiary
  system} ---needed to show that under suitable conditions a solution to
the wave equations implies a solution to the CEFE. As already
mentioned, in the spinorial formulation of the CEFE the metric is not
part of the unknowns.  This observation is important since, whenever
the wave operator $\square \equiv \nabla_a\nabla^{a}$ is applied to
any tensor of rank one or more, there will appear derivatives of the
connection which, in terms of the metric, represent second order
derivatives.  Thus, if the metric is part of the unknowns, the
principal part of the operator $\square$ is altered by the presence
of these derivatives. This is an extra complication that needs to be
taken into account in the analysis of \cite{Pae13}.

The construction of wave equations to the fields of the CEFE gives
access to a set of methods of the theory of partial differential
equations alternative to that used for first order symmetric
hyperbolic systems ---see e.g. \cite{Ren08} for a discussion on
this. For example, the analysis in \cite{Pae13} is motivated by the
analysis of the characteristic problem on a cone for which a detailed
theory is available for quasilinear wave equations.

\medskip
As an application of the hyperbolic reduction procedure put forward in
the present article, we provide an analysis of the non-linear
stability of the Milne Universe. The Milne Universe is a
spatially flat Friedman-Lema\^{i}tre-Robertson-Walker (FLRW)  solution to the
Einstein field equations with vanishing cosmological constant ---see
e.g. \cite{GriPod09}. The Milne Universe can be seen to be a part of the
Minkowski spacetime written in comoving coordinates adapted to the
world line of a particle. Accordingly, analysing the non-linear
stability of the Milne Universe is essentially equivalent to obtaining
a proof of the the semiglobal stability of the Minkowski spacetime
---see \cite{Fri86a}. The stability of the Milne Universe has been
analysed by different methods in \cite{AndMon03} ---see also \cite{And04b}.  In our
case the stability result follows from the theory of quasilinear wave
equations, in particular the property of Cauchy stability, as given in \cite{HugKatMar77}. In
broad terms, our stability result for the Milne Universe can be phrased
as:

\begin{main}
Initial data for the conformal wave equations close
enough to the data for the Milne Universe give rise to a solution to
the Einstein field equations which exist globally to the future and
has the an asymptotic structure similar to that of the Milne Universe.
\end{main}

\subsection*{Outline of the article}

This article is organised as follows: in Section \ref{Section:CEFE} we
briefly present the spinorial version of the conformal Einstein field
equations which will be the starting point of the analysis. In Section
\ref{Section:WaveEquations}, we give the derivation of the wave
equations. For the sake of clarity, we first present a general
procedure and then discuss the peculiarities of each equation. In
Section \ref{PropagationConstraintsSubsidiarySystem}, we derive the
subsidiary system and analyse the propagation of the contraints. In
Section \ref{Milne}, we give a semiglobal existence and stability
result for the Milne spacetime as an application of the equations
derived in Section \ref{Section:WaveEquations}. In Section
\ref{Section:Conclusions} some concluding remarks are presented. In
order to ease the presentation of the article, we have moved part of
the calculations to a series of appendices. In Appendix
\ref{Appendix:SpinorialRelations} we recall some general spinor
relations and deduce the general form for the spinorial Ricci
identities for a connection which is metric but not necessarily
torsion free.  In Appendix \ref{Appendix:Torsion} a brief discussion
of the relation between the transition spinor and the torsion is
provided. In Appendix \ref{Appendix:SubsidiarySystem} detailed
computations for the subsidiary system are presented. In Appendix
\ref{Appendix:CFE} we recall the the frame version of the CEFE.  In
Appendix \ref{Appendix:PDETheory} an adapted version of a existence
and Cauchy stability result for wave equations is given.

\subsection*{Notations and conventions}
 The signature convention for  (Lorentzian) spacetime metrics will be $
 (+,-,-,-)$.  In what follows $\{_a ,_b , _c , . . .\}$ will be used
 as tensor indices and $\{_\bma ,_\bmb , _\bmc , . . .\}$ will be used
 as spacetime frame indices taking the values ${ 0, . . . , 3 }$. In
 this manner, given a basis $\{\bme_{\bma}\}$ a tensor will be denoted
 with $T_{ab}$ while its components in the given basis will be denoted
 by $T_{\bma \bmb}\equiv T_{ab}\bme_{\bma}{}^{a}\bme_{\bmb}{}^{b}$
 . Most of the analysis will require the use of spinor
 and $\{ _\bmA , _\bmB , _\bmC , . . .\}$ will denote frame
 spinorial indices with respect to some specified spin dyad ${
   \{\delta_\bmA{}^{A} \} }.$ We will follow the conventions and notation
 of Penrose \& Rindler \cite{PenRin84}.  In addition,
 $D^{+}(\mathcal{A})$, $H(\mathcal{A})$, $J^{+}(\mathcal{A})$ and
 $I^{+}(\mathcal{A})$ will denote the future domain of dependence, the
 Cauchy horizon, Causal and Chronological future of $\mathcal{A}$, 
 respectively ---see \cite{HawEll73,Wal84}.

\section{The conformal Einstein field equations in spinorial form}
\label{Section:CEFE}

In \cite{Pen63} Penrose introduced a geometric technique
to study the far fields of isolated gravitational systems in which,
given a spacetime $(\tilde{\mathcal{{M}}},\tilde{g}_{ab})$ satisfying
the Einstein field equations (\textit{the physical spacetime}), we
consider another spacetime $(\mathcal{M},g_{ab})$ (\textit{the
  unphysical spacetime}) such that
\begin{equation}\label{EinsteinFE}
g_{ab}=\Xi^2 \tilde{g}_{ab}
\end{equation}
\noindent where $\Xi$ is a scalar field (the \emph{conformal
  factor}). If we try to derive an equation for the unphysical metric
$g_{ab}$ using the Einstein field equations \[\tilde{R}_{ab}=\lambda
\tilde{g}_{ab}\] making the straightforward computation using the
conformal transformation rules for the curvature tensors implied by
\eqref{EinsteinFE} we will obtain an expression which is singular at
the conformal boundary, i.e. the points for which $\Xi=0$. An approach
to deal with this problem was given in \cite{Fri81a} where a
regular set of equations for the unphysical metric was derived. These
equations are known as the \emph{conformal Einstein field equations}
(CEFE). In order to recall the (vacuum) spinorial version of these equations
let us introduce a frame $\{\bme_{\bmA \bmA'}\}$ and an associated 
spinorial dyad $\{\epsilon_{\bmA}{}^{A}\}$. Let $\Gamma_{\bmA
  \bmA'}{}^{\bmC\bmC'}{}_{\bmD\bmD'} $ denote the spinorial
counterpart of the connection coefficients of the Levi-Civita
connection $\nabla$ respect to the metric $\bmg$. This spinor can be
decomposed as 
\[ 
\Gamma_{\bmA \bmA'}{}^{\bmC \bmC'}{}_{\bmB \bmB'} \equiv \Gamma_{\bmA
  \bmA'}{}^{\bmC}{}_{\bmB}\delta_{\bmB'}{}^{\bmC'} +
\bar{\Gamma}_{\bmA \bmA'}{}^{\bmC'}{}_{\bmB'}\delta_{\bmB}{}^{\bmC} 
\]
\noindent where
\begin{equation} 
\Gamma_{\bmA \bmA'}{}^{\bmC}{}_{\bmB}\equiv \tfrac{1}{2}\Gamma_{\bmA
  \bmA'}{}^{\bmC \bmQ'}{}_{\bmB \bmQ'},
\label{ReducedConnection}
\end{equation}
are the reduced spin connection coefficients.  Let
$\phi_{\bmA \bmB \bmC \bmD} \equiv \Xi^{-1}\Psi_{\bmA \bmB \bmC \bmD}$
represent the so-called \emph{rescaled Weyl spinor}  ---$\Psi_{\bmA \bmB \bmC
  \bmD}$ is the Weyl spinor and $\Xi$ is the conformal factor. In the
following $L_{\bmA \bmA' \bmB \bmB'}$ will denote the \emph{spinorial
counterpart of the Schouten tensor} of the metric $\bmg$.  Likewise,
$s$ will denote a concomitant of the conformal factor defined by $s
\equiv \frac{1}{4}\square \Xi - \Lambda \Xi$ with $\Lambda
\equiv-24R$, where $R$ is the Ricci scalar of the metric $\bmg$. In
addition, $R^{\bmC}{}_{\bmD \bmA \bmA' \bmB \bmB'}$ and
$\rho^{\bmC}{}_{\bmD \bmA \bmA' \bmB \bmB'}$ will represent shorthands
for the \textit{geometric curvature} and the \textit{algebraic
  curvature}, given explicitly as
\begin{subequations}
\begin{eqnarray}
 && R^{\bmC}{}_{\bmD \bmA \bmA' \bmB \bmB'} \equiv
 \bme_{\bmA \bmA'}\left( \Gamma_{\bmB \bmB'}{}^{\bmC}{}_{\bmD} \right)
 -\bme_{\bmB \bmB'} \left( \Gamma_{\bmA \bmA'}{}^{\bmC}{}_{\bmD}
 \right) \nonumber \\
&&  \hspace{3cm}- \Gamma_{\bmF \bmB'}{}^{\bmC}{}_{\bmD}\Gamma_{\bmA
  \bmA'}{}^{\bmF}{}_{\bmB}-\Gamma_{\bmB
  \bmF'}{}^{\bmC}{}_{\bmD}\bar{\Gamma}_{\bmA
  \bmA'}{}^{\bmF'}{}_{\bmB'} +\Gamma_{\bmF \bmA'}{}^{\bmC}{}
_{\bmD}\Gamma_{\bmB\bmB'}{}^{\bmF}{}_{\bmA}   \nonumber \\
&& \hspace{3cm} + \Gamma_{\bmA
   \bmF'}{}^{\bmC}{}_{\bmD}\bar{ \Gamma}_{\bmB
   \bmB'}{}^{\bmF'}{}_{\bmA'} + \Gamma_{\bmA \bmA'}{}^{\bmC}{}_{\bmE}\Gamma_{\bmB \bmB'}{}^{\bmE}{}_{\bmD}-\Gamma_{\bmB \bmB'}{}^{\bmC}{}_{\bmE}\Gamma_{\bmA \bmA'}{}^{\bmE}{}_{\bmD},    \\
&&\rho_{\bmA \bmB \bmC \bmC' \bmD \bmD'}\equiv\Psi_{\bmA \bmB \bmC \bmD}\epsilon_{\bmC' \bmD'}+ L_{\bmB \bmC' \bmD \bmD'}\epsilon_{\bmC \bmA}-L_{\bmB \bmD' \bmC \bmC'}\epsilon_{\bmD \bmA}. \label{AlgebraicCurvatureDecomp}
\end{eqnarray}
\end{subequations}

\medskip
Now, let us define the following spinors which will be collectively
called \emph{zero-quantities}:
\begin{subequations}
\begin{flalign} & \Sigma_{\bmA \bmA'}{}^{\bmQ \bmQ'}{}_{\bmB
\bmB'}\bme_{\bmQ\bmQ'} \equiv \left[\bme_{\bmB \bmB'},\bme_{\bmA
\bmA'}\right] -\Gamma_{\bmB \bmB'}{}^{\bmQ}{}_{\bmA}\bme_{\bmQ
\bmA'}-\bar{\Gamma}_{\bmB \bmB'}{}^{\bmQ'}{}_{\bmA'}\bme_{\bmA \bmQ'}
\nonumber \\ & \hspace{8cm} + \Gamma_{\bmA
\bmA'}{}^{\bmQ}{}_{\bmB}\bme_{\bmQ \bmB'} + \bar{\Gamma}_{\bmA
\bmA'}{}^{\bmQ'}{}_{\bmB'}\bme_{\bmB \bmQ'},& \label{ZeroTorsion} \\ &
\Delta_{\bmC \bmD \bmB \bmB'} \equiv
\nabla_{(\bmC}{}^{\bmQ'}L_{\bmD)\bmQ' \bmB \bmB'}+ \phi_{\bmC \bmD
\bmB \bmQ}\nabla^{\bmQ}{}_{\bmB'}\Xi, \label{ZeroDelta} & \\ &
\Lambda_{\bmB \bmB' \bmC \bmD} \equiv
\nabla^{\bmQ}{}_{\bmB'}\phi_{\bmA \bmB \bmC \bmQ}, \label{ZeroLambda}
& \\ & Z_{\bmA \bmA' \bmB \bmB'} \equiv \nabla_{\bmA
\bmA'}\nabla_{\bmB \bmB'}\Xi + \Xi L_{\bmA \bmA' \bmB
\bmB'}-s\epsilon_{\bmA \bmB}\epsilon_{\bmA' \bmB'}, \label{ZeroZ} & \\
& \Xi^{\bmC}{}_{\bmD \bmA \bmA' \bmB \bmB'} \equiv R^{\bmC}{}_{\bmD
\bmA \bmA' \bmB \bmB'}-\rho^{\bmC}{}_{\bmD \bmA \bmA' \bmB
\bmB'}, \label{ZeroGeometricAlgebraicCurvature} & \\ & Z_{\bmA \bmA'}
\equiv \nabla_{\bmA \bmA'}s+ L_{\bmA \bmA' \bmC \bmC'}\nabla^{\bmC
\bmC'}\Xi. & \label{ZeroDefinitionofs}
\end{flalign}
\end{subequations}
In terms of these definitions, the \emph{spinorial version of the
CEFE} can be succinctly expressed as
\begin{subequations}
\begin{eqnarray}
& \Sigma_{\bmA \bmA'}{}^{\bmQ \bmQ'}{}_{\bmB \bmB'}\bme_{\bmQ\bmQ'} =0,
 \qquad  \Delta_{\bmC \bmD \bmB \bmB'}  =0,  \qquad   \Lambda_{\bmB \bmB' \bmC \bmD} =0, &  \\
 &Z_{\bmA \bmA' \bmB \bmB'}  =0,  \qquad   \Xi^{\bmC}{}_{\bmD \bmA \bmA' \bmB \bmB'} =0, \qquad Z_{\bmA \bmA'} =0.&
\end{eqnarray} \label{CEFEzeroquantities}
\end{subequations}

\medskip
\noindent
\textbf{Remark}. Observe that all the equations are first order
except for the fourth equation, which can be written as a first order
equation by defining $\Sigma_{\bmA\bmA'}\equiv \nabla_{\bmA
  \bmA'}\Xi$. Also, notice that the trace $Z = Z_{\bmA \bmA'}{}^{\bmA \bmA'}$
renders the definition of the field $s$. More details about the CEFE
and its derivation can be found in \cite{Fri83,Fri91}.

\section{The spinorial wave equations}
\label{Section:WaveEquations}
In this section a set of wave equations  is derived from the spinorial version of
the CEFE. Since the approach for obtaining the equations is
 similar for most of the zero-quantities, we first provide a general
discussion of the procedure. In the subsequent parts of this section we
address the peculiarities of each equation. The results of this
section are summarised in Proposition \ref{CEWE}.

\subsection{General procedure for obtaining the wave equations}
\label{Section:WaveGeneralProcedure}
Before deriving each of the wave equations let us illustrate the general procedure
with a model equation. To this end consider an equation of the form 
\begin{equation}
\nabla^{\bmE}{}_{\bmA'}N_{\bmE\bmA \mathcal{K}}=0,
\label{ModelEquation}
\end{equation} 
where $ N_{\bmE \bmA \mathcal{K}} \equiv
\nabla_{(\bmE}{}^{\bmB'}M_{\bmA) \bmB'\mathcal{K}} $ and
${}_\mathcal{K}$ is an arbitrary string of spinor indices. We can
exploit the symmetries using the following decomposition of a spinor
of the same index structure
 \[
T_{\bmE \bmA \mathcal{K}} = T_{(\bmE\bmA)\mathcal{K}}+
\tfrac{1}{2}\epsilon_{\bmE \bmA}T_{\bmQ}{}^{\bmQ}{}_{\mathcal{K}},
\]
and recast $N_{\bmE \bmA\mathcal{K}}$ as  
\[
N_{\bmE
  \bmA\mathcal{K}}=\nabla_{\bmE}{}^{\bmB'}M_{\bmA \bmB' \mathcal{K}}+
\frac{1}{2}\epsilon_{\bmE \bmA}\nabla^{\bmQ\bmB'}M_{\bmQ \bmB'
  \mathcal{K}}.
\] 
Observe that the model equation \eqref{ModelEquation} determines the
symmetrised derivative $ \nabla_{(\bmE}{}^{\bmB'}M_{\bmA) \bmB'
\mathcal{K}}$, while the divergence $\nabla^{\bmQ \bmB'}M_{\bmQ \bmB'
\mathcal{K}}$ can be freely specified. Thus, let $ F_{\mathcal{K}}(x)
\equiv \nabla^{\bmQ \bmB'}M_{\bmQ \bmB' \mathcal{K}}$ be a smooth but
otherwise arbitrary spinor. This spinor, encoding the freely
specifiable part of $N_{\bmE \bmA \mathcal{K}}$, is the \textit{gauge
source function} for our model equation. Taking this discussion into
account, the model equation can be reexpressed as
\begin{flalign} 
\nabla^{\bmE}{}_{\bmA'}N_{\bmE \bmA} & =
\nabla^{\bmE}{}_{\bmA'}\nabla_{\bmE}{}^{\bmB'}M_{\bmA \bmB'
\mathcal{K}} + \frac{1}{2}\nabla_{\bmA \bmA'}F_{\mathcal{K}}(x)
\nonumber \\ & =\nabla_{\bmE(\bmA'
}\nabla_{\bmB')}{}^{\bmE}M_{\bmA}{}^{\bmB'}{}_{\mathcal{K}}
+\tfrac{1}{2}\epsilon_{\bmA' \bmB'}\nabla_{\bmE
\bmQ'}\nabla^{\bmE\bmQ'}M_{\bmA}{}^{\bmB'}{}_{\mathcal{K}} +
\frac{1}{2}\nabla_{\bmA
\bmA'}F_{\mathcal{K}}(x)=0 \label{protoModelWave}
\end{flalign}

\noindent where in the second row we have used, again, the
decomposition of a 2-valence spinor in its symmetric and trace
parts. Finally, recalling the definition of the operators 
\[\square
\equiv \nabla_{\bmA \bmA'}\nabla^{\bmA \bmA'}, \hspace{0.5cm} \text{}
\hspace{0.5cm}\square_{\bmA
\bmB}=\nabla_{\bmQ'(\bmA}\nabla_{\bmB)}{}^{\bmQ'}, 
\]
 we rewrite equation \eqref{protoModelWave} as
\begin{equation}
\square M_{\bmA \bmA' \mathcal{K}} - 2
\square_{\bmA'\bmB'}M_{\bmA}{}^{\bmB'}{}_{ \mathcal{K}} - \nabla_{\bmA
  \bmA'}F_{\mathcal{K}}(x)=0. \label{ModelWave} 
\end{equation}
 The spinorial Ricci identities allow us to rewrite
$\square_{\bmA'}{}^{\bmB'}M_{\bmA \bmB' \mathcal{K}}$ in terms of the
curvature spinors ---namely, the Weyl spinor $\Psi_{\bmA \bmB \bmC
\bmD}=\Xi\phi_{\bmA \bmB \bmC \bmD}$, the Ricci spinor $\Phi_{\bmA
\bmB \bmA' \bmB'}$, and the Ricci scalar $\Lambda=-24R$--- and $M_{\bmA
\bmB' \mathcal{K}}$.
 
\medskip
 In the rest of the section we will derive the particular wave
equations implied by each of the zero-quantities following an
analogous procedure as the one used for the model equation.

\subsection{Wave equation for the frame (no-torsion condition)}
\label{Waveframe}

The zero-quantity $\Sigma_{\bmA \bmA'}{}^{\bmQ \bmQ'}{}_{\bmB \bmB'}$
encodes the no-torsion condition. The equation \eqref{ZeroTorsion} can
be conveniently rewritten introducing an arbitrary frame
$\{\bmc_{\bma}\}$, which allow us to write $\bme_{\bmA \bmA'} =
e_{\bmA \bmA'}{}^{\bma}\bmc_{\bma} $ . Taking this into account we
rewrite the zero-quantity for the no-torsion condition as
\begin{equation}
\Sigma_{\bmA \bmA'}{}^{\bmQ \bmQ'}{}_{\bmB \bmB'}e_{\bmQ
  \bmQ'}{}^{\bmc} = {\nabla}_{\bmB \bmB'} ( e_{\bmA \bmA'}{}^{\bmc}) -
{\nabla}_{\bmA \bmA'} (e_{\bmB \bmB'}{}^{\bmc}) -
C_{\bma}{}^{\bmc}{}_{\bmb} e_{\bmA \bmA'}{}^{\bma} e_{\bmB
  \bmB'}{}^{\bmb},
\label{noTorsionCommutation}
\end{equation}
where $C_{\bma}{}^{\bmc}{}_{\bmb}$ are the commutation coefficients
of the frame, defined by the relation
$[\bmc_{\bma},\bmc_{\bmb}]=C_{\bma}{}^{\bmc}{}_{\bmb}\bmc_{\bmc}$. For
conciseness, we have introduced the notation $\nabla_{\bmA
  \bmA'}e_{\bmB \bmB'}{}^{\bmc}$ which is to be interpreted as a 
shorthand for the longer expression 
\[
\nabla_{\bmA \bmA'}e_{\bmB \bmB'}{}^{\bmc} \equiv e_{\bmA
  \bmA'}{}^{\bmc}-\Gamma_{\bmA \bmA'}{}^{\bmQ}{}_{\bmB}e_{\bmQ
  \bmB'}{}^{\bmc}-\bar{\Gamma}_{\bmA \bmA'}{}^{\bmQ'}{}_{\bmB'}e_{\bmB
  \bmQ'}{}^{\bmc}.
\] 
Using the irreducible decomposition of a spinor representing an
antisymmetric tensor we obtain that
\begin{equation}
\Sigma_{\bmA \bmA'}{}^{\bmQ \bmQ'}{}_{\bmB \bmB'}e_{\bmQ
  \bmQ'}{}^{\bmc} = \epsilon_{\bmA \bmB} \bar{\Sigma}_{\bmA'
  \bmB'}{}^{\bmc} +\epsilon_{\bmA' \bmB'} \Sigma_{\bmA
  \bmB}{}^{\bmc}\label{splitTorsion}
\end{equation}
where  
\[ 
\Sigma_{\bmA \bmB}{}^{\bmc} \equiv \frac{1}{2}\Sigma_{(\bmA
  |\bmD'|}{}^{\bmQ \bmQ'}{}_{\bmB)}{}^{\bmD'}e_{\bmQ \bmQ'}{}^{\bmc}
 \] 
is a reduced zero-quantity which can be written in terms of the frame
coefficients using equation \eqref{noTorsionCommutation} as
\[
\Sigma_{\bmA \bmB}{}^{\bmc} ={\nabla}_{(\bmA}{}^{\bmD'}e_{\bmB)
\bmD'}{}^{\bmc} + \tfrac{1}{2}e_{(\bmA}{}^{\bmD'
\bma}e_{\bmB)\bmD'}{}^{\bmb}C_{\bma}{}^{\bmc}{}_{\bmb}.
\]
Using the decomposition of a valence-2 spinor in the first term of the
right-hand side we get
\[
\Sigma_{\bmA \bmB}{}^{\bmc}= {\nabla}_{\bmA}{}^{\bmD'}e_{\bmB
\bmD'}{}^{\bmc}+ \tfrac{1}{2}\epsilon_{\bmA \bmB} {\nabla}^{\bmP
\bmD'}e_{\bmP \bmD'}{}^{\bmc} + \tfrac{1}{2}e_{(\bmA}{}^{\bmD'
\bma}e_{\bmB)\bmD'}{}^{\bmb}C_{\bma}{}^{\bmc}{}_{\bmb}.
\]
Introducing the \textit{coordinate gauge source function}
$F^{\bmc}(x)={\nabla}^{\bmP \bmD'}e_{\bmP \bmD'}{}^{\bmc}$, a wave
equation can then be deduced from the
condition
\[{
\nabla}^{\bmA}{}_{\bmE'}\Sigma_{\bmA \bmB}{}^{\bmc} =0. 
\]
Observe that this equation is satisfied if $\Sigma_{\bmA
\bmB}{}^{\bmc}=0$ ---that is, if the corresponding CEFE is
satisfied. Adapting the general procedure described in Section
~\ref{Section:WaveGeneralProcedure} as required, we get
\begin{equation*}
{\square} e_{\bmB \bmE'}{}^{\bmc}-2{\square}_{\bmE' \bmD'}e_{\bmB}{}^{\bmD'\bmc}-{\nabla}_{\bmB \bmE'}F^{\bmc}(\bmx)-{\nabla}^{\bmA}{}_{\bmE'} \left( e_{(\bmA}{}^{\bmD' \bma}e_{\bmB)\bmD'}{}^{\bmb}C_{\bma}{}^{\bmc}{}_{\bmb} \right)=0. 
\end{equation*}
Finally, using the spinorial Ricci identities and rearranging the last term we get the following wave equation
\begin{multline}
 {\square} e_{\bmB \bmE'}{}^{\bmc}     -2 e^{\bmQ \bmD'
   \bmc}\bar{{\Phi}}_{\bmQ \bmB \bmE' \bmD'}
 +6{{\Lambda}}\hspace{1mm} e_{\bmB \bmE'}{}^{\bmc} -e_{(\bmA}{}^{\bmD'
   \bma}e_{\bmB)
   \bmD'}{}^{\bmb}{\nabla}^{\bmA}{}_{\bmE'}C_{\bma}{}^{\bmc}{}_{\bmb}
 \\
- 2 C_{\bma}{}^{\bmc}{}_{\bmb}e_{(\bmA}{}^{\bmD' \bma}
{\nabla}^{\bmA}{}_{|\bmE'|}e_{\bmB)\bmD'}{}^{\bmb} -{\nabla}_{\bmB
\bmE'}F^{\bmc}(x) =0. \label{WaveForframe}
\end{multline}

\subsection{Wave equation for the connection coefficients }

The spinorial counterpart of the Riemann tensor can be decomposed as
\[
R_{ \bmA \bmA' \bmB \bmB' \bmC \bmC' \bmD \bmD'} = R_{\bmA \bmB \bmC \bmC' \bmD \bmD'}\epsilon_{\bmB' \bmA'}+ \bar{R}_{\bmA' \bmB' \bmC \bmC' \bmD \bmD'}\epsilon_{\bmB \bmA}.
\]
where the \emph{reduced curvature spinor} $R_{\bmA \bmB \bmC \bmC'
\bmD \bmD'}$ is expressed in terms of the spin connection coefficients
as
\begin{multline} 
 R_{\bmA \bmB \bmC \bmC' \bmD \bmD'}+ \Sigma_{\bmC \bmC'}{}^{\bmQ
   \bmQ'}{}_{\bmD \bmD'}\Gamma_{\bmQ \bmQ' \bmA \bmB}  = \nabla_{\bmC
   \bmC'}\Gamma_{\bmD \bmD' \bmA \bmB}-\nabla_{\bmD \bmD'}\Gamma_{\bmC
   \bmC' \bmA \bmB}  \\  
 \hspace{6.5cm}  + \Gamma_{\bmC \bmC'}{}^{\bmQ}{}_{\bmB}\Gamma_{\bmD \bmD' \bmQ \bmA}  
-\Gamma_{\bmD \bmD'}{}^{\bmQ}{}_{\bmB }\Gamma_{\bmC \bmC' \bmQ \bmA}.
 \label{GeometricCurvature}
\end{multline} 
In the last equation, $\nabla_{\bmD \bmD'}\Gamma_{\bmC \bmC' \bmA
\bmB}$ has been introduced for convenience as a shorthand for the
longer expression
\[
\nabla_{\bmD \bmD'}\Gamma_{\bmC \bmC' \bmA \bmB} \equiv \bme_{\bmD
  \bmD'}\Gamma_{\bmC \bmC' \bmA \bmB} - \Gamma_{\bmD
  \bmD'}{}^{\bmQ}{}_{\bmC}\Gamma_{\bmQ\bmC' \bmA
  \bmB}-\bar{\Gamma}_{\bmD \bmD'}{}^{\bmQ}{}_{\bmC}\Gamma_{\bmQ \bmC'
  \bmA \bmB}-\Gamma_{\bmD \bmD'}{}^{\bmQ}{}_{\bmB}\Gamma_{\bmC \bmC'
  \bmA \bmQ}. 
\]
Now, observe that the zero quantity $\Xi_{\bmA \bmB \bmC \bmC' \bmD
  \bmD'}$ defined in equation \eqref{ZeroGeometricAlgebraicCurvature}
has the symmetry $\Xi_{\bmA \bmB\bmC \bmC' \bmD \bmD'}= \Xi_{(\bmA
  \bmB) \bmC \bmC' \bmD \bmD' }= -\Xi_{(\bmA \bmB) \bmD \bmD' \bmC
  \bmC' } $. Exploiting this fact, the reduced spinors  associated to
the geometric and algebraic curvatures $R_{\bmA \bmB \bmC \bmC' \bmD
  \bmD'}$ and $\rho_{\bmA \bmB \bmC \bmC'\bmD \bmD'}$ can be split, respectively, as 
\[
 R_{\bmA \bmB \bmC \bmC' \bmD \bmD'}=\epsilon_{\bmC' \bmD'}R_{\bmA
   \bmB \bmC \bmD} + \epsilon_{\bmC \bmD}R_{\bmA \bmB \bmC' \bmD'},
 \qquad
\rho_{\bmA \bmB \bmC \bmC' \bmD \bmD'} = \epsilon_{\bmC'
\bmD'}\rho_{\bmA \bmB \bmC \bmD} + \epsilon_{\bmC \bmD}\rho_{\bmA \bmB
\bmC' \bmD'},
\]
where
\[
R_{\bmA \bmB \bmC \bmD} = \tfrac{1}{2}R_{\bmA \bmB ( \bmC |\bmE'|\bmD)}{}^{\bmE'}, \qquad  R_{\bmA \bmB\bmC' \bmD'}= \tfrac{1}{2} R_{\bmA \bmB \bmE (\bmC'}{}^{\bmE}{}_{\bmD')},
\]  
are the \emph{reduced geometric curvature spinors}. Analogous
definitions are introduced for the algebraic
curvature\footnote{Observe that in contrast with the split
\eqref{splitTorsion} used for the no-torsion condition, the reduced
spinors $R_{\bmA \bmB \bmC \bmD}$ and $R_{\bmA \bmB \bmC' \bmD'}$ are
not complex conjugate of each other.}. The adjetive \emph{geometric} is used
here to emphasise the fact that $R_{\bmA \bmB \bmC \bmD}$ and $R_{\bmA
\bmB \bmC' \bmD'}$ are expressed in terms of the reduced connection
coefficients while $\rho_{\bmA \bmB \bmC \bmD}$ and $\rho_{\bmA \bmB
\bmC' \bmD'}$, the reduced algebraic curvature spinors, are written in
terms of the Weyl spinor $\Psi_{\bmA \bmB \bmC \bmD}$ and the
spinorial counterpart of the Schouten tensor $L_{\bmA \bmA'
\bmB\bmB'}$. Together, these two reduced geometric and algebraic
curvature spinors give the reduced zero quantities
\[ 
\Xi_{\bmA \bmB \bmC \bmD} = R_{\bmA \bmB \bmC \bmD}-\rho_{\bmA \bmB
  \bmC \bmD}, \hspace{1cm} \Xi_{\bmA \bmB \bmC' \bmD'}= R_{\bmA \bmB
  \bmC' \bmD'}-\rho_{\bmA \bmB \bmC' \bmD'}.
\]

\medskip
\noindent
\textbf{Remark.} Observe that although $R_{\bmA
\bmB \bmC \bmD}$ and $R_{\bmA \bmB \bmC' \bmD'}$ are independent,
their derivatives are related through the \emph{second Bianchi identity},
which implies that 
\[
\nabla^{\bmC}{}_{\bmD'}R_{\bmA \bmB \bmC \bmD} =
\nabla^{\bmC'}{}_{\bmD}R_{\bmA \bmB \bmC' \bmD'}.
\]
This observation is also true for the algebraic
curvature as a consequence of the conformal field equations $\Delta_{\bmC
  \bmD \bmB \bmB'}=0$ and $\Lambda_{\bmB \bmB' \bmC \bmD}=0$ since
they encode the second Bianchi identity written as differential
conditions on the spinorial counterpart of the Schouten tensor and the
Weyl spinor. To verify the last statement, recall that the equation
for the Schouten tensor encoded in $\Delta_{\bmC \bmD \bmB \bmB'}=0$
comes from the frame equation
\begin{equation}
 \nabla_{\bma}C^{\bma}{}_{ \bmb \bmc \bmd} = \nabla_{\bmc}L_{\bmd
   \bmb}-\nabla_{\bmd}L_{\bmc \bmb},
\end{equation}
which can be regarded as the second Bianchi identity written in terms
of the Schouten and Weyl tensors. This is can be easily checked,
since the last equation is obtained from the substitution of the
expression for the Riemann tensor in terms of the Weyl and Schouten
tensors (i.e. the algebraic curvature) in the second Bianchi
identity. This means that, as long as the conformal field equations
$\Delta_{\bmC \bmD \bmB \bmB'}=0$ and $\Lambda_{\bmB \bmB' \bmC
\bmD}=0$ are satisfied we can write
\[
\nabla^{\bmC}{}_{\bmD'}\rho_{\bmA \bmB \bmC \bmD} =
\nabla^{\bmC'}{}_{\bmD}\rho_{\bmA \bmB \bmC' \bmD'}.
\]
Therefore, the reduced quantities $\Xi_{\bmA \bmB \bmC \bmD}$ and
$\Xi_{\bmA \bmC \bmC' \bmD'}$ are related via
\[
\nabla^{\bmC}{}_{\bmD'}\Xi_{\bmA \bmB \bmC \bmD} =
\nabla^{\bmC'}{}_{\bmD}\Xi_{\bmA \bmB \bmC' \bmD'}.
\]

\medskip
Now, we compute explicitly the reduced geometric and algebraic
curvature.  Recalling the definition of $\rho_{\bmA \bmB \bmC \bmC'
\bmD \bmD'}$ in terms of the Weyl spinor and the spinorial counterpart
of the Schouten tensor as given in equation
\eqref{AlgebraicCurvatureDecomp} it follows that
\[ 
\rho_{\bmA \bmB \bmC \bmD}=\Psi_{\bmA \bmB \bmC \bmD} + L_{\bmB \bmE'
  (\bmD}{}^{\bmE'}\epsilon_{\bmC)\bmA}
\] 
or, equivalently
\[
\rho_{\bmA \bmB \bmC \bmD}=\Xi\phi_{\bmA \bmB \bmC \bmD} + 2\Lambda
(\epsilon_{\bmD\bmB}\epsilon_{\bmC \bmA}+\epsilon_{\bmC
\bmB}\epsilon_{\bmD \bmA}).
\]
Similarly,
\begin{equation}
 \rho_{\bmA \bmB \bmC' \bmD'} = \Phi_{\bmA \bmB \bmC' \bmD'}. \label{RedAlgebraicCurvatureRicciSpinor}
\end{equation}
Computing the reduced version of the geometric curvature from
expression \eqref{GeometricCurvature} we get
\begin{subequations}
\begin{eqnarray}
&& \hspace{-1cm} R_{\bmA \bmB \bmC \bmD} = -\tfrac{1}{2}\Sigma_{(\bmC
|\bmE'|}{}^{\bmQ \bmQ'}{}_{\bmD)}{}^{\bmE'}\Gamma_{\bmQ \bmQ' \bmA
\bmB} + \nabla_{(\bmC| \bmE'|} \Gamma_{\bmD)}{}^{\bmE'}{}_{\bmA\bmB} +
\Gamma_{(\bmC|\bmE'}{}^{\bmQ}{}_{\bmB|}{}\Gamma_{\bmD)} {}^{
\bmE'}{}_{ \bmQ \bmA}, \label{ReducedGeometricCurvature}
\\ 
&& \hspace{-1cm} R_{\bmA \bmB \bmC' \bmD'}= -\tfrac{1}{2} \Sigma_{\bmE(\bmC'}{}^{\bmQ
\bmQ' \bmE}{}_{\bmD')}\Gamma_{\bmQ \bmQ' \bmA \bmB}+ \nabla_{\bmE
(\bmC'}\Gamma_{\bmD')}{}^{\bmE}{}_{\bmA \bmB} +
\Gamma_{\bmE(\bmC'}{}^{\bmQ}{}_{|\bmB|}{}\Gamma^{\bmE}{}_{\bmD')\bmQ
\bmA}.
\end{eqnarray}
\end{subequations}
If the no-torsion condition \eqref{noTorsionCommutation} is satisfied,
then the first term in each of the last expressions vanishes. In this
manner one obtains an expression for the reduced geometric curvature
purely in terms of the reduced connection coefficients and, in turn, a
wave equation from either $\nabla^{\bmC}{}_{\bmD'}\Xi_{\bmA \bmB \bmC
\bmD} $ or $\nabla^{\bmC'}{}_{\bmD}\Xi_{\bmA \bmB \bmC' \bmD'}$. In
what follows, for concreteness we will consider
\[ 
\nabla^{\bmC'}{}_{\bmD}\Xi_{\bmA \bmB \bmC' \bmD'} =0. 
\]
Adapting the procedure described in Section
~\ref{Section:WaveGeneralProcedure} and taking into account equations
\eqref{RedAlgebraicCurvatureRicciSpinor} and
\eqref{ReducedGeometricCurvature} one obtains
\begin{equation} 
 \square \Gamma_{ \bmD \bmD' \bmA \bmB} -2\square_{\bmD \bmE}
 \Gamma^{\bmE}{}_{ \bmD' \bmA \bmB} -\nabla_{\bmD'\bmD}F_{\bmA
   \bmB}(x)  + 2 \nabla^{\bmC'}{}_{\bmD}
 \Gamma_{\bmE(\bmC'}{}^{\bmQ}{}_{|\bmB|}{}\Gamma^{\bmE}{}_{\bmD')\bmQ
   \bmA}  =2 \nabla^{\bmC'}{}_{\bmD} \Phi_{\bmA \bmB \bmC' \bmD'}. \label{eqProtoWaveConnection2}
\end{equation}
The gauge source function that appears in the last expression is the
\textit{frame gauge source function} defined by 
\[
F_{\bmA \bmB}(x) =
\nabla^{\bmP \bmQ'}\Gamma_{\bmP \bmQ' \bmA \bmB}.
\]
 Using the spinorial Ricci identities to replace
$\square_{\bmD \bmE}\Gamma^{\bmE}{}_{\bmD' \bmA \bmB}$ in
equation\eqref{eqProtoWaveConnection2} and exploiting the symmetry
$\Gamma^{\bmE}{}_{\bmD' \bmA \bmB}=\Gamma^{\bmE}{}_{\bmD' (\bmA
\bmB)}$ we get
\begin{flalign} 
&\square_{\bmD
\bmE}\Gamma^{\bmE}{}_{\bmD' \bmA \bmB}=-3\Lambda\Gamma_{\bmD \bmD'
\bmA \bmB} + \Gamma^{\bmE \bmH'}{}_{\bmA \bmB}\Phi_{\bmD' \bmH'\bmD
\bmE} \nonumber \\ 
& \hspace{4cm} +2\Xi \phi_{\bmD \bmE
\bmH(\bmA}\Gamma^{\bmE}{}_{|\bmD'|}{}^{\bmH}{}_{\bmB)}
-2\Gamma_{(\bmA|\bmD'\bmD|\bmB)}-2\Gamma^{\bmE}{}_{\bmD'\bmE(\bmB}\epsilon_{|\bmD|\bmA)}. \label{BoxSymConnection2}
\end{flalign}
Substituting the last expression into \eqref{eqProtoWaveConnection2} we get the wave equation
\begin{flalign*}
&\square \Gamma_{\bmD \bmD' \bmA \bmB}-2(  \Gamma^{\bmE \bmH'}{}_{\bmA
  \bmB}\Phi_{\bmD' \bmH'\bmD \bmE}  -3\Lambda\Gamma_{\bmD \bmD' \bmA
  \bmB} + 2\Xi \phi_{\bmD \bmE
  \bmH(\bmA}\Gamma^{\bmE}{}_{|\bmD'|}{}^{\bmH}{}_{\bmB)}   \nonumber
\\ 
&
\hspace{2.5cm}-2\Gamma_{(\bmA|\bmD'\bmD|\bmB)}-2\Gamma^{\bmE}{}_{\bmD'\bmE(\bmB}\epsilon_{|\bmD|\bmA)})
+
2\nabla^{\bmC'}{}_{\bmD}\Gamma_{\bmE(\bmC'}{}^{\bmQ}{}_{|\bmB|}\Gamma^{\bmE}{}_{\bmD')\bmQ\bmA}
\nonumber \\
&  \hspace{7cm} -2 \nabla^{\bmC'}{}_{\bmD}\Phi_{\bmB \bmA \bmC' \bmD'}
- \nabla_{\bmD' \bmD}F_{\bmA \bmB}(x) =0. \label{waveEqGamma2}
\end{flalign*}

\subsection{Wave equation for the Ricci spinor}

The zero-quantity defined by equation \eqref{ZeroDelta} is expressed
in terms of the spinorial counterpart of the Schouten tensor. The
spinor $L_{\bmA \bmA' \bmB \bmB'}$ can be decomposed in terms of the
Ricci spinor $\Phi_{\bmA \bmA' \bmB \bmB'}$ and $\Lambda$ as
\begin{equation}
\label{DecompositionSchoutenSpinor} 
 L_{\bmA \bmA' \bmB \bmB'}= \Phi_{\bmA \bmA' \bmB \bmB'}-\Lambda
\epsilon_{\bmA \bmB}\epsilon_{\bmA'\bmB'}
\end{equation}
 ---see Appendix \ref{Appendix:SpinorialRelations} for more details. In the context of the CEFE the
 field $\Lambda$ can be regarded as a gauge source function. Thus, in what
follows we regard the equation $\Delta_{\bmC \bmA \bmB \bmB'}=0$ as
differential conditions on $\Phi_{\bmA \bmA' \bmB \bmB'}$.

In order to derive a wave equation for the Ricci spinor we consider
\[ 
\nabla^{\bmC}{}_{\bmE'}\Delta_{\bmC\bmD \bmB \bmB'}=0.
\]
Proceeding, again, as described in
Section~\ref{Section:WaveGeneralProcedure} and using that
$\nabla^{\bmC}{}_{\bmE'}\phi_{\bmC \bmD \bmB \bmQ}=0$ ---that is,
assuming that the equation encoded in the the zero-quantity $\Lambda_{\bmC \bmD
\bmB \bmQ}$ is satisfied--- we get
\begin{equation*}
\square L_{\bmD \bmB \bmE' \bmB'} -2 \square_{\bmE' \bmQ'}L_{\bmD
\bmB}{}^{\bmQ'}{}_{\bmB'} - \nabla_{\bmD \bmE'}\nabla^{\bmE
\bmQ'}L_{\bmE \bmQ' \bmB \bmB'}-2 \phi_{\bmC \bmD \bmB
\bmQ}\nabla^{\bmC}{}_{\bmE'}\nabla^{\bmQ}{}_{\bmB'}\Xi=0 .
\end{equation*}
Using the decomposition \eqref{DecompositionSchoutenSpinor} and
symmetrising in ${}_{\bmC\bmD}$ we further obtain that
\begin{equation}
\square \Phi_{\bmD \bmB \bmE' \bmB'} -2 \square_{\bmE'
\bmQ'}\Phi_{\bmD \bmB}{}^{\bmQ'}{}_{\bmB'} - \nabla_{(\bmD
|\bmE'}\nabla^{\bmE \bmQ'}L_{\bmE \bmQ'| \bmB) \bmB'}-2 \phi_{\bmC
\bmD \bmB \bmQ}\nabla^{\bmC}{}_{\bmE'}\nabla^{\bmQ}{}_{\bmB'}\Xi=0
. \label{BoxtracelessRicciAlmost}
\end{equation}
To find a satisfactory wave equation for the Ricci tensor we need to
rewrite the last three terms of equation
\eqref{BoxtracelessRicciAlmost}.  To compute the third term observe
that the second contracted Bianchi identity as in equation 
\eqref{SecondBianchiContracted} and the decomposition of the Schouten spinor
given by equation \eqref{DecompositionSchoutenSpinor} render
\[
\nabla^{\bmE \bmQ'}L_{\bmE \bmQ' \bmB \bmB'} = \nabla^{\bmE \bmQ'}\Phi_{\bmE \bmQ' \bmB \bmB'}- \epsilon_{\bmE \bmB}\epsilon_{\bmQ' \bmB'}\nabla^{\bmE \bmQ'}\Lambda = -4\nabla_{\bmB \bmB'}\Lambda.
\]
Thus, one finds that
 \begin{equation} \nabla_{(\bmD |\bmE'}\nabla^{\bmE \bmQ'}L_{\bmE \bmQ' |\bmB) \bmB'}=-4\nabla_{\bmE' (\bmD}\nabla_{\bmB) \bmB'}\Lambda. \label{Simplification1}
 \end{equation}
This last expression is satisfactory since, as already mentioned, the
Ricci scalar $R$ (or equivalently $\Lambda$) can be regarded as a
gauge source function ---the so-called \emph{conformal gauge source
function} \cite{Fri83}. In order to replace the last term of equation
\eqref{BoxtracelessRicciAlmost} we use field equation encoded in
$Z_{\bmA \bmA' \bmB \bmB'}=0$ and the decomposition
\eqref{DecompositionSchoutenSpinor}, to obtain
\begin{equation}
\phi_{\bmC \bmD \bmB
\bmQ}\nabla^{\bmC}{}_{\bmE'}\nabla^{\bmQ}{}_{\bmB'}\Xi= -\Xi\phi_{\bmC
\bmD \bmB \bmQ}L^{\bmC}{}_{\bmE'}{}^{\bmQ}{}_{\bmB'}=-\Xi \phi_{\bmC
\bmD \bmB
\bmQ}\Phi^{\bmC}{}_{\bmE'}{}^{\bmQ}{}_{\bmB'}. \label{Simplification2}
\end{equation}
Finally, computing $\square_{\bmE' \bmQ'} \Phi_{\bmD \bmB
}{}^{\bmQ'}{}_{\bmB'}$ and substituting equations
\eqref{Simplification1} and \eqref{Simplification2} we conclude that
\begin{multline}
\square \Phi_{\bmD \bmB \bmE' \bmB'} -4
\Phi^{\bmP}{}_{(\bmB}{}^{\bmQ'}{}_{|\bmB'|}\Phi_{\bmD)}{}_{\bmP \bmE'
  \bmQ'} + 6 \Lambda \Phi_{\bmD \bmB \bmE' \bmB'} -2 \Xi
\bar{\phi}_{\bmE' \bmQ' \bmB' \bmH'}\Phi_{\bmD \bmB}{}^{\bmQ' \bmH'}
\\ 
+4\Lambda \Phi_{\bmD \bmB}{}^{\bmQ'}{}_{(\bmE'}\epsilon_{\bmQ')\bmB'}  + 2\phi_{\bmC \bmD \bmB \bmQ}\Phi^{\bmC}{}_{\bmE'}{}^{\bmQ}{}_{\bmB'}  + 4 \nabla_{\bmE' (\bmD}\nabla_{\bmB)\bmB'}\Lambda =0. 
\end{multline}

\subsection{Wave equation for the rescaled Weyl spinor}

 Proceeding as in the previous subsections,  consider the equation 
\begin{equation}
\nabla_{\bmD}{}^{\bmB'} \Lambda_{\bmB \bmB' \bmA \bmC } =0.
\label{AbstractWaveWeyl}
\end{equation}
Observe that in this case we do not need a gauge source function since
we already have a unsymmetrised derivative in the definition of
$\Lambda_{\bmB \bmB' \bmA \bmC}$. Following the procedure described in
Section ~\ref{Section:WaveGeneralProcedure} we get
\begin{equation*}
\square \phi_{\bmA \bmB \bmC \bmD }-2\square_{\bmD \bmQ}\phi_{\bmA \bmB \bmC}{}^{\bmQ}=0.
\end{equation*}
Thus, to complete the discussion we need to calculate $\square_{\bmD
\bmQ}\phi_{\bmA \bmB \bmC}{}^{\bmQ}$. Using the spinorial Ricci
identities we obtain
\[
\square_{\bmD \bmQ}\phi_{\bmA \bmB \bmC}{}^{\bmQ}= \Xi\phi_{\bmF \bmQ
\bmA \bmD}\phi_{\bmB \bmC}{}^{\bmF \bmQ} + \Xi \phi_{\bmF \bmQ \bmD
\bmB}\phi_{\bmA \bmC}{}^{\bmF \bmQ} + \Xi\phi_{\bmF \bmQ \bmC \bmD}
\phi_{\bmA \bmB}{}^{\bmF \bmQ} - 6 \Lambda \phi_{\bmA \bmB \bmC \bmD}
\]
The symmetries of $\phi_{\bmA \bmB \bmC \bmD}$  simplify the equation since
\[
\square_{(\bmD |\bmQ}\phi_{\bmA| \bmB) \bmC}{}^{\bmQ}= 3\Xi \phi^{\bmF \bmQ}{}_{(\bmA \bmB}\phi_{\bmC \bmD)\bmF \bmQ}-6\Lambda \phi_{\bmA \bmB \bmC \bmD}.
\]
Taking into account the last expression we obtain the following wave
equation for the rescaled Weyl spinor
\begin{equation}
\square \phi_{\bmA \bmB \bmC \bmD} -  6\Xi \phi^{\bmF \bmQ}{}_{(\bmA \bmB}\phi_{\bmC \bmD)\bmF \bmQ} + 12\Lambda \phi_{\bmA \bmB \bmC \bmD}=0.
\end{equation}
Observe that the wave equation for the rescaled Weyl spinor is remarkably simple.

\subsection{Wave equation for the field s }

 Since $s$ is a scalar field, the general procedure described in
Section \ref{Section:WaveGeneralProcedure} does not provide any
computational advantage. The required wave equation is derived from considering
\[ 
\nabla^{\bmA \bmA'}Z_{\bmA \bmA'}=0
\]
Explicitly, the last equation can be written as
\[ 
\square s + \nabla^{\bmA \bmA'}\Phi_{\bmA \bmC \bmA'
  \bmC'}\nabla^{\bmC \bmC}\Xi + \Phi_{\bmA \bmC \bmA'
  \bmC'}\nabla^{\bmA \bmA'}\nabla^{\bmC\bmC'}\Xi =0.
\]
Using the the contracted second Bianchi identity
\eqref{SecondBianchiContracted} to replace the second term and the
field equation $Z_{\bmA \bmB \bmC \bmD}$ along with the decomposition
\eqref{DecompositionSchoutenSpinor} to replace the third term we get
\[
\square s  -\Xi \Phi_{\bmA \bmC \bmA' \bmC'}\Phi^{\bmA \bmC \bmA' \bmC'} -3\nabla_{\bmC \bmC'}\Lambda\nabla^{\bmC \bmC'}\Xi=0.
\]

\subsection{Wave equation for the conformal factor }
A wave equation for the conformal factor follows directly from the
contraction $Z_{\bmA \bmA'}{}^{\bmA \bmA'}$ and the decomposition
\eqref{DecompositionSchouten}:
\[
\square \Xi = 4\left( s + \Lambda \Xi \right).
\]

\subsection{Summary of the analysis}
\label{Section:SummaryWaveEquations}

We summarise the results of this section in the following proposition:

\begin{proposition}
\label{CEWE}
Let
\[
 F^\bma(x), \qquad F_{\bmA \bmB}(x), \qquad \Lambda(x) 
\]
denote smooth functions on $\mathcal{M}$ such that 
\[
\nabla^{\bmQ \bmQ'}e_{\bmQ \bmQ'}{}^{\bma}=F^{\bma}(x), \qquad
\nabla^{\bmQ \bmQ'}\Gamma_{\bmQ \bmQ' \bmA \bmB}=F_{\bmA \bmB}(x)
\qquad \nabla^{\bmQ \bmQ'}\Phi_{\bmP \bmQ \bmP'
  \bmQ'}=-3\nabla_{\bmP \bmP'}\Lambda(x).
\]
If the CEFE \eqref{CEFEzeroquantities} are satisfied on $\mathcal{U}
\subset \mathcal{M}$, then one has that
\begin{subequations}\label{WaveEquations}
\begin{eqnarray}
&& {\square} e_{\bmB \bmE'}{}^{\bmc}     -2 e^{\bmQ \bmD'
  \bmc}\bar{{\Phi}}_{\bmQ \bmB \bmE' \bmD'}
+6{{\Lambda}}\hspace{1mm} e_{\bmB \bmE'}{}^{\bmc} -e_{(\bmA}{}^{\bmD'
  \bma}e_{\bmB)
  \bmD'}{}^{\bmb}{\nabla}^{\bmA}{}_{\bmE'}C_{\bma}{}^{\bmc}{}_{\bmb}
\nonumber \\   
&& \hspace{6cm} - 2  C_{\bma}{}^{\bmc}{}_{\bmb}e_{(\bmA}{}^{\bmD'
  \bma} {\nabla}^{\bmA}{}_{|\bmE'|}e_{\bmB)\bmD'}{}^{\bmb}
-{\nabla}_{\bmB \bmE'}F^{\bmc}(x) =0, 
\\ \nonumber  \\ 
\nonumber
&& \square \Gamma_{\bmD \bmD' \bmA \bmB}-2(  \Gamma^{\bmE
  \bmH'}{}_{\bmA \bmB}\Phi_{\bmD' \bmH'\bmD \bmE}
-3\Lambda\Gamma_{\bmD \bmD' \bmA \bmB} + 2\Xi \phi_{\bmD \bmE
  \bmH(\bmA}\Gamma^{\bmE}{}_{|\bmD'|}{}^{\bmH}{}_{\bmB)}
-2\Gamma_{(\bmA|\bmD'\bmD|\bmB)} \\ 
\nonumber  && \hspace{0.7cm}
-2\Gamma^{\bmE}{}_{\bmD'\bmE(\bmB}\epsilon_{|\bmD|\bmA)})  + 2
\nabla^{\bmC'}{}_{\bmD}\Gamma_{\bmE(\bmC'}{}^{\bmQ}{}_{|\bmB|}\Gamma^{\bmE}{}_{\bmD')\bmQ
  \bmA}  -2 \nabla^{\bmC'}{}_{\bmD}\Phi_{\bmB \bmA \bmC' \bmD'}  -
\nabla_{\bmD' \bmD}F_{\bmA \bmB}(x) =0, \\ 
\\ 
\nonumber && \square \Phi_{\bmD \bmB \bmE' \bmB'} -4
\Phi^{\bmP}{}_{(\bmB}{}^{\bmQ'}{}_{|\bmB'|}\Phi_{\bmD)}{}_{\bmP \bmE'
  \bmQ'} + 6 \Lambda \Phi_{\bmD \bmB \bmE' \bmB'} -2 \Xi
\bar{\phi}_{\bmE' \bmQ' \bmB' \bmH'}\Phi_{\bmD \bmB}{}^{\bmQ' \bmH'}
\\ 
&& \hspace{2.3cm} + \hspace{1mm}4\Lambda \Phi_{\bmD \bmB}{}^{\bmQ'}{}_{(\bmE'}\epsilon_{\bmQ')\bmB'}   + 4 \nabla_{\bmE' (\bmD}\nabla_{\bmB)\bmB'}\Lambda + 2\phi_{\bmC \bmD \bmB \bmQ}\Phi^{\bmC}{}_{\bmE'}{}^{\bmQ}{}_{\bmB'}=0,
\\ 
\nonumber\\ 
&& \square s  -\Xi \Phi_{\bmA \bmC \bmA' \bmC'}\Phi^{\bmA \bmC \bmA'
  \bmC'} -3\nabla_{\bmC \bmC'}\Lambda\nabla^{\bmC \bmC'}\Xi=0,  \\ 
\nonumber  \\ 
&& \square \phi_{\bmA \bmB \bmC \bmD} -  6\Xi \phi^{\bmF
  \bmQ}{}_{(\bmA \bmB}\phi_{\bmC \bmD)\bmF \bmQ} + 12\Lambda
\phi_{\bmA \bmB \bmC \bmD}=0,  \\ 
\nonumber  \\ 
&& \square \Xi - 4\left( s + \Lambda \Xi \right)=0,  
\end{eqnarray}
\end{subequations}
on $\mathcal{U}$.
\end{proposition}

\medskip
\noindent 
\textbf{Remark.} The unphysical metric is not part of the
unknowns of the system of equations of the spinorial version of the
CEFE. This observation is of relevance in the present context because when the operator
$\square$ is applied to a spinor $N_{ \mathcal{K}}$ of non-zero range
one obtains first derivatives of the connection ---if the metric is part
of the unknowns then these first derivatives of the connection
representing second derivatives of $g$ would enter into the principal
part of the operator $\square$.  Therefore, since in this setting the
metric is not part of the unknowns, the principal part of the operator
$\square$ is given by $\epsilon^{\bmA
\bmB}\epsilon^{\bmA'\bmB'}\bme_{\bmA \bmA'}\bme_{\bmB \bmB'}$.

\medskip
\noindent
\textbf{Remark.} In the sequel let $\{ \bme, \,\Gamma, \, \Phi, \,
\bmphi\}$ denote vector-valued unknowns encoding the independent
components of $\{ e_{\bmA \bmA'}{}^{\bmc}, \, \Gamma_{\bmC \bmC' \bmA
\bmB}, \, \Phi_{\bmA \bmA' \bmB \bmB'}, \, \phi_{\bmA \bmB \bmC \bmD}
\}$ and let $\mathbf{u} \equiv (\bme, \,\Gamma, \,\Phi, \,\bmphi,\,
s,\, \Xi)$. Additionally, let $\partial \mathbf{u}$ denote
collectively the derivatives of $\mathbf{u}$.  With this notation the
wave equations of Proposition \ref{CEWE} can be recast as a quasilinear wave
equation for $\mathbf{u}$ having, in local coordinates, the form
\begin{equation} 
g^{\mu\nu}(\mathbf{u})\partial_\mu \partial_\nu\mathbf{u} + \bmF(x;
\mathbf{u}, \partial \mathbf{u} ) =0,
\label{genFormwithU}
\end{equation}
where  $\bmF$ is a vector-valued function of its arguments and
$g^{\mu\nu}$ denotes the components, in local coordinates, of
contravariant version of a Lorentzian metric  $\bmg$.  In accordance
with our notation $g^{\mu\nu}\equiv \eta^{\bma
  \bmb}\bme_{\bma}{}^\mu\bme_{\bmb}{}^\nu$ where, in local
coordinates, one writes $\bme_\bma=\bme_\bma{}^\mu \bmpartial_\mu$.

\section{Propagation of the constraints and the derivation of the
  subsidiary system} 
\label{PropagationConstraintsSubsidiarySystem}

The starting point of the derivation of the wave equations discussed
in the previous section was the CEFE. Therefore, any solution to the
CEFE is a solution to the wave equations. It is now natural to ask:
under which conditions a solution to the wave equations
\eqref{WaveEquations} will imply a solution to the CEFE?  The general
strategy to answer this question is to use the spinorial wave
equations of Proposition \ref{CEWE} to construct a subsidiary system of
homogeneous wave equations for the zero-quantities and impose
vanishing initial conditions. Then, using a standard existence and
uniqueness result, the unique solution satisfying the data will be
given by the vanishing of each zero-quantity. This means that under
certain conditions (encoded in the initial data for the subsidiary
system) a solution to the spinorial wave equations will imply a
solution to the original CEFE. The procedure to construct the
subsidiary equations for the zero quantities is similar to the
construction of the wave equations of Proposition \ref{CEWE}. There is,
however, a key difference: the covariant derivative is, a priori, not
assumed to be a Levi-Civita connection. Instead we assume that the
connection is metric but not necessarily torsion-free. We will denote
this derivative by $\widetriangle{\nabla}$. Therefore, whenever a
commutator of covariant derivatives appears, or in spinorial terms the
operator $\widetriangle{\square}_{\bmA \bmB}\equiv
\widetriangle{\nabla}_{\bmC'(\bmA}\widetriangle{\nabla}_{\bmB)}{}^{\bmC'}$,
it is necessary to use the $\widetriangle{\nabla}$-spinorial Ricci
identities involving a non-vanishing torsion spinor ---this
generalisation is given in the Appendix A and is
required in the discussion of the subsidiary equations where the
torsion is, in itself, a variable for which a subsidiary equation
needs to be constructed.

\medskip
As in the previous section, the procedure for obtaining the subsidiary
system is similar for each zero-quantity. Therefore, we first give a
general outline of the procedure.

\subsection{General procedure for obtaining the subsidiary system and
  the propagation of the constraints}
\label{Section:GenericSubsidiarySystem}

In the general procedure described in Section
\ref{Section:WaveGeneralProcedure}, the spinor $N_{\bmE \bmA
\mathcal{K}} $ played the role of a zero-quantity, while the spinor
$M_{\bmA \bmB'\mathcal{K}}$ played the role of the variable for which
the wave equation \eqref{ModelWave} was to be derived. In the
construction of the subsidiary system we are not interested in finding
an equation for $M_{\bmA \bmB'\mathcal{K}}$ but in deriving an
equation for $N_{\bmE \bmA \mathcal{K}} $ under the hypothesis that
the wave equation for $M_{\bmA \bmB'\mathcal{K}}$ is satisfied.  As
already discussed, since we cannot assume that the connection is
torsion-free the equation for $N_{\bmE \bmA \mathcal{K}}$ has to be
written in terms of the metric connection $\widetriangle{\nabla}$.

\medskip
 Before deriving the subsidiary equation let us emphasise an important
point. In Section \ref{Section:WaveGeneralProcedure} we defined
$N_{\bmE\bmA\mathcal{K}} \equiv \nabla_{\bmE}{}^{\bmB'}M_{\bmA \bmB'
\mathcal{K}}$. Then, decomposing this quantity as usual we
obtained 
\[
N_{\bmE \bmA\mathcal{K}}={\nabla}_{\bmE}{}^{\bmB'}M_{\bmA
\bmB' \mathcal{K}}+ \tfrac{1}{2}\epsilon_{\bmE
\bmA}{\nabla}^{\bmQ\bmB'}M_{\bmQ \bmB' \mathcal{K}}.
\] 
 At this point in the discussion of Section \ref{Section:WaveGeneralProcedure} we introduced a
gauge source function ${\nabla}^{\bmP\bmQ'}M_{\bmP\bmQ'\mathcal{K}}=F_{\mathcal{K}}$. 
Now, instead of
directly deriving an equation for $N_{\bmE \bmA \mathcal{K}}$ we have
derived an equation using the modified quantity 
\[
\widehat{N}_{\bmE
\bmA\mathcal{K}}\equiv{\nabla}_{\bmE}{}^{\bmB'}M_{\bmA \bmB'
\mathcal{K}}+ \tfrac{1}{2}\epsilon_{\bmE
\bmA}F_{\mathcal{K}}.
\]
Accordingly, the wave equations of
Proposition \ref{CEWE} can be succinctly written as $
{\nabla}^{\bmE}{}_{\bmA'}\widehat{N}_{\bmE \bmA\mathcal{K}} =0$. Later
on, we will have to show that, in fact, $\widehat{N}_{\bmE
\bmA\mathcal{K}}={N}_{\bmE \bmA \mathcal{K}} $ if the appropriate
initial conditions are satisfied. In addition, observe that $
{\nabla}^{\bmA}{}_{\bmC'}\widehat{N}_{\bmE \bmA \mathcal{K}}$ can be
written in terms of the connection $\widetriangle{\nabla}$ by means of
a \textit{transition spinor} $Q_{\bmA \bmA' \bmB \bmC}$ ---see
Appendix \ref{Appendix:Torsion} for the definition. Using equation
\eqref{defApplyTransition} of Appendix \ref{Appendix:Torsion} we get
\begin{equation} 
\label{waveTransition}
\widetriangle{\nabla}^{\bmA}{}_{\bmC'} \widehat{N}_{\bmA \bmB
  \mathcal{K}}= \nabla^{\bmA}{}_{\bmC'}\widehat{N}_{\bmA \bmB
  {\mathcal{K}}}  -Q^{\bmA}{}_{\bmC' \bmA}{}^{\bmH}\widehat{N}_{\bmH
  \bmB \mathcal{K} }-Q^{\bmA}{}_{\bmC' \bmB
  \mathcal{K}}{}^{\bmH}\widehat{N}_{\bmA \bmH \mathcal{K}} -\cdots-
Q^{\bmA}{}_{\bmC'\bmK}{}^{\bmH}\widehat{N}_{\bmA \bmB \cdots \bmH} 
 \end{equation}
where $_\bmK$ is the last index of the string $_{\mathcal{K}}$. For a
connection which is metric, the transition spinor can be written
entirely in terms of the torsion as
\begin{equation} 
Q_{\bmA \bmA'\bmB \bmC} \equiv -2 \Sigma_{\bmB \bmA \bmA' \bmC} - 2
\Sigma_{\bmA(\bmC | \bmA' | \bmB)}- 2\bar{\Sigma}_{\bmA' ( \bmC |
  \bmQ'}{}^{\bmQ'}\epsilon_{\bmA|\bmB)}.
\label{TransitionSpinor}
\end{equation}
If the wave equation $ {\nabla}^{\bmA}{}_{\bmC'}\widehat{N}_{\bmA \bmB
\mathcal{K}} =0$ is satisfied, the first term of equation
\eqref{waveTransition} vanishes. Therefore, the wave equations of
Proposition \ref{CEWE} can be written in terms of the connection $\widetriangle{\nabla}$ as
\begin{equation}
\widetriangle{\nabla}^{\bmA}{}_{\bmC'} \widehat{N}_{\bmA \bmB \mathcal{K}}=  -Q^{\bmA}{}_{\bmC' \bmA}{}^{\bmH}\widehat{N}_{\bmH \bmB \mathcal{K}}-Q^{\bmA}{}_{\bmC' \bmB}{}^{\bmH}\widehat{N}_{\bmA \bmH \mathcal{K}} - ..... - Q^{\bmA}{}_{\bmC' \bmK}{}^{\bmH}\widehat{N}_{\bmA \bmB... \bmH}. 
\label{WaveTransition}
\end{equation}
\noindent In what follows, the right hand side of the last equation
will be denoted by $W_{\bmB \bmC'\mathcal{K}}. $

\subsubsection{The subsidiary system}
\label{Sec:TheSubsidiarySystem}

 Now, we want to show that by setting the appropriated initial
conditions, if the wave equation
${\nabla}^{\bmA}{}_{\bmE'}\widehat{N}_{\bmA \bmB \mathcal{K}} =0$
holds then $\widehat{N}_{\bmA \bmB \mathcal{K}} =0$. The strategy will
be to obtain an homogeneous wave equation for $\widehat{N}_{\bmA \bmB
\mathcal{K}}$ written in terms of the connection
$\widetriangle{\nabla}$. First, observe that
$\widetriangle{\nabla}^{\bmQ'}{}_{\bmP} \widehat{N}_{\bmA \bmB
\mathcal{K}}$ can be decomposed as
\begin{equation}
\widetriangle{\nabla}^{\bmQ'}{}_{\bmP}\widehat{N}_{\bmA
\bmB\mathcal{K}}=\widetriangle{\nabla}{}^{\bmQ'}{}_{(\bmP}\widehat{N}_{\bmA)\bmB
\mathcal{K}} + \tfrac{1}{2} \epsilon_{\bmP
\bmA}\widetriangle{\nabla}^{\bmQ'}{}_{\bmE}\widehat{N}^{\bmE}{}_{\bmB
\mathcal{K}}. \label{crucial}
\end{equation}
Replacing the second term using \eqref{WaveTransition} ---i.e. using
that the wave equation ${\nabla}^{\bmA}{}_{\bmE'}\widehat{N}_{\bmA
  \bmB\mathcal{K}}{} =0 $ holds---  we get that 
\[
\widetriangle{\nabla}^{\bmQ'}{}_{\bmP}\widehat{N}_{\bmA
\bmB\mathcal{K}}=\widetriangle{\nabla}{}^{\bmQ'}{}_{(\bmP}\widehat{N}_{\bmA)\bmB
\mathcal{K}} + \tfrac{1}{2} \epsilon_{\bmP \bmA} W^{\bmQ'}{}_{\bmB
\mathcal{K}}
\]
Applying $\widetriangle{\nabla}^{\bmP}{}_{\bmQ'}$ to the previous
equation and expanding the symmetrised term in the right-hand side one
obtains
\begin{eqnarray*}
&&\widetriangle{\nabla}^{\bmP}{}_{\bmQ'}\widetriangle{\nabla}^{\bmQ'}{}_{\bmP}
\widehat{N}_{\bmA \bmB}= \tfrac{1}{2}
\widetriangle{\nabla}^{\bmP}{}_{\bmQ'}\left(
  \widetriangle{\nabla}^{\bmQ'}{}_{\bmP}\widehat{N}_{\bmA\bmB} +
  \widetriangle{\nabla}^{\bmQ'}{}{}_{\bmA}\widehat{N}_{\bmP \bmB}
\right) + \tfrac{1}{2}\widetriangle{\nabla}_{\bmA
  \bmQ'}W^{\bmQ'}{}_{\bmB \mathcal{K}},  \\
&& \hspace{2.7cm} =
-\tfrac{1}{2}\widetriangle{\square}\widehat{N}_{\bmA \bmB} -
\tfrac{1}{2}\widetriangle{\nabla}_{\bmP\bmQ'}
\widetriangle{\nabla}^{\bmQ'}{}_{\bmA}\widehat{N}^{\bmP}{}_{\bmB} +
\tfrac{1}{2}\widetriangle{\nabla}_{\bmA \bmQ'}W^{\bmQ'}{}_{\bmB
\mathcal{K}},
\\ 
&& \hspace{2.7cm} =
-\tfrac{1}{2}\widetriangle{\square}\widehat{N}_{\bmA \bmB} -
\tfrac{1}{2}\left(\widetriangle{\square}_{\bmP
\bmA}\widehat{N}^{\bmP}{}_{\bmB} + \tfrac{1}{2}\epsilon_{\bmP
\bmA}\widetriangle{\square}\widehat{N}^{\bmP}{}_{\bmB} \right ) +
\tfrac{1}{2}\widetriangle{\nabla}_{\bmA \bmQ'}W^{\bmQ'}{}_{\bmB
\mathcal{K}}.
\end{eqnarray*}
From this expression, after some rearrangements we obtain
\[
\widetriangle{\square}\widehat{N}_{\bmA \bmB \mathcal{K}}=2\widetriangle{\square}_{\bmP \bmA} \widehat{N}^{\bmP}{}_{\bmB\mathcal{K}} - 2 \widetriangle{\nabla}_{\bmA \bmQ'}W^{\bmQ'}{}_{\bmB \mathcal{K}}.
\]
It only remains to reexpress the right-hand side of the above equation
using the $\widetriangle{\nabla}$- spinorial Ricci identities. This
can be computed for each zero-quantity using the expressions given in
Appendix \ref{Appendix:SpinorialRelations}. Observe that the result is always an homogeneous
expression in the zero-quantities and its first derivatives. The last
term also shares this property since the transition spinor can be
completely written in terms of the torsion, as shown in equation
\eqref{TransitionSpinor}, which is one of the
zero-quantities. Finally, once the homogeneous wave equation is
obtained we set the initial conditions
\[
\widehat{N}_{\bmA \bmB \mathcal{K}}|_{\mathcal{S}}=0 \hspace{7mm}
\text{ and}  \hspace{7mm} (\widetriangle{\nabla}_{\bmE
  \bmE'}\widehat{N}_{\bmA \bmB \mathcal{K}})|_{\mathcal{S}}=0
\]
on a initial hypersurface $\mathcal{S}$, and using a standard result
of existence and uniqueness for wave equations we conclude that the
unique solution satisfying this data is $\widehat{N}_{\bmA \bmB
\mathcal{K}}=0$.

\medskip
\noindent
\textbf{Remark.} The crucial step in the last derivation was the
assumption that the equation
${\nabla}^{\bmA}{}_{\bmE'}\widehat{N}_{\bmA \bmB \mathcal{K}} =0 $ is
satisfied ---i.e. the wave equation \eqref{ModelWave} for $M_{\bmA \bmB
\mathcal{K}}$.

\subsubsection{Initial data for the subsidiary system}
\label{Section:InitialDataSubsidiarySystem}

Now, we take a closer look at the initial conditions
 \begin{equation*}
 \widehat{N}_{\bmA \bmB \mathcal{K}}|_{\mathcal{S}}=0, \qquad
(\widetriangle{\nabla}_{\bmE \bmE'}\widehat{N}_{\bmA \bmB
\mathcal{K}})|_{\mathcal{S}}=0.
\end{equation*}
As will be shown in the sequel, these conditions will be used to
construct initial data for the wave equations of Proposition \ref{CEWE}. The
important observation is that only $\widehat{N}_{\bmA \bmB
\mathcal{K}}|_{\mathcal{S}}=0$ is essential, while
$\widetriangle{\nabla}_{\bmE \bmE'}\widehat{N}_{\bmA
\bmB}|_{\mathcal{S}}=0$ holds by virtue of the condition
${\nabla}^{\bmA \bmA'}\widehat{N}_{\bmA \bmB \mathcal{K}}=0$. In order
to show this, first observe that as the spatial derivatives of
$\widehat{N}_{\bmA \bmB \mathcal{K}}$ can be determined from
$\widehat{N}_{\bmA\bmB \mathcal{K}}|_{\mathcal{S}}=0$, it follows that
$(\widetriangle{\nabla}_{\bmE \bmE'}\widehat{N}_{\bmA \bmB
\mathcal{K}})|_{\mathcal{S}}=0$ is equivalent to only specify the
derivative along the normal to $\mathcal{S}$.

\medskip
Let $\tau^{\bmA \bmA'}$ be an Hermitian spinor corresponding to a
timelike vector such that $\tau^{\bmA \bmA'}|_{\mathcal{S}}$ is the
normal to $\mathcal{S}$. The spinor $\tau^{\bmA \bmA'}$ can be used to
perform a \emph{space spinor split} of the derivative
$\widetriangle{\nabla}_{\bmA \bmA'}$:
\[ 
\widetriangle{\nabla}_{\bmA \bmA'}= \tfrac{1}{2}\tau_{\bmA
  \bmA'}\mathcal{P}-\tau_{\bmA'}{}^{\bmQ}\mathcal{D}_{\bmA \bmQ}
\]
where 
\[
\mathcal{P}\equiv\tau^{\bmA \bmA'}\widetriangle{\nabla}_{\bmA \bmA'}
\qquad \text{and} \qquad \mathcal{D}_{\bmA \bmB}\equiv
\tau_{(\bmB}{}^{\bmA'}\widetriangle{\nabla}_{\bmA)\bmA'} 
\]
 denote, respectively, the derivative along the direction given by
$\tau^{\bmA \bmA'}$ and $\mathcal{D}_{\bmA \bmB}$ is the \textit{Sen
connection} relative to $\tau^{\bmA \bmA'}$
\hspace{0.1mm}\footnote{Under certain conditions the Sen connection
coincides with the Levi-Civita connection of the intrinsic 3-metric of
the hypersurfaces orthogonal to $\tau^{\bmA \bmA'}$.}.  We have chosen
the normalisation $\tau^{\bmA \bmA'}\tau_{\bmA \bmA'}= 2$, in
accordance with the conventions of \cite{Fri91}. Using this split and
$\widehat{N}_{\bmA \bmB \mathcal{K}}|_{\mathcal{S}}=0$ it follows that
\[
\widetriangle{\nabla}_{\bmE \bmE'}\widehat{N}_{\bmA \bmB
  \mathcal{K}}|_{\mathcal{S}}=\tfrac{1}{2}(\tau_{\bmE
  \bmE'}\mathcal{P}\widehat{N}_{\bmA \bmB \mathcal{K}})|_{\mathcal{S}}. 
\] 
Therefore, requiring $ \widetriangle{\nabla}_{\bmE
\bmE'}\widehat{N}_{\bmA \bmB}|_{\mathcal{S}} =0$ is equivalent to
having $( \mathcal{P}\widehat{N}_{\bmA \bmB
\mathcal{K}})|_{\mathcal{S}} =0 $ as previously stated. Now, observe
that the wave equation $ {\nabla}^{\bmA \bmA'}\widehat{N}_{\bmA \bmB
\mathcal{K}}=0 $ or, equivalently, $\widetriangle{\nabla}^{\bmA
\bmA'}\widehat{N}_{\bmA \bmB\mathcal{K}}=W{}^{\bmA'}{}_{\bmB
\mathcal{K}}$ implies $( \widetriangle{\nabla}^{\bmA
\bmA'}\widehat{N}_{\bmA \bmB
\mathcal{K}})|_{\mathcal{S}}=W{}^{\bmA'}{}_{\bmB
\mathcal{K}}|_{\mathcal{S}}$ \footnote{ Recall that
$W{}^{\bmA'}{}_{\bmB \mathcal{K}}|_{\mathcal{S}}$ is given entirely in
terms of zero-quantities since the transition spinor can be written in
terms of the torsion.}. Therefore, if we require that all the
zero-quantities vanish on the initial hypersurface $\mathcal{S}$ then
$( \widetriangle{\nabla}^{\bmA \bmA'}\widehat{N}_{\bmA \bmB
\mathcal{K}})|_{\mathcal{S}}=0$.  Using, again, the space spinor
decomposition of $\widetriangle{\nabla}_{\bmA \bmA'}$ and considering
$\widehat{N}_{\bmA \bmB\mathcal{K}}|_{\mathcal{S}}=0$ we get $
(\tau^{\bmA \bmA'}\mathcal{P}\widehat{N}_{\bmA
\bmB\mathcal{K}})|_{\mathcal{S}} = 0$ which also implies that
$(\mathcal{P}\widehat{N}_{\bmA \bmB \mathcal{K}})|_{\mathcal{S}}=0$.

\medskip
Summarising, the only the condition that is needed is that all the
zero-quantities vanish on the initial hypersurface $\mathcal{S}$ since
the condition $(\widetriangle{\nabla}_{\bmE \bmE'}\widehat{N}_{\bmA
\bmB \mathcal{K}})|_{\mathcal{S}}=0$ is always satisfied by virtue of
the wave equation ${\nabla}^{\bmA \bmA'}\widehat{N}_{\bmA \bmB
\mathcal{K}}=0$.

\subsubsection{Propagation of the constraints}

We still need to show that $\widehat{N}_{\bmA \bmB}{} =N_{\bmA \bmB}$. One can write
\[
N_{\bmA \bmB \mathcal{K}}- \widehat{N}_{\bmA \bmB
  \mathcal{K}}=\tfrac{1}{2}\epsilon_{\bmA \bmB}Q_{\mathcal{K}}, 
\] 
where $Q_{\mathcal{K}}$ encodes the difference between
$\widehat{N}_{\bmA \bmB \mathcal{K}}$ and $N_{\bmA \bmB
\mathcal{K}}$. Computing the trace of the last equation and taking
into account the definition of $N_{\bmA \bmB \mathcal{K}}$ one finds that $
\widehat{N}^{\bmA}{}_{\bmA \mathcal{K}}=Q_{\mathcal{K}} $. Now,
invoking the results derived in the last subsection it follows that if
the wave equation ${\nabla}^{\bmA}{}_{\bmE'}\widehat{N}_{\bmA \bmB
\mathcal{K}} =0$ is satisfied and all the zero-quantities vanish on
the initial hypersurface $\mathcal{S}$ then $\widehat{N}_{\bmA \bmB
\mathcal{K}}=0$. This observation also implies that if
$\widehat{N}^{\bmA}{}_{\bmA \mathcal{K}}|_{\mathcal{S}}=0$ then
$\widehat{N}^{\bmA}{}_{\bmA \mathcal{K}}=0 $. The later result,
expressed  in terms of
$Q_{\mathcal{K}}$ means that  if $ Q_{\mathcal{K}}|_{\mathcal{S}}=0$
then $ Q_{\mathcal{K}}=0 $. Therefore, requiring that all the
zero-quantities vanish on $\mathcal{S}$ and that the wave equation
$\nabla^{\bmA \bmA'}\widehat{N}_{\bmA \bmB \mathcal{K}}=0$ holds
everywhere, is enough to ensure that 
\[
\widehat{N}_{\bmA \bmB\mathcal{K}} =N_{\bmA \bmB\mathcal{K}}
\]
everywhere. Moreover, $\widehat{N}_{\bmA \bmB\mathcal{K}}=0$ implies
that $N_{\bmA \bmB\mathcal{K}}=0$ and the gauge conditions
hold. Namely, one has that
\[
\nabla^{\bmA \bmB'}M_{\bmA \bmB'
\mathcal{K}}=F_{\mathcal{K}}(x).
\]

\subsection{Subsidiary system and propagation of the constraints}
\label{IntroPropSubsidiary}

The essential ideas of the Section
\ref{Section:GenericSubsidiarySystem} can be applied to every single
zero-quantity. One only needs to take into account the particular
index structure of each zero-quantity  encoded in the string of spinor
indices ${}_\mathcal{K}$. The problem then reduces to  the computation of
\[ 
\widetriangle{\square}_{\bmP \bmA}
\widehat{N}^{\bmP}{}_{\bmB\mathcal{K}}, \qquad 
\widetriangle{\nabla}_{\bmA \bmQ'}W^{\bmQ'}{}_{\bmB \mathcal{K}},
\]
the result of which is to be substituted into
\begin{equation}
\widetriangle{\square}\widehat{N}_{\bmA
\bmB\mathcal{K}}=2\widetriangle{\square}_{\bmP \bmA}
\widehat{N}^{\bmP}{}_{\bmB \mathcal{K}} -2\widetriangle{\nabla}_{\bmA
\bmQ'}W^{\bmQ'}{}_{\bmB \mathcal{K}} .
\label{SubsidiarySystemForm}
\end{equation}
The latter can be succinctly computed using the equations
\eqref{TorsionSpinorialRicciIdentities} in Appendix
\ref{Appendix:SpinorialRelations}. The explicit form can be easily
obtained and renders long expressions for each zero-quantity.  The key
observation from these computations is that
\eqref{SubsidiarySystemForm} leads to an homogeneous wave
equation. The explicit form is given in Appendix
\ref{Appendix:SubsidiarySystem}. These results can be summarised in
the following proposition:

\begin{proposition}
\label{Proposition:SubsidiaryEquations}
Assume that the wave equations 
\begin{eqnarray*} 
 {\nabla}^{\bmA}{}_{\bmE'}\widehat{\Sigma}_{\bmA \bmB}{}^{\bmc}=0,  &
 \qquad &  {\nabla}^{\bmC'}{}_{\bmD}\widehat{\Xi}_{\bmA \bmB \bmC' \bmD'}=0, 
  \\
 {\nabla}^{\bmC}{}_{\bmE'}\widehat{\Delta}_{\bmC \bmD \bmB \bmB'}=0, 
&\qquad & {\nabla}_{\bmE}{}^{\bmB'}{\Lambda}_{\bmB \bmB' \bmA \bmC}=0,
\\
 \nabla_{\bmA \bmA'}Z^{\bmA \bmA'}=0,  &\qquad &  Z_{\bmA \bmA'}{}^{\bmA \bmA'}=0,
\end{eqnarray*}
are satisfied everywhere. Then the zero-quantities satisfy the
homogeneous wave equations
\begin{eqnarray*}
&& \widetriangle{\square}\widehat{\Sigma}_{\bmA \bmB}{}^{\bmc} -2\widetriangle{\square}_{\bmP \bmA}\widehat{\Sigma}^{\bmP}{}_{\bmB}{}^{\bmc} + 2 \widetriangle{\nabla}_{\bmA \bmQ'}W[{\Sigma}]{}^{\bmQ'}{}_{\bmB}{}{}^{\bmc} =0, \nonumber  \\
 && \widetriangle{\square}\widehat{\Xi}_{\bmA \bmB \bmC' \bmD'} -2\widetriangle{\square}_{\bmP' \bmC'}\widehat{\Xi}_{\bmA \bmB}{}^{\bmP'}{}_{\bmD'} + 2\widetriangle{\nabla}_{\bmC' \bmQ}W[{\Xi}]{}^{\bmQ}{}_{\bmA \bmB \bmD'} =0, \\
 && \widetriangle{\square}\widehat{\Delta}^{\bmP}{}_{\bmD \bmB \bmB'} -2\widetriangle{\square}_{\bmP \bmC}\widehat{\Delta}^{\bmP}{}_{\bmD \bmB \bmB'} + 2 \widetriangle{\nabla}_{\bmC \bmQ'}W[{\Delta}]{}^{\bmQ'}{}_{\bmD \bmB \bmB'} =0,  \nonumber \\
 && \widetriangle{\square}{\Lambda}_{\bmB \bmB' \bmA \bmC} -2\widetriangle{\square}_{\bmP' \bmB'}{\Lambda}_{\bmB}{}^{\bmP'}{}_{\bmA \bmC} + 2\widetriangle{\nabla}_{\bmB' \bmQ }W[\Lambda]{}^{\bmQ}{}_{\bmB \bmA \bmC}=0,  \nonumber  \\
&& \widetriangle{\nabla}_{\bmA \bmA'}Z^{\bmA \bmA'} - W[Z]{}^{\bmA \bmA'}{}_{\bmA \bmA'}=0,
\end{eqnarray*}
where
\begin{eqnarray*}
& W[{\Sigma}]{}^{\bmQ'}{}_{\bmB}{}{}^{\bmc} \equiv
\widetriangle{\nabla}^{\bmQ'}{}_{\bmE}\widehat{\Sigma}^{\bmE}{}_{\bmB}{}^{\bmc},
\quad  W[{\Xi}]{}^{\bmQ}{}_{\bmA \bmB \bmD'} \equiv
\widetriangle{\nabla}^{\bmQ}{}_{\bmE'}\widehat{\Xi}_{\bmA \bmB}{}^{
  \bmE'}{}_{\bmD'},  \quad   W[{\Delta}]{}^{\bmQ'}{}_{\bmD \bmB
  \bmB'} \equiv
\widetriangle{\nabla}^{\bmQ'}{}_{\bmF}\widehat{\Delta}^{\bmF}{}_{\bmD
  \bmB \bmB'}, & \\ 
& W[\Lambda]{}^{\bmQ}{}_{\bmB \bmA \bmC} \equiv
\widetriangle{\nabla}_{\bmE'}{}^{\bmQ}{\Lambda}_{\bmB}{}^{\bmE'}{}_{\bmA
  \bmC}, \quad  W[Z]{}^{\bmA \bmA'}{}_{\bmA \bmA'} \equiv
\widetriangle{\nabla}^{\bmA \bmA'}Z_{\bmA \bmA'}.& 
\end{eqnarray*} 
\end{proposition}

We will refer to the set of equations given in the last proposition as
the \emph{subsidiary system}. It should be noticed that the terms of
the form $\widetriangle{\square}_{\bmP \bmA}
\widehat{N}^{\bmP}{}_{\bmB \mathcal{K}} $ and $W^{\bmQ'}{}_{\bmB
\mathcal{K}}$ can be computed using the $\widetriangle{\nabla}$-Ricci
identities and the transition spinor $Q_{\bmA \bmA' \bmB \bmC}$
respectively. Using the subsidiary equations from the previous
proposition one readily obtains the following \emph{Reduction Lemma}:

\begin{proposition}
\label{ReductionLemma}
If the initial data for the subsidiary system of Proposition
\ref{Proposition:SubsidiaryEquations} is given by
\begin{eqnarray*}
 &\widehat{\Sigma}_{\bmA \bmB}{}^{\bmc}|_{\mathcal{S}} =0, \quad
 \widehat{\Xi}_{\bmA \bmB \bmC' \bmD'}|_{\mathcal{S}} =0, \quad
 \widehat{\Delta}_{ \bmA \bmB \bmC \bmC'}|_{\mathcal{S}} =0, \quad
 {\Lambda}_{\bmB \bmB' \bmA \bmC}|_{\mathcal{S}} =0, \quad  Z_{\bmA
   \bmA'}|_{\mathcal{S}} =0, &
\end{eqnarray*}
where $\mathcal{S}$ in an spacelike hypersurface, and the wave
equations of Proposition \ref{CEWE} are satisfied everywhere, then one has a
solution to the vacuum CEFE ---in other words
\begin{eqnarray*}
 &{\Sigma}_{\bmA \bmB}{}^{\bmc} =0, \hspace{0.5cm}  {\Xi}_{\bmA \bmB
   \bmC' \bmD'} =0, \hspace{0.5cm} {\Delta}_{ \bmA \bmB \bmC \bmC'}
 =0, \hspace{0.5cm}{\Lambda}_{\bmB \bmB' \bmA \bmC} =0, \hspace{0.5cm}
 Z_{\bmA \bmA'} =0,& 
\end{eqnarray*}
in $D(\mathcal{S})$. Moreover, whenever $\Xi \neq 0$, the solution to the CEFE implies a
solution to the vacuum Einstein field equations.
\end{proposition}

\begin{proof}
It can be verified, using the $\widetriangle{\nabla}$-Ricci identities
given in the Appendix \ref{Appendix:SpinorialRelations}, that the equations of Proposition
\ref{Proposition:SubsidiaryEquations} are homogeneous wave equations
for the zero-quantities. Notice, however, that the equation for
$Z_{\bmA \bmA'}$ is not a wave equation but of first order and
homogeneous. Therefore, if we impose that the zero-quantities vanish
on an initial spacelike hypersurface $\mathcal{S}$ then by the
homogeneity of the equations we have that
\begin{eqnarray*}
 &\widehat{\Sigma}_{\bmA \bmB}{}^{\bmc}=0, \quad
 \widehat{\Xi}_{\bmA \bmB \bmC' \bmD'}=0, \quad
 \widehat{\Delta}_{ \bmA \bmB \bmC \bmC'}=0,
 \quad {\Lambda}_{\bmB \bmB' \bmA \bmC} =0, \quad
 Z_{\bmA \bmA'} =0,& 
\end{eqnarray*}
everywhere on $D(\mathcal{S})$. Moreover, since initially $\widehat{\Sigma}_{\bmA
  \bmB}{}^{\bmc}={\Sigma}_{\bmA \bmB}{}^{\bmc}$, \hspace{0.1cm}
$\widehat{\Xi}_{\bmA \bmB \bmC' \bmD'}={\Xi}_{\bmA \bmB \bmC' \bmD'}$
and $\widehat{\Delta}_{ \bmA \bmB \bmC \bmC'}={\Delta}_{ \bmA \bmB
  \bmC \bmC'}$, we have that ${\Sigma}_{\bmA \bmB}{}^{\bmc}=0$, $
{\Xi}_{\bmA \bmB \bmC' \bmD'}=0$,  $ {\Delta}_{ \bmA \bmB \bmC
  \bmC'}=0$  on $D(\mathcal{S})$. In addition, using that  a solution to the CEFE implies a
solution to the Einstein field equations whenever $\Xi \neq 0$
\cite{Fri83}, it follows that a solution to the wave equations of
Proposition \ref{CEWE} with initial data consistent with the initial
conditions given in Proposition \ref{ReductionLemma} will imply a
solution to the vacuum Einstein field equations whenever $\Xi \neq 0$.
\end{proof}

\medskip
\noindent 
\textbf{Remark.} It is noticed that the initial
data for the subsidiary equations gives a way to specify the data for
the wave equations of Proposition \ref{CEWE}. This observation is
readily implemented through a space spinor formalism which mimics the
hyperbolic reduction process to extract a first order hyperbolic
system out of the CEFE ---see e.g. \cite{Fri91}. In oder to illustrate
this procedure let us consider the data for the rescaled Weyl spinor
encoded in $\Lambda_{\bmA' \bmB \bmC \bmD}|_{\mathcal{S}} =0$.

\subsubsection{Initial data for the wave equations  (rescaled Weyl
  spinor)}
\label{sec:InitialDataWave-Weyl}

We need to provide the initial data
\[ 
\phi_{\bmA \bmB \bmC \bmD}|_{\mathcal{S}}, \qquad
\mathcal{P}\phi_{\bmA \bmB \bmC \bmD}|_{\mathcal{S}}.
\]
A convenient way to specify the initial data is to use the space
spinor formalism to split the equations encoded in $\Lambda_{\bmA'
\bmB \bmC \bmD}=0$. From this split, a system of evolution and
contraint equations can be obtained. Recall that $\Lambda_{\bmA'\bmB
\bmC \bmD}\equiv \nabla^{\bmQ}{}_{\bmA'}\phi_{\bmA \bmB \bmC \bmQ}
$. Making use of the the decomposition of $\nabla_{\bmA
\bmB}\equiv\tau_{\bmB}{}^{\bmA'}\nabla_{\bmA \bmA'}$ in terms of the
operators $\mathcal{P}$ and $\mathcal{D}_{\bmA \bmB}$ we get
\[ 
\Lambda_{\bmA \bmB \bmC \bmD} = -\tfrac{1}{2}\mathcal{P}\phi_{\bmA
  \bmB\bmC\bmD} + \mathcal{D}^{\bmQ}{}_{\bmA}\phi_{\bmB \bmC \bmD
  \bmQ} 
\]
Evolution and constraint equations are obtained, respectively, from considering
\begin{eqnarray*}
&& E_{\bmA \bmB \bmC \bmD} \equiv -2\Lambda_{\bmA \bmB \bmC \bmD}
=\mathcal{P}\phi_{\bmA \bmB \bmC \bmD} - 2
\mathcal{D}^{\bmQ}{}_{(\bmA}\phi_{\bmB \bmC \bmD) \bmQ}=0,
\hspace{0.5cm}\text{(evolution equation)}  \\
&&
C_{\bmC \bmD}\equiv \Lambda^{\bmQ}{}_{\bmQ\bmC \bmD}=\mathcal{D}^{\bmP \bmQ}\phi_{\bmP\bmQ \bmC \bmD}=0 \hspace{3.8cm}\text{(constraint equation)}.
\end{eqnarray*}
Restricting the last equations to the initial hypersurface
$\mathcal{S}$ it follows that the initial data $\phi_{\bmP \bmQ \bmC
\bmD}|_{\mathcal{S}}$ must satisfy $C_{\bmC \bmD}|_{\mathcal{S}}=0$
and the initial data for $(\mathcal{P}\phi_{\bmP \bmQ \bmC
\bmD})|_{\mathcal{S}}$ can be read form $E_{\bmA \bmB \bmC \bmD}|_{\mathcal{S}}=0$.

\medskip
The procedure for the other equations is analogous and can be
succinctly obtained by revisiting the derivation of the first order
hyperbolic equations derived from the CEFE using the space spinor
formalism ---see for instance \cite{Fri91}.

\section{Analysis of the Milne Universe} \label{Milne}

\begin{figure}[t]
\centerline{\includegraphics[width=0.25\textwidth]{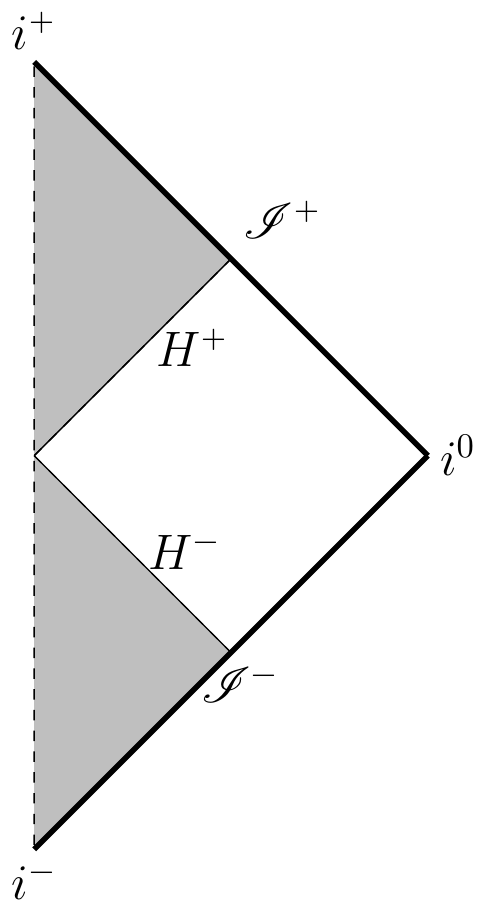}}
\caption{Penrose Diagram for the Milne Universe. The Milne Universe
  diagram correspond to the non-spatial (shaded area) portion of the
  Penrose diagram of the Minkowski spacetime. The boundary $H^{+}\cup
  H^{-}$ corresponds to the limit of the region where the coordinates
  $(t,\chi)$ are well defined.}
\label{fig:Milne}
 \end{figure}

As an application of the hyperbolic reduction procedure described in
the previous sections we analyse the stability of the \emph{Milne
Universe},
$(\tilde{\mathring{\mathcal{M}}},\tilde{\mathring{\bmg}})$. This
spacetime is a Friedman-Lema\^{i}tre-Robinson-Walker vacuum solution
with vanishing Cosmological constant, energy density and pressure. In
fact, it represents flat space written in comoving coordinates of the
world-lines starting at $t=0$ ---see \cite{GriPod09}. This means that
the Milne Universe can be seen as a portion of the Minkowski
spacetime, which we know is conformally related to the \emph{Einstein
Cosmos}, $(\mathcal{M}_{E} \equiv \mathbb{R}\times
\mathbb{S}^3,\mathring{\bmg})$ (sometimes also called the
\emph{Einstein cylinder}) ---see Figure \ref{fig:Milne}. The metric
$\tilde{\mathring{\bmg}}$ of the Milne Universe is given in comoving
coordinates $(t,\chi,\theta,\varphi)$ by
\begin{equation}
\tilde{\mathring{\bmg}} = \mathbf{d}t \otimes \mathbf{d}t - t^2\left(
\mathbf{d}\chi \otimes \mathbf{d}\chi + \sin^{2}\chi \left(
\mathbf{d}\theta \otimes \mathbf{d}\theta + \sin^2 \theta \mathbf{d}\varphi \otimes
\mathbf{d}\varphi\right)\right)
\label{MilneMetric}
\end{equation} 
 where
 \[
 t\in (-\infty,\infty), \hspace{0.5cm} \chi \in
[0,\infty), \hspace{0.5cm} \theta \in[0,\pi], \hspace{0.5cm}
\phi\in[0,2\pi).
\] 
In fact, introducing the coordinates
\[
\bar{r}\equiv t \sinh\chi, \hspace{1cm} \bar{t}\equiv t \cosh\chi 
\]
the metric reads
\begin{equation*}
\tilde{\mathring{\bmg}}= \mathbf{d}\bar{t} \otimes \mathbf{d}\bar{t} -
\mathbf{d}\bar{r}\otimes \mathbf{d}\bar{r} -\bar{r}^2 \left(
\mathbf{d}\theta \otimes \mathbf{d}\theta
+\sin^{2}\theta\mathbf{d}\varphi \otimes \mathbf{d}\varphi\right).
\end{equation*}
Therefore $\bar{t}^2 - \bar{r}^2 >0$, and the Milne Universe
corresponds to the non-spatial region of Minkowski spacetime as shown
in the Penrose diagram of Figure \ref{fig:Milne}.  As already
discussed, this metric is conformally related to the metric
$\mathring{\bmg}$ of the Einstein Cosmos. More precisely, one has that
\[
\mathring{\bmg}= \mathring{\Xi}^2 \tilde{\mathring{\bmg}}
\]
where the metric of the Einstein cylinder,  $\mathring{\bmg}$,  is given by
\[
\mathring{\bmg}\equiv \mathbf{d}T \otimes \mathbf{d}T- \bmhbar,  
\]
 with $\bmhbar$ denoting the standard metric of $\mathbb{S}^3$
\[
\bmhbar \equiv \mathbf{d}\psi \otimes \mathbf{d}\psi + \sin^2 \psi
\mathbf{d}\theta \otimes \mathbf{d}\theta + \sin^2 \psi \sin^2 \theta
\mathbf{d}\varphi \otimes \mathbf{d}\varphi.
\]
The conformal factor relating the metric of the Milne Universe to
metric of the Einstein Universe is given by
\[
\mathring{\Xi}=\cos{T} + \cos {\psi},  
\]
and the coordinates $(T,\psi)$ are related to $(\bar{t},\bar{r})$ via
\[
  T=\arctan(\bar{t}+\bar{r})+\arctan(\bar{t}-\bar{r}), \hspace{1cm}
\psi = \arctan(\bar{t}+\bar{r})-\arctan(\bar{t}-\bar{r}).
\]
Equivalently,  in  terms of the original coordinates $t$ and $\chi$ we have
\[
  \chi = \arctan\left( \frac{\sin \psi}{\sin T}\right),  \hspace{1cm}
 t = \sqrt{\frac{\cos \psi-\cos T}{\cos \psi+ \cos T}}.
\]
Therefore, the Milne Universe is conformal to the domain
\[
\tilde{\mathring{\mathcal{M}}} =\{ p \in \mathcal{M}_{E} | \hspace{0.1cm}0 \leq \psi  < \pi,
 \hspace{0.2cm} \psi - \pi < T  < \pi -\psi,  \hspace{0.2cm} |T| >
 \psi \}.
 \]

\subsection{The Milne Universe as a solution to the wave equations}

Since the Milne Universe is a solution to the the Einstein field
equations, it follows that the pair $(\mathring{\bmg},\mathring{\Xi})$
implies a solution to the CEFE which, in turn, constitutes a solution
to the wave equations of Proposition \ref{CEWE}. Following the
discussion of Section \ref{Section:CEFE}, this solution consists of
the frame fields
\[
\{e_{\bma}{}^{\bmc},\Gamma_{\bma}{}^{\bmb}{}_{\bmc},L_{\bma
  \bmb},d^{\bma}{}_{\bmb \bmc \bmd},\Sigma_{\bma},\Xi,s \}
\]
 or, equivalently, the spinorial fields
\[
\{e_{\bmA \bmA'}{}^{\bmc},\Gamma_{\bmA
  \bmA'}{}^{\bmB}{}_{\bmC},\Phi_{\bmA \bmA' \bmB \bmB'},\phi_{\bmA
  \bmB \bmC \bmD},\Sigma_{\bmA \bmA'},\Sigma,\Xi,s \}
\]
where we have written $\Sigma_{\bma}\equiv \nabla_{\bma}\Xi$ and
$\nabla_{\bmA \bmA'}\Xi\equiv\Sigma_{\bmA \bmA'}$ as a shorthand for
the derivative of the conformal factor. 

\medskip
For later use, we notice that in the Einstein Cosmos
$(\mathcal{M}_E,\mathring{\bmg})$ we have
\begin{eqnarray*}
\textbf{Weyl}[\mathring{\bmg}]=0, \qquad
\textbf{R}[\mathring{\bmg}]=-6, \qquad
\textbf{Schouten}[\mathring{\bmg}]=\tfrac{1}{2}\left( \mathbf{d}
T \otimes \mathbf{d} T + \bmhbar \right).
\end{eqnarray*}
The spinorial version of the above tensors can be more easily
expressed in terms of a frame.  To this end, now consider a 
geodesic on the Einstein Cosmos $(\mathcal{M}_E, \mathring{\bmg})$
given by
\[
x(\tau)=(\tau, x_\star), \hspace{0.5cm} \tau \in
\mathbb{R}, 
\]
where $x_\star \in \mathbb{S}^3$ is fixed. Using the
congruence of geodesics generated  varying $x_\star$ over
$\mathbb{S}^3$ we obtain a Gaussian system of coordinates
$(\tau,x^{\alpha})$ on the Einstein cylinder $\mathbb{R}\times
\mathbb{S}^3$ where $(x^{\alpha})$ are some local coordinates on
$\mathbb{S}^3$. In addition, in a slight abuse of notation \emph{we
identify the standard time coordinate $T$ on the Einstein cylinder with
the parameter $\tau$ of the geodesic}.

\subsubsection{Frame expressions}
\label{Section:BackgroundFrameExpressions}

A globally defined orthonormal frame on the Einstein Cosmos
$(\mathcal{M}_{E},\mathring{\bmg})$ can be constructed by
first considering the linearly independent vector fields in
$\mathbb{R}^4$
\begin{eqnarray*}
 && \bmc_{1}\equiv w\frac{\partial}{\partial
   z}-z\frac{\partial}{\partial w} + x\frac{\partial}{\partial y} -
 y\frac{\partial}{\partial x}, \\
&&
\bmc_{2}\equiv w\frac{\partial}{\partial y} -y
\frac{\partial}{\partial w} + z\frac{\partial}{\partial x} - x
\frac{\partial}{\partial z}, \\
&&
\bmc_{3}\equiv w\frac{\partial}{\partial x} - x\frac{\partial}{\partial w}+ y \frac{\partial}{\partial z}-z\frac{\partial}{\partial y},
\end{eqnarray*}
where $(w,x,y,z)$ are Cartesian coordinates in $\mathbb{R}^4$. The
vectors $\{\bmc_\bmi\}$ are tangent to $\mathbb{S}^3$ and form a
global frame for $\mathbb{S}^3$ ---see e.g. \cite{RyaShe75}.  This
spatial frame can be extended to a spacetime frame $\{
\mathring{\bme}_\bma \}$ by setting $\mathring{\bme}_{\bmzero}\equiv
\bmpartial_\tau$ and $\mathring{\bme}_\bmi \equiv\bmc_\bmi$. Using
this notation we observe that the components of the basis respect to
this frame are given by
$\mathring{\bme}_{\bma}=\delta_{\bma}{}^{\bmb}\bmc_{\bmb}\equiv
\ring{e}_{\bma}{}^{\bmb}\bmc_{\bmb} $. With respect to this orthogonal
basis the components of the Schouten tensor are given by
\[
\mathring{L}_{\bma
  \bmb}=\delta_{\bma}{}^{0}\delta_{\bma}{}^{0}-\tfrac{1}{2}\eta_{\bma
  \bmb}. 
\]
so that the components of the traceless Ricci tensor
are given by 
\[
\mathring{R}_{\{\bma \bmb\}}= 2
\delta_{\bma}{}^{0}\delta_{\bmb}{}^{0}-\frac{1}{2}\eta_{\bma \bmb} 
\]
where the curly bracket around the indices denote the symmetric
trace-free part of the tensor. In addition, 
\[
\mathring{d}_{\bma \bmb \bmc \bmd}=0
\]
 since the Weyl tensor vanishes.

\medskip
Now, let $\mathring{\gamma}_{\bmi}{}^{\bmj}{}_{\bmk}$ denote the
connection coefficients of the Levi-Civita connection $\bmD$ of
$\bmhbar$ with respect to the spatial frame $\{\bmc_{\bmi}\}$. Observe
that the structure coefficients defined by
$[\bmc_{\bmi},\bmc_{\bmj}]=C_{\bmi}{}^{\bmk}{}_{\bmj}\bmc_{\bmk}$ are
given by
\[
 \ring{\gamma}_{\bmi}{}^{\bmk}{}_{\bmj}=
 -\epsilon_{\bmi}{}^{\bmk}{}_{\bmj}
\]
where $\epsilon_{\bmi}{}^{\bmk}{}_{\bmj}$ is the 3-dimensional
Levi-Civita totally antisymmetric tensor. Taking into account that
$\mathring{\bme}_{\bm0}{}=\bmpartial_{\tau}$ is a timelike Killing
vector of $\mathring{\bmg}$, we can readily obtain the connection
coefficients $\mathring{\Gamma}_{\bma}{}^{\bmb}{}_{\bmc}$, of the
Levi-Civita connection $\mathring{\bmnabla}$ of the metric
$\mathring{\bmg}$, with respect to the basis
$\{\mathring{\bme}_{\bma}\}$. More precisely, one has that
\[ 
\mathring{\Gamma}_{\bma}{}^{\bmb}{}_{\bmc}=-\epsilon_{\bm0
  \bma}{}^{\bmb}{}_{\bmc}.
\]

\medskip
 For the conformal factor and its concomitants we readily obtain
\[
\mathring{\Sigma}\equiv \mathring{\Sigma}_\bmzero = -\sin \tau, \qquad \mathring{\Sigma}_\bmthree=-\sin{\psi},
\qquad\mathring{\Sigma}_\bmone=\mathring{\Sigma}_\bmtwo=0, \qquad   \mathring{s}=-\tfrac{1}{2}(\cos \tau
+ \cos \psi). 
\]

\subsubsection{Spinorial expressions}

In order to obtain the spinor frame form of the last expressions let
$\tau^{\bmA \bmA'}$ denote the spinorial counterpart of the vector
$\sqrt{2}\bmpartial_{\tau}$  so that  $\tau_{\bmA \bmA'}\tau^{\bmA
  \bmA'}=2$. With this choice,  consider a spinor dyad
$\{\epsilon_{\bmA}{}^{A} \}= \{o^{A},\iota^{A} \}$ adapted to
$\tau^{\bmA \bmA'}$  ---i.e. a spinor dyad such that  
\begin{equation}
\tau_{\bmA \bmA'}= \epsilon_{\bm0}{}^{A}\epsilon_{\bm0'}{}^{A'} + \epsilon_{\bm1}{}^{A}\epsilon_{\bm1'}{}^{A'}. \label{Tau-AdpatedDyad}
\end{equation}
The spinor $\tau^{\bmA\bmA'}$ can be used to introduce a space spinor
formalism similar to the one discussed in Section
\ref{Section:InitialDataSubsidiarySystem}.  Now, the \emph{Infeld-van
der Waerden symbols} $\sigma_{\bmA \bmA'}{}^{\bmb}$ are given by the
matrices
\begin{eqnarray*}
&& \sigma_{\bmA \bmA'}{}^{\bm0} \equiv \tfrac{1}{\sqrt{2}} \begin{pmatrix}   1 & 0  \\   0 & 1  \\  \end{pmatrix},  \hspace{1.2cm}  \sigma_{\bmA \bmA'}{}^{\bm1} \equiv \tfrac{1}{\sqrt{2}} \begin{pmatrix}   0 & 1  \\   1 & 0  \\  \end{pmatrix}, \\ \\ &&   \sigma_{\bmA \bmA'}{}^{\bm2} \equiv \tfrac{1}{\sqrt{2}} \begin{pmatrix}   0 & -\mbox{i}  \\   \mbox{i} & 0  \\  \end{pmatrix},  \hspace{1cm} \sigma_{\bmA \bmA'}{}^{\bm3} \equiv \tfrac{1}{\sqrt{2}} \begin{pmatrix}   1 & 0  \\   0 & -1  \\  \end{pmatrix}  .     
\end{eqnarray*}
One directly finds that in the present case
\begin{equation}
\mathring{e}_{\bmA \bmA'}{}^{\bmb}\equiv \sigma_{\bmA
  \bmA'}{}^{\bma}\mathring{e}_{\bma}{}^{\bmb}  =\sigma_{\bmA \bmA'}{}^{\bmb} .
\label{BackgroundFrameExplicit}
\end{equation}
Now, decomposing $\mathring{e}_{\bmA \bmA'}{}^{\bmb}$ using the space spinor
formalism induced by $\tau^{\bmA \bmA'}$ one has that
\begin{equation*} \label{splitbasis}
 \mathring{e}_{\bmA \bmA'}{}^{\bma}= \frac{1}{2}\tau_{\bmA \bmA'}\mathring{e}^{\bma}-\tau^{\bmQ}{}_{\bmA'}\mathring{e}_{\bmA \bmQ}{}^{\bma}
\end{equation*}
where
\[
\mathring{e}^\bma \equiv \mathring{e}_{\bmA\bmA'}{}^\bma \tau^{\bmA\bmA'}, \qquad
\mathring{e}_{\bmA\bmB}{}^\bma \equiv \tau_{(\bmA}{}^{\bmQ'} \mathring{e}_{\bmB)\bmQ'}{}^\bma.
\]
Comparing the above expression with equation
\eqref{BackgroundFrameExplicit} one readily finds that the
coefficients  $\mathring{e}^{\bma}$ and $\mathring{e}_{\bmA \bmB}{}^{\bma}$ are given by
\begin{eqnarray*}  
&& \mathring{e}^{\bm0}=\sqrt{2}, \hspace{2cm} \mathring{e}_{\bmA \bmB}{}^{\bm0}=0, \label{splitbasisinfeld1} \\ 
&& \mathring{e}^{\bmi}=0, \hspace{2.3cm} \mathring{e}_{\bmA \bmB}{}^{\bmi}= \sigma_{\bmA \bmB}{}^{\bmi}{}_{}. \label{splitbasisinfeld2}
\end{eqnarray*}
where $\sigma_{\bmA \bmB}{}^{\bmi}$ are the spatial Infeld-van der
Waerden symbols, given by the matrices
\begin{eqnarray*}
\sigma_{\bmA \bmB}{}^{\bm1} \equiv \tfrac{1}{\sqrt{2}} \begin{pmatrix}   -1 & 0  \\   0 & 1  \\  \end{pmatrix},  \hspace{2cm}  \sigma_{\bmA \bmB}{}^{\bm2} \equiv \tfrac{1}{\sqrt{2}} \begin{pmatrix}   \mbox{i} & 0  \\   0 & \mbox{i}  \\  \end{pmatrix} , \hspace{2cm}  \sigma_{\bmA \bmB}{}^{\bm3} \equiv \tfrac{1}{\sqrt{2}} \begin{pmatrix}   0 & 1  \\   1 & 0  \\  \end{pmatrix}  .    
\end{eqnarray*}
The above expressions provide a direct way of recasting the frame
expressions of Section \ref{Section:BackgroundFrameExpressions} in spinorial terms. 

\medskip
Now, denoting by $2\mathring{\Phi}_{\bmA \bmA' \bmB \bmB'}$ the spinorial
counterpart of $\mathring{R}_{\{\bma \bmb \}}$ we get
\[ 
\mathring{\Phi}_{\bmA \bmA' \bmB \bmB'}= \tfrac{1}{2}\sigma_{\bmA
  \bmA'}{}^{\bma}\sigma_{\bmB \bmB'}{}^{\bmb}\mathring{R}_{\{\bma \bmb\}}=
\sigma_{\bmA \bmA'}{}^{\bm0}\sigma_{\bmB
  \bmB'}{}^{\bm0}-\tfrac{1}{4}\epsilon_{\bmA \bmB}\epsilon_{\bmA'
  \bmB'}.
\]
From \eqref{Tau-AdpatedDyad} we see that $\tau_{\bmA
  \bmA'}=\sqrt{2}\sigma_{\bmA \bmA'}{}^{\bm0}$. Accordingly, 
\[ 
\mathring{\Phi}_{\bmA \bmA' \bmB \bmB'}= \tfrac{1}{2}\tau_{\bmA \bmA'}\tau_{\bmB
  \bmB'}-\tfrac{1}{4}\epsilon_{\bmA \bmB}\epsilon_{\bmA' \bmB'}.
\]
To obtain the reduced spin connection coefficients we proceed as
follows: let $\mathring{\Gamma}_{\bmA \bmA'}{}^{\bmB \bmB'}{}_{\bmC \bmC'}$
denote the spinorial counterpart of
$\mathring{\Gamma}_{\bma}{}^{\bmb}{}_{\bmc}$. Now, since
$\mathring{\Gamma}_{\bma}{}^{\bmb}{}_{\bmc}=-\epsilon_{\bm0
  \bma}{}^{\bmb}{}_{\bmc}$, we can readily compute its spinorial
counterpart by recalling the spinorial version of the volume form 
\[ 
\epsilon_{\bmA \bmA' \bmB \bmB' \bmC \bmC' \bmD \bmD'}=
\mbox{i}(\epsilon_{\bmA \bmC}\epsilon_{\bmB \bmD}\epsilon_{\bmA'
  \bmD'}\epsilon_{\bmB' \bmC'} - \epsilon_{\bmA \bmD}\epsilon_{\bmB
  \bmC}\epsilon_{\bmA' \bmC'}\epsilon_{\bmB' \bmD'}). 
\]
It follows then that
\[
\mathring{\Gamma}_{\bmB \bmB'}{}^{\bmC \bmC'}{}_{\bmD
  \bmD'}=-\tfrac{1}{\sqrt{2}}\tau^{\bmA \bmA'}\epsilon_{\bmA \bmA'
  \bmB \bmB'}{}^{\bmC \bmC'}{}_{\bmD \bmD'} =  -
\tfrac{\mbox{i}}{\sqrt{2}}(\tau^{\bmC}{}_{\bmD'}\epsilon_{\bmB
  \bmD}\epsilon_{\bmB'}{}^{\bmC'}-\tau_{\bmD}{}^{\bmC'}\epsilon_{\bmB}{}^{\bmC}\epsilon_{\bmB'
  \bmD'}  ) 
 \]
Combining the last expression with the definition of the reduced spin
connection coefficients $\mathring{\Gamma}_{\bmA \bmA'}{}^{\bmC}{}_{\bmB}$ given
in  equation \eqref{ReducedConnection} one obtains
\[
\mathring{\Gamma}_{\bmB \bmB'}{}^{\bmC}{}_{\bmD}=  -\tfrac{\mbox{i}}{2\sqrt{2}}(\tau^{\bmC}{}_{\bmQ'}\epsilon_{\bmB \bmD} \delta_{\bmB'}{}^{\bmQ'} -\tau_{\bmD}{}^{\bmQ'}\epsilon_{\bmB'\bmQ'}\delta_{\bmB}{}^{\bmC}) = -\tfrac{\mbox{i}}{2\sqrt{2}} (\epsilon_{\bmB \bmD}\tau^{\bmC}{}_{\bmB'} + \tau_{\bmD \bmB'}\delta_{\bmB}{}^{\bmC}   ) 
\]
Thus, one concludes that
\[ 
\mathring{\Gamma}_{\bmA \bmA' \bmB \bmC}=
-\tfrac{\mbox{i}}{\sqrt{2}}\epsilon_{\bmA(}\epsilon_{\bmC}\tau_{\bmB)\bmA'}  
\]
Finally, for the rescaled Weyl tensor we simply have   
\[
\mathring{\phi}_{\bmA \bmB \bmC \bmD}=0.
\]

\subsubsection{Gauge source functions for the Milne spacetime}

The expressions for $\mathring{\Gamma}_{\bma}{}^{\bmb}{}_{\bmc}$ and
$\mathring{e}_{\bmb}{}^{\bma}$ derived in the previous sections allow to readily compute
the gauge source functions associated to the conformal representation
of the Milne Universe under consideration.  Regarding
$\mathring{e}_{\bmb}{}^{\bma}$ as the component of a contravariant tensor we
compute
\[ 
\mathring{\nabla}^{\bmb}\mathring{e}_{\bmb}{}^{\bma}=\eta^{\bmc
  \bmb}\mathring{\nabla}_{\bmc}\mathring{e}_{\bmb}{}^{\bma}=\eta^{\bmc
  \bmb}(\mathring{\bme}_{\bmc}(\mathring{e}_{\bmb}{}^{\bma})-\mathring{\Gamma}_{\bmc}{}^{\bme}{}_{\bmb}\mathring{e}_{\bme}{}^{\bma})
\]
where $\mathring{\bme}_{\bmc}=\mathring{e}_\bmc{}^\bme\bmc_\bme$. Using
that for this case $\mathring{e}_{\bmb}{}^{\bma}=\delta_{\bmb}{}^{\bma}$, we get
\[ 
\mathring{\nabla}^{\bmb}\mathring{e}_{\bmb}{}^{\bma}=-\eta^{\bmc
  \bmb}\mathring{\Gamma}_{\bmc}{}^{\bma}{}_{\bmb}= \eta^{\bmc \bmb}\epsilon_{\bm0
  \bmc}{}^{\bma}{}_{\bmb}=0. 
\]
Therefore, the coordinate gauge source function vanishes. That is, one has that
\[ 
\mathring{F}^{\bma}(x)= \mathring{\nabla}^{\bmA \bmA'}\mathring{e}_{\bmA \bmA'}{}^{\bma}=0.
\]

\medskip
We proceed in similar way to compute the frame gauge source
function. One has that 
\begin{eqnarray*}
&& \mathring{\nabla}^{\bma}\mathring{\Gamma}_{\bma}{}^{\bmb}{}_{\bmc}=\eta^{\bmd \bma}\mathring{\nabla}_{\bmd}\mathring{\Gamma}_{\bma}{}^{\bmb}{}_{\bmc}=\eta^{\bmd \bma}\big(\mathring{\bme}_\bmd(\mathring{\Gamma}_{\bma}{}^{\bmb}{}_{\bmc})+ 
\mathring{\Gamma}_{\bmd}{}^{\bmb}{}_{\bme}\mathring{\Gamma}_{\bma}{}^{\bme}{}_{\bmc}-\mathring{\Gamma}_{\bmd}{}^{\bme}{}_{\bma}\mathring{\Gamma}_{\bme}{}^{\bmb}{}_{\bmc}
- \mathring{\Gamma}_{\bmd}{}^{\bme}{}_{\bmc}\mathring{\Gamma}_{\bme}{}^{\bmb}{}_{\bma}    \big)
\\ 
&& \hspace{3.7cm}=-\eta^{\bmd \bma} \mathring{\bme}_\bmd(\epsilon_{\bm0 \bma
}{}^{\bmd}{}_{\bmc}) + \eta^{\bmd \bma}\epsilon_{\bm0
  \bmd}{}^{\bmb}{}_{\bme}\epsilon_{\bm0 \bma}{}^{\bme}{}_{\bmc} -
\eta^{\bmd \bma}\epsilon_{\bm0 \bmd}{}^{\bme}{}_{\bma}\epsilon_{\bm0
  \bme}{}^{\bmb}{}_{\bmc}-\eta^{\bmd \bma}\epsilon_{\bm0
  \bmd}{}^{\bme}{}_{\bmc}\epsilon_{\bm0 \bme}{}^{\bmb}{}_{\bma} \\ 
&& \hspace{3.7cm} =\epsilon_{\bm0}{}^{\bma \bmb}{}_{\bme}\epsilon_{\bm0 \bma}{}^{\bme}{}_{\bmc}-\epsilon_{\bm0 \bma}{}^{\bmb}{}_{\bme}\epsilon_{\bm0}{}^{\bma \bme}{}_{\bmc} =0.   
\end{eqnarray*}
Therefore, using the irreducible decomposition of
$\mathring{\Gamma}_{\bmA \bmA'}{}^{\bmB\bmB'}{}_{\bmC \bmC'}$ in terms
of $\mathring{\Gamma}_{\bmA \bmA'}{}^{ \bmB}{}_{\bmC}$ given in
\eqref{ReducedConnection}, we conclude that
\[
 \mathring{F}_{\bmA \bmB}(x)=\nabla^{\bmQ \bmQ'}\mathring{\Gamma}_{\bmQ \bmQ'\bmA \bmB}=0.
\]

\medskip
Finally, the conformal gauge source function is determined by the
value of the Ricci scalar, in this case $R=-6$. It follows then that
\[
\Lambda=\tfrac{1}{4}. 
\]

\subsubsection{Summary} 
\label{summaryMilneSpinors}

We collect the main results in the following proposition:

\begin{proposition}
\label{Proposition:InitialData}
The fields $(\mathring{\Xi},\ring{\Sigma},\ring{\Sigma}_{\bmi},\ring{s},\ring{e}_{\bma}{}^{\bmb},\ring{\Gamma}_{\bma}{}^{\bmb}{}_{\bmc},\ring{L}_{\bma \bmb},\ring{d}^{\bma}{}_{\bmb \bmc \bmd} )$
given by
\begin{eqnarray*}
& \ring{\Xi}= \cos \tau + \cos \psi, \qquad  \ring{\Sigma}= -\sin
\tau, \qquad \ring{\Sigma}_{\bm3}=-\sin\psi, \qquad \ring{e}_{\bma}{}^{\bmb}=\delta_{\bma}{}^{\bmb},& \\
 &\ring{\Gamma}_{\bma}{}^{\bmb}{}_{\bmc}= -\epsilon_{\bm0 \bma
 }{}^{\bmb}{}_{\bmc}, \qquad  \ring{L}_{\bma
   \bmb}=2\delta_{\bma}{}^{\bm0}\delta_{\bmb}{}^{\bm0}-\frac{1}{2}\eta_{\bma
   \bmb}, \qquad  \ring{d}^{\bma}{}_{\bmb\bmc \bmd}=0,
 \qquad \ring{s}=-\frac{1}{2}(\cos \tau + \cos \psi), &
\end{eqnarray*}
or, alternatively, in spinorial terms, the fields
$(\ring{\Xi},\ring{\Sigma},\ring{\Sigma}_{\bmA
\bmA'},\ring{s},\ring{e}_{\bmA \bmA'}{}^{\bmb}, \ring{\Gamma}_{\bmA
\bmA'}{}^{\bmB}{}_{\bmC},\ring{\Phi}_{\bmA \bmA'\bmB
\bmB'},\ring{\phi}_{\bmA \bmB \bmC \bmD} )$ with 
$\ring{\Xi}$, $\ring{\Sigma}$, $\ring{s}$ as above and
\begin{eqnarray*}
& \ring{e}_{\bmA \bmA'}{}^{\bmb}=\sigma_{\bmA \bmA'}{}^{\bmb},
\qquad \ring{\Gamma}_{\bmA \bmA' \bmB
  \bmC}=-\tfrac{\mbox{\emph{i}}}{\sqrt{2}}\epsilon_{\bmA(\bmB}\tau_{\bmC) \bmA'},
\qquad \ring{\Phi}_{\bmA \bmA' \bmB \bmB'}=\frac{1}{2}\tau_{\bmA
  \bmA'}\tau_{\bmB \bmB'}-\tfrac{1}{4}\epsilon_{\bmA
  \bmB}\epsilon_{\bmA \bmC}, &  \\
 & \ring{\phi}_{\bmA \bmB \bmC \bmD}=0  \qquad \Sigma_{\bmA
   \bmA'}=-\sin \psi\sigma_{\bmA \bmA'}{}^{\bm3},& 
\end{eqnarray*}
defined on the Einstein cylinder $\mathbb{R}\times\mathbb{S}^3$
constitute  a solution to the conformal Einstein field equations
representing the Milne Universe. The  gauge source functions
associated to this representation are given by
\[
 \ring{F}^{\bma}(x)=0, \qquad \ring{F}_{\bmA \bmB}(x)=0,
 \qquad  \ring{\Lambda}=\tfrac{1}{4}. 
 \]

\end{proposition}

\subsection{Perturbation of initial data}

In order to discuss the stability of the Milne Universe using the
conformal wave equations we need to find a way to parametrise
perturbations of initial data close to the data for the exact solution.

\medskip
A \emph{basic initial data set for the conformal field equations}
consists of a collection $(\mathcal{S}, \bmh,\bmK, \Omega, \Sigma)$
such that $\mathcal{S}$ denotes a 3-dimensional manifold, $\bmh$ is a
3-metric, $\bmK$ is a symmetric tensor, $\Omega$ and $\Sigma$ are
scalar functions on the manifold $\mathcal{S}$ satisfying
\begin{eqnarray*}
&& 2\Omega D_i D^{i}\Omega - 3D_{i}\Omega D^{i}\Omega + \tfrac{1}{2}\Omega^2 r -3\Sigma^2 -\tfrac{1}{2}\Omega^2( K^2 - K_{ij}K^{ij} ) + 2\Omega\Sigma K = 0 \\
&& \Omega^3 D^{i}(\Omega^{-2}K_{ik})-\Omega (D_k K -2\Omega^{-1}D_k \Sigma) =0
\end{eqnarray*}
where $D$ is the covariant Levi-Civita connection of $\bmh$, $r$ is
its Ricci scalar and $K \equiv h^{ij}K_{ij}$. These equations are,
respectively, the \emph{conformal counterpart of the Hamiltonian and
momentum constraints in vacuum}.  In what follows, it will be assumed
that $\mathcal{S}$ is diffeomorphic to $\mathbb{S}^3$ and write
$\mathcal{S}\approx \mathbb{S}^3$. As the initial hypersurface
$\mathcal{S}$ is compact, one can assume, without lost of generality, that
$\Omega=1$ and $\Sigma=0$.

\subsection{Wave maps}
\label{Section: Wave maps}

Since we are assuming $\mathcal{S} \approx \mathbb{S}^3$ there exists
a diffeomorphism $\psi: \mathcal{S} \rightarrow \mathbb{S}^3$. The
freedom encoded in the choice of diffeomorphism can be used to both
fix the gauge and to discuss the parametrisation of the perturbation
data in a appropriate setting.

\medskip
The diffeomorphism $\psi$ can be used to pull-back a coordinate system
$x=(x^{\alpha})$ in $\mathbb{S}^3$ to a coordinate system
$\wideparen{x}$ in $\mathcal{S}$ since $ \wideparen{x}= \psi^{*} x= x
\circ \psi$. Exploiting the fact that $\psi$ is a diffeomorphism we
can define not only the pull-back $\psi^{*}: T^{*}\mathbb{S}^3
\rightarrow T^{*}\mathcal{S} $ but also the push-forward of its inverse $ (
\psi^{-1})_{*} : T\mathbb{S}^3 \rightarrow T\mathcal{S}$. Using this, we
can push-forward vector fields $\bmc_{\bmi}$ on $T\mathbb{S}^3$ and
pull-back their covector fields $\bmalpha^{\bmi}$ on $T^{*}\mathcal{S}$ as
\[ 
\wideparen{\bmc}_{\bmi}=(\psi^{-1})_{*}\bmc_{\bmi}, \qquad
\wideparen{\bmalpha}^{\bmi}=\psi^{*}\bmalpha^{\bmi}.    
\]

Now, let $\wideparen{\bme}_{\bmi} \equiv
\wideparen{e}_{\bmi}{}^{\bmj}\wideparen{\bmc}_{\bmj}$ be an
$\bmh$-orthonormal frame and let us denote with $\bmD$ and
$\slashed{\bmD}$ the Levi-Civita connection of $\bmh$ and $\bmhbar$
respectively. Observe that the covectors $\bmalpha^{\bmi}$ and
$\wideparen{\bmalpha}^{\bmi}$ can be expanded in terms of their
respective basis of 1-forms as
\[ 
\bmalpha^{\bmi}=\alpha^{\bmi}{}_\alpha\mathbf{d} x^\alpha, \hspace{2cm}
\wideparen{\bmalpha}^{\bmi}=\wideparen{\alpha}^{\bmi}{}_\alpha \mathbf{d}
\wideparen{x}^\alpha. 
\]
 Therefore, defining the spatial coordinate gauge source function as 
\[
 f^{\bmi}(x) \equiv {D}^{\bmj}\wideparen{e}_{\bmj}{}^{\bmi}, 
\]
it can be shown that the right-hand side of the last equation can be computed to yield
\begin{eqnarray}\label{spatialgaugecoord} 
&&f^{\bmi}(x)={D}^\beta{\alpha}^{\bmi}{}_\beta=  \alpha^{\bmi}{}_\alpha
{D}^\beta {D}_\beta x^\alpha + h^{\beta \gamma}\frac{\partial
  x^\alpha}{\partial \wideparen{x}^\beta}\frac{\partial x^\delta }{\partial
  \wideparen{x}^\gamma }\slashed{D}_\delta \alpha^{\bmi}{}_{\alpha}. 
\end{eqnarray}
Finally, writing $f^\alpha (x)=f^{\bmi}(x)\bmc_{\bmi}{}^\alpha$ where
$c_{\bmi}{}^\alpha = \langle \mathbf{d} x^\alpha, \bmc_{\bmi}\rangle$
we find that
\[ 
{D}^\beta {D}_\beta x^\alpha + c_{\bmi}{}^\alpha \left( h^{\beta \gamma}\frac{\partial
    x^\epsilon}{\partial \wideparen{x}^\beta } \frac{\partial x^\delta}{\partial
    \wideparen{x}^\gamma}\slashed{D}_\delta \alpha^{\bmi}{}_\epsilon -f^{\bmi}(\bmx)
\right) = 0.
\]
This last observation lead us to the notion of harmonic map. A
diffeomorphism $\psi: \mathcal{S} \rightarrow \mathbb{S}^3$ is said to
be a harmonic map if
\begin{equation} 
h^{\beta\gamma}\frac{\partial x^\epsilon}{\partial \wideparen{x}^\delta}\frac{\partial
  x^\delta}{\partial \wideparen{x}^\gamma  }
\slashed{D}_{\delta}\alpha^{\bmi}{}_{\epsilon}=f^{\bmi}(x).
\label{wavemapcondition}   
\end{equation}
Notice that the requirement \eqref{wavemapcondition} implies that $
{D}^{\beta}{D}_{\beta} x^\alpha =0 $.

\medskip
In the present analysis it will be assumed that $\psi$ is the the
identity map. Under this assumption the wave map condition
\eqref{wavemapcondition}  reduces to $
h^{\beta\gamma}\slashed{D}_\gamma\alpha^{\bmi}{}_\beta=f^{\bmi}(x)$. However,
$\slashed{D}_\gamma \alpha^{\bmi}{}_\beta
=-\gamma_\gamma{}^\delta {}_\beta \alpha^{\bmi}{}_\delta$, which
finally implies that 
\[
f^{\bmi}(x)= -h^{\beta\gamma}\gamma_\gamma {}^\alpha {}_\beta \alpha^{\bmi}{}_\alpha.
\]
Observe that for $\mathbb{S}^3$ we have that
$\gamma_{\beta}{}^{\alpha}{}_{\gamma}=-\epsilon_{\beta}{}^{\alpha}{}_{\gamma}$, thus
$f^{\bmi}(x)=0$.

\medskip
\noindent 
\textbf{Remark.} The choice of the identity map ($\wideparen{x} = x$)
for $\psi$ means that we are going to use coordinates on
$\mathbb{S}^3$ to coordinatise $\mathcal{S}$.

\subsection{Data for a perturbation of the Milne Universe}

The purpose of this section is to provide a description of initial
data corresponding to perturbations of data for the Milne Universe. To
this end, first notice from the last subsection that although the
frame $\{\bmc_{\bmi}\}$ is $\bmhbar$-orthonormal, it is not necessarily
orthogonal respect to the intrinsic 3-metric $\bmh$ on
$\mathcal{S}$. Now, let $\{\bme_{\bmi} \}$ denote a $\bmh$-orthonormal
frame over $TS$ and let $\{ \bmomega^{\bmi}\}$ be the associate
cobasis. Now, assume that there exist vector fields $\{
\breve{\bme}_\bmi\}$ such that an $\bmh$-orthonormal frame
$\{\bme_{\bmi}\}$ is related to an $\bmhbar$-orthonormal frame
$\bmc_\bmi$ through $\bme_{\bmi}=\bmc_{\bmi}+ \breve{\bme}_{\bmi}
$. This last requirement is equivalent to introducing coordinates on
$\mathcal{S}$ such that
\begin{equation}
\bmh = \ring{\bmh} + \breve{\bmh} = \bmhbar + \breve{\bmh}. 
\label{hperturbation}
\end{equation}
Notice that the notation  $\hspace{0.1cm}\ring{}\hspace{0.1cm}$ is
used to denote the value in the exact (background) solution while
$\hspace{0.1cm}\breve{}\hspace{0.1cm} $ is used to denote the perturbation. 

In order to measure the size of the perturbed initial data, we
introduce the Sobolev norms defined for any spinor quantity
$N_{\mathcal{K}}$ with $_\mathcal{K}$ being an arbitrary string of
frame spinor indices, as
\[
||{N}_{ \mathcal{K} }||_{\mathcal{S}_{\star},m}\equiv
\sum_{{\mathcal{K}}}^{{}}|| {N}_{ \mathcal{K}}||_{\mathbb{S}^3,m}
\]
 where $\sum_{{\mathcal{K}}}^{{}}$ is the sum over all the frame
spinor indices encoded in $_\mathcal{K}$ and
\[
|| {N}_{ \mathcal{K}}||_{\mathbb{S}^3,m}= \left(
  \sum_{l=0}^{m}\sum_{\alpha_1,...,\alpha_l}^{3}\int_{\mathbb{S}^3}(\partial_{\alpha_1}...\partial_{\alpha_l}N_{\mathcal{K}})^2
  d^3x\right)^{1/2}. 
\]
Observe that since the indices in $_\mathcal{K}$ are frame indices,
the quantities $N_{\mathcal{K}}$ are scalars.

\medskip
Suppose now that the initial spacelike hypersurface $\mathcal{S}$ is
described  by the condition $\tau=\tau_0$. Then,  restricting the
results of  Proposition \ref{Proposition:InitialData} to the initial
hypersurface $\mathcal{S}$, we get the following initial data for the
wave equations of Proposition \ref{CEWE}:
\begin{eqnarray*}
& \ring{\Xi}|_{\mathcal{S}}=\cos\tau_{0}+ \cos \psi, \qquad
\ring{e}_{\bmA \bmA'}{}^{\bmb}|_{\mathcal{S}}=\sigma_{\bmA
  \bmA'}{}^{\bmb}, \qquad \ring{\Gamma}_{\bmA \bmA' \bmB
  \bmC}|_{\mathcal{S}}=-\frac{\mbox{i}}{\sqrt{2}}\epsilon_{\bmA(\bmB}\tau_{\bmC)
  \bmA'},& \\ 
&\ring{\Phi}_{\bmA \bmA' \bmB
  \bmB'}|_{\mathcal{S}}=\frac{1}{2}\tau_{\bmA \bmA'}\tau_{\bmB
  \bmB'}-\frac{1}{4}\epsilon_{\bmA \bmB}\epsilon_{\bmA \bmC},  \qquad \Sigma_{\bmA \bmA'}=-\sin \psi \sigma_{\bmA
   \bmA}{}^{3}|_{\mathcal{S}}, \qquad \ring{\phi}_{\bmA \bmB \bmC
   \bmD}|_{\mathcal{S}}=0,&  \\ 
& \ring{s}|_{\mathcal{S}}=-\frac{1}{2}(\cos\tau_0+ \cos \psi),
\qquad  \mathcal{P}\ring{\Xi}|_{\mathcal{S}}=
\Sigma|_{\mathcal{S}}=-\frac{1}{2}\sin\tau_0, \qquad 
\mathcal{P}\ring{e}_{\bmA \bmA'}{}^{\bmb}|_{\mathcal{S}}=0,
&\\
 & \mathcal{P}\ring{\Gamma}_{\bmA \bmA' \bmB
   \bmC}|_{\mathcal{S}}=0, \qquad \mathcal{P}\ring{\Phi}_{\bmA
   \bmA' \bmB \bmB'}|_{\mathcal{S}}= 0, &\\ 
&\mathcal{P}\ring{\phi}_{\bmA \bmB \bmC \bmD}|_{\mathcal{S}}=0,
\qquad \mathcal{P}\ring{s}|_{\mathcal{S}}=\frac{1}{2}\sin
\tau_{0}. &
\end{eqnarray*}
In a manner consistent with the split \eqref{hperturbation}, we make
use of the above expressions to consider a  perturbation of the
initial  data  on $\mathcal{S}$ of the form
\begin{eqnarray*}
  &{\Xi}|_{\mathcal{S}}= \ring{\Xi}|_{\mathcal{S}} +
  \breve{\Xi}|_{\mathcal{S}}, \qquad {e}_{\bmA
    \bmA'}{}^{\bmb}|_{\mathcal{S}}=\ring{e}_{\bmA
    \bmA'}{}^{\bmb}|_{\mathcal{S}} + \breve{e}_{\bmA
    \bmA'}{}^{\bmb}|_{\mathcal{S}}, &\\
&   {\Gamma}_{\bmA \bmA' \bmB \bmC}|_{\mathcal{S}}=\ring{\Gamma}_{\bmA
  \bmA' \bmB \bmC}|_{\mathcal{S}} + \breve{\Gamma}_{\bmA \bmA' \bmB
  \bmC}|_{\mathcal{S}}, \qquad {\Phi}_{\bmA \bmA' \bmB
  \bmB'}|_{\mathcal{S}}=\ring{\Phi}_{\bmA \bmA' \bmB
  \bmB'}|_{\mathcal{S}}+ \breve{{\Phi}}_{\bmA \bmA' \bmB
  \bmB'}|_{\mathcal{S}},  &\\
& {\phi}_{\bmA \bmB \bmC \bmD}|_{\mathcal{S}}=\breve{\phi}_{\bmA \bmB
  \bmC \bmD}|_{\mathcal{S}}, \qquad s|_{\mathcal{S}}=\ring{s}|_{\mathcal{S}}+ \breve{s}|_{\mathcal{S}},  &
\\ 
& \Sigma_{\bmA \bmA'}|_{\mathcal{S}}=\ring{\Sigma}_{\bmA \bmA'}|_{\mathcal{S}} + \breve{\Sigma}_{\bmA \bmA'}|_{\mathcal{S}}.&
\end{eqnarray*}
together with 
\begin{eqnarray*}
 &\Sigma|_{S}= \ring{\Sigma}|_{\mathcal{S}}+
 \breve{\Sigma}|_{\mathcal{S}},\qquad  \mathcal{P}{e}_{\bmA
   \bmA'}{}^{\bmb}|_{\mathcal{S}}=\mathcal{P}\breve{e}_{\bmA
   \bmA'}{}^{\bmb}|_{\mathcal{S}}, &\\
 &\mathcal{P}{\Gamma}_{\bmA \bmA' \bmB \bmC}|_{\mathcal{S}}=
 \mathcal{P}\breve{\Gamma}_{\bmA \bmA' \bmB \bmC}|_{\mathcal{S}},
 \qquad  \mathcal{P}{\Phi}_{\bmA \bmA' \bmB
   \bmB'}|_{\mathcal{S}}=\mathcal{P}\breve{\Phi}_{\bmA \bmA' \bmB
   \bmB'}|_{\mathcal{S}}  &\\ 
& \mathcal{P}{\phi}_{\bmA \bmB \bmC \bmD}|_{\mathcal{S}}=\mathcal{P}\breve{\phi}_{\bmA \bmB \bmC \bmD}|_{\mathcal{S}},  \qquad \mathcal{P}s|_{\mathcal{S}}=\mathcal{P}\ring{s}|_{\mathcal{S}} + \mathcal{P}\breve{s}|_{\mathcal{S}}.&
\end{eqnarray*}

The above fields are solution to the equations implied by the
initial data for the subsidiary system given in Proposition
\ref{Proposition:SubsidiaryEquations} \footnote{An example of how to
explicitly obtain the equations satisfied by the data has to satisfy,
was given in Section~\ref{sec:InitialDataWave-Weyl} for the case of
the rescaled Weyl spinor. This construction required the use of a
space spinor formalism. These equations are encoded in the requirement
that the zero-quantities vanish on the initial
hypersurface. }. Observe that these expressions are consistent with
equation \eqref{hperturbation}.

\subsection{Perturbed solutions}
\label{perturbedsolution}

Consistent with the discussion of the previous subsection, we will
split the field unknowns into a background and a perturbation
part. More precisely, we write
\begin{eqnarray*}
& \Xi = \ring{\Xi} + \breve{\Xi}, \qquad  \Sigma_{\bmA
  \bmA'}=\ring{\Sigma}_{\bmA \bmA'}+\breve{\Sigma}_{\bmA \bmA'},
\qquad e_{\bmA \bmA'}{}^{\bmb}= \ring{e}_{\bmA \bmA'}{}^{\bmb}
+ \breve{e}_{\bma \bmA'}{}^{\bmb}, \qquad  s=\ring{s}+
\breve{s},& \\ 
&\Gamma_{\bmA \bmA'}{}^{\bmB}{}_{\bmC}= \ring{\Gamma}_{\bmA
  \bmA'}{}^{\bmB}{}_{\bmC} + \breve{\Gamma}_{\bmA
  \bmA'}{}^{\bmB}{}_{\bmC}, \qquad \Phi_{\bmA \bmA' \bmB
  \bmB'}= \ring{\Phi}_{\bmA \bmA' \bmB \bmB'}+ \breve{\Phi}_{\bmA
  \bmA' \bmB \bmB'},& \\ 
&\phi_{\bmA \bmB \bmC \bmD}= \breve{\phi}_{\bmA \bmB \bmC \bmD}.&
\end{eqnarray*}
Following the notation used in the second remark in Section \ref{Section:SummaryWaveEquations}, we collect
the independent components of the unknowns as a single vector-valued variable
$\mathbf{u}$ and write
\[
\mathbf{u} = \ring{\mathbf{u}} + \breve{\mathbf{u}}.
\]
The components of the contravariant metric tensor $g^{\mu
\nu}(x,\mathbf{u})$ in the vector-valued wave equation
\eqref{genFormwithU} can be written as the metric for the background
solution $\ring{\mathbf{u}}$ plus a term depending on the unknown
$\mathbf{u}$
 \begin{equation}
g^{\mu\nu}(x;\mathbf{u})=\ring{g}^{\mu\nu}(x;\ring{\mathbf{u}}) +
\breve{g}^{\mu\nu}(x;\mathbf{u}),
\label{SimpleMetricSplit}
\end{equation}
The latter can be expressed, alternatively, in spinorial terms as
\begin{equation}\label{SplitMetricSpinor}
g^{\mu \nu}(x;\mathbf{u})=\epsilon^{\bmA\bmA'}\epsilon^{\bmB \bmB'}\bme_{\bmA \bmA'}{}^{\mu}\bme_{\bmB \bmB'}{}^{\nu} = \epsilon^{\bmA\bmA'}\epsilon^{\bmB \bmB'} (\ring{\bme}_{\bmA \bmA'}{}^{\mu}\ring{\bme}_{\bmB \bmB'}{}^{\nu} + \breve{e}_{\bmA \bmA'}{}^{\mu}\breve{e}_{\bmB \bmB'}{}^{\nu}  ). 
\end{equation}
Substituting the split \eqref{SimpleMetricSplit} into equation
\eqref{genFormwithU} we get,
\[
(g^{\mu\nu}(x;\ring{\mathbf{u}}) + \breve{g}^{\mu\nu}(x;\mathbf{u})
)\partial_{\mu}\partial_{\nu}(\ring{\mathbf{u}} + \breve{\mathbf{u}})
+ \bmF(x;\mathbf{u},\partial {\mathbf{u}}) =0. 
\]
Now, noticing that $\ring{\mathbf{u}}$ is, in fact, a solution to
\[
\ring{g}^{\mu\nu}(x;\ring{\mathbf{u}})\partial_{\mu} \partial_{\nu}\ring{\mathbf{u}} + \bmF(x; \ring{\mathbf{u}}, \partial \ring{\mathbf{u}} ) =0
\]
it follows then that
\[
\ring{g}^{\mu
  \nu}(x;\ring{\mathbf{u}})\partial_{\mu} \partial_{\nu}\breve{\mathbf{u}}
+ \breve{g}^{\mu
  \nu}(x;\mathbf{u})\partial_{\mu}\partial_{\nu}\ring{\mathbf{u}} +
\breve{g}^{\mu
  \nu}(x;\mathbf{u})\partial_{\mu}\partial_{\nu}\breve{\mathbf{u}} +
\bmF(x;\mathbf{u},\partial \mathbf{u}) -\bmF(x;
\ring{\mathbf{u}},\partial \ring{\mathbf{u}}  )  =0.
\]
Finally,  since the background solution $\ring{\mathbf{u}}$ is known then the last equation can be recast as
\[ 
(\ring{g}^{\mu \nu}(x) + \breve{g}^{\mu \nu}(x;
\breve{\mathbf{u}}))\partial_{\mu}\partial_{\nu}\breve{\mathbf{u}} =
\bmF(x;\breve{\mathbf{u}},\partial \breve{\mathbf{u}} ) .
\]
The above equation is in a form where the  local existence and Cauchy
stability theory of quasilinear wave equations as given in, say,
\cite{HugKatMar77} can be applied. Notice  that $\mathring{\bmg}(x)$
is Lorentzian since it corresponds to the metric of the background
solution ---i.e. the metric of the Einstein Cosmos. Now, consider
initial data $(\mathbf{u}_{\star}, \partial_t\mathbf{u}_{\star})$ close enough
to initial data
$(\ring{\mathbf{u}}_{\star}, \partial_t\ring{\mathbf{u}}_{\star})$ for
the Milne Universe ---that is, take
\begin{equation}
(\mathbf{u}_{\star}, \partial_{t}\mathbf{u}_{\star}) \in
B_{\varepsilon}(\ring{\mathbf{u}}_{\star},\partial_{t}\ring{\mathbf{u}}_{\star}
),
\label{NeighbourhoodData}
\end{equation}
 where the notion of closeness is encoded in 
\[
B_\varepsilon (\mathbf{u}_\star,\mathbf{v}_\star) \equiv \big\{
(\mathbf{w}_1,\mathbf{w}_2) \in  H^m(\mathcal{S},\mathbb{C}^N)
\times H^m(\mathcal{S},\mathbb{C}^N) \; | \; \parallel
\mathbf{w}_1 - \mathbf{u}_\star\parallel_{\mathcal{S},m}
+ \parallel \mathbf{w}_2 -
\mathbf{v}_\star\parallel_{\mathcal{S},m} \leq \varepsilon\big\}.
\]
 Using that $\mathbf{u}= \ring{\mathbf{u}} + \breve{\mathbf{u}}$, the
requirement \eqref{NeighbourhoodData} is equivalent to saying that the
initial data for the perturbation is small in the sense that
\begin{equation}\label{incondepsilon}
 ||\breve{\mathbf{u}}_{\star}   ||_{\mathcal{S},m} + ||\partial_{t}\breve{\mathbf{u}}_{\star} ||_{\mathcal{S},m} < \varepsilon .
 \end{equation}
With this remark in mind and recalling that $\ring{\mathbf{u}}$ is
explicitly known, observe that from equation \eqref{SimpleMetricSplit} it follows that
\[ 
g^{\mu\nu}(x;\mathbf{u})=\ring{g}^{\mu\nu}(x)+\breve{g}^{\mu\nu}(x;\breve{\mathbf{u}}).
\]
Since the variable $\breve{\mathbf{u}}_{\star}$ is a vector-valued
function collecting 
the independent components of the conformal fields,  and in particular 
$(\breve{e}_{\bmA \bmA'}{}^{\mu}{})_{\star}$, it follows that for
sufficiently small initial $\breve{\mathbf{u}}_{\star}$ the
perturbation $\breve{g}^{\mu \nu}$ will be small. Therefore, choosing
$\varepsilon$ small enough we can guarantee that the metric $\ring{g}^{\mu\nu}(x)+
\breve{g}^{\mu \nu}(x;\breve{\mathbf{u}}_{\star})$ is initially
Lorentzian. To pose an initial value problem for these equations we have 
to fix the inital data sets that are to be considered. In this case, 
we are interested in hyperboloidal initial data sets. The reason
behind this choice will be clarified in the sequel.
 In what follows, consider initial data $\mathbf{u}_\star$ for
 the conformal Einstein equations prescribed on an hypersurface 
\[
\mathcal{H} \equiv\{p\in \mathbb{R}\times \mathbb{S}^3 \,|\, 0 \leq \psi (p)\leq \pi-\tau_{0}, \,\tau(p)= \tau_{0} \}
\]
with boundary
 \[
\mathcal{Z} \equiv\{p\in \mathbb{R}\times \mathbb{S}^3 \,|\,
\psi(p)=\pi-\tau_{0},\, \tau(p)=\tau_{0} \}
\] 
for some constant $\tau_0>0$ and 
where, recalling the discussion of Section \ref{Section: Wave maps},
the coordinates $(x^{\alpha})$ on $\mathbb{S}^3$ are being used to
coordinatise $\mathcal{S}$. Therefore $\mathcal{H}$ is to be regarded
as a subset of $\mathcal{S} \approx \mathbb{S}^3$. The initial
data is said to be an hyperboloidal data set if
\begin{equation}\label{initialHyperboloidalData}
\Xi|_{\mathcal{H} \backslash \mathcal{Z}}>0, \hspace{1cm}
\Xi|_{\mathcal{Z}}=0, \hspace{1cm}
(\Sigma_{\bma}\Sigma^{\bma})|_{\mathcal{Z}}=0, \hspace{1cm}
\mathcal{P}\Xi|_{\mathcal{Z}}=\Sigma < 0 .
\end{equation}
The above conditions tell us, \emph{a priori}, that $\Xi$ acts as a
boundary defining function which vanishes on $\mathcal{Z}$.  This is a
key observation that will be used in Section
\ref{sec:conformalboundary} where the conformal structure of the
solutions is to be analysed. To arrive to a similar conclusion about
$\Sigma$ let us shift the coordinate $\tau$ by an amount of
$\frac{1}{2}\pi$, namely $\wideparen{\tau}=\tau-\frac{1}{2}\pi$. 

In the case of the background solution $\mathring{\mathbf{u}}$ we can
analogously define an hypersurface $\mathring{\mathcal{H}}$ with
boundary $\mathring{\mathcal{Z}}$ where $\mathring{\mathcal{H}}$ is as
a region of $\mathbb{S}^3$ satisfying analogous conditions as those
given in \eqref{initialHyperboloidalData}.  In that case, the initial
hypersurface $\wideparen{\tau}=0$ corresponds to the so-called
\emph{standard Minkowski hyperboloid}. The conformal factor for the
background solution, ---i.e. the Milne spacetime is given by
\[ 
\ring{\Xi}= \cos\psi - \sin\wideparen{\tau}.  
\]
Therefore
\begin{eqnarray*}
\ring{\Sigma}=\mathcal{P}\ring{\Xi} = -\cos\wideparen{\tau}, \qquad
\ring{\Sigma}_{3}=  -\sin \psi.
\end{eqnarray*}
Observe that $\mathring{\Sigma}<0$ for $\wideparen{\tau} \in
(0,\frac{1}{2}\pi)$ and $\ring{\Sigma}>0$ for $\wideparen{\tau}\in
(\frac{1}{2}\pi,\frac{3}{2}\pi)$. Thus, recalling that $\Sigma =
\ring{\Sigma} + \breve{\Sigma}$ then $\Sigma<0$ for
$\wideparen{\tau}\in (0,\frac{1}{2}\pi)$ and $\Sigma>0$ for
$\wideparen{\tau}\in (\frac{1}{2}\pi,\frac{3}{2}\pi)$ for $\varepsilon$
small enough, therefore there is at least one point where $\Sigma=0$.

\medskip
Now, recall that the initial data
$(\mathbf{u}_{\star},\partial_{\tau}\mathbf{u}_{\star} )$ is only
defined in the initial hypersurface
$\mathcal{H}$ while the data
$(\ring{\mathbf{u}}_\star, \partial_{\tau}\ring{\mathbf{u}}_\star )$
is defined in the whole of $\mathcal{S} \approx \mathbb{S}^3$. To be able to apply
the theory of \cite{HugKatMar77} we need to extend the data
$(\mathbf{u}_\star,\partial_{\tau}\mathbf{u}_\star) $ to
$\mathcal{S}_{\star}$. In what follows let
$\mathbf{w}_\star=(\mathbf{u}_{\star},\partial_{\tau}\mathbf{u}_{\star}
)$.  We can extend the initial data by invoking the \emph{Extension Theorem}
which states that  there exists a linear operator 
\[
\mathcal{E}:H^{m}(\mathcal{H}, \mathbb{C}^N)\rightarrow
H^{m}(\mathbb{S}^3,\mathbb{C}^N)
\]
 such that if $\mathbf{w}_{\star} \in H^{m}(\mathcal{H},\mathbb{C}^N)$
then $\mathcal{E}{\mathbf{w}}_{\star}(x)={\mathbf{w}}_{\star}(x)$
almost everywhere in $\mathcal{H}$ and
\[ 
||\mathcal{E}\mathbf{w}_\star||_{m,\mathcal{S}}\leq
K||\mathbf{w}_{\star}||_{m,\mathcal{H}}
\]
 where $K$ is a universal constant for fixed $m$ ---see e.g. \cite{Eva98}. Hence, using
equation $\eqref{incondepsilon}$ we can make $
||\mathcal{E}\mathbf{w}_\star||_{m,\mathcal{S}}$ small as necessary by
making $\varepsilon$ small ---that is, the size of the extended data is 
controlled by the data in the initial hypersurface $\mathcal{H}$. Therefore, the extended data will be given
by
\[
\mathcal{E}\mathbf{u}_{\star}=\ring{\mathbf{u}}_{\star}+
\mathcal{E}\breve{\mathbf{u}}_{\star} \qquad
\mathcal{E}\partial_{\tau}\mathbf{u}_{\star}=\partial_{\tau}\ring{\mathbf{u}}_{\star}+
\mathcal{E}\partial_{\tau}\breve{\mathbf{u}}_{\star}
\]
which are well defined on $H^{m}(\mathbb{S}^3,\mathbb{C}^N)$. Using
equation \eqref{incondepsilon} we observe that
\[ 
||\mathcal{E}\breve{\mathbf{u}}_\star||_{\mathcal{H},m} + ||
\mathcal{E}\partial_{\tau}\breve{\mathbf{u}}_\star||_{\mathcal{H},m}
\leq K \varepsilon . 
\]

\medskip
\noindent
\textbf{Remark.} The fact that the extension of the data obtained in
the previous paragraph is not unique and it does not necessarily
satisfy the constraints of Proposition \ref{ReductionLemma} is not a
problem in our analysis since
\[
D^{+}(\mathcal{H})\cap I^{+}(\mathcal{S}\backslash \mathcal{H})=\emptyset. 
\]
The proof of the last statement follows by contradiction. Let $q\in
D^{+}(\mathcal{H})\cap I^{+}(\mathcal{S}\backslash
\mathcal{H})$. Then, in the one hand we have that $q \in
I^{+}(\mathcal{S}\backslash \mathcal{H})$, so that it follows that there
exists a future timelike curve $\gamma$ from $p \in
\mathcal{S}\backslash\mathcal{H}$ to $q$. On the other hand $q \in
D^{+}(\mathcal{H}) $ which means that every past in extendible causal
curve through $q$ intersects $\mathcal{H}$, therefore $p \in
\mathcal{H}$. This is a contradiction since $p \in
\mathcal{S}\backslash\mathcal{H}$.

\medskip

We are now in position to make use of a local
existence and Cauchy stability result adapted from \cite{HugKatMar77}
---see Appendix \ref{Appendix:PDETheory}, to establish the following theorem:

\begin{theorem}[\textbf{\em Existence and Cauchy stability}]  
\label{Theorem:ExistenceCauchyStability}
Let $(\mathbf{u}_\star,\partial_{\tau}\mathbf{u}_\star )=(
\ring{\mathbf{u}}_\star +
\breve{\mathbf{u}}_\star,\partial_{\tau}\ring{\mathbf{u}}_\star
+ \partial_{\tau}\breve{\mathbf{u}}_\star)$ be hyperboloidal initial
data for the conformal wave equations on an 3-dimensional manifold
$\mathcal{H}$ where
$(\ring{\mathbf{u}}_\star,\partial_{\tau}\ring{\mathbf{u}}_\star)$
denotes initial data for the Milne Universe. Let
$(\mathcal{E}\mathbf{u}_\star,\mathcal{E}\partial_{\tau}\mathbf{u}_\star)$
denote the extension of these data to
$\mathcal{S}_{}\approx\mathbb{S}^3$. Then, for $m \geq 4$ and
$\wideparen{\tau}_{\bullet}\geq \frac{3}{2}\pi$ there exist an
$\varepsilon > 0$ such that:

\begin{itemize}
\item[(i)] For $||  \breve{\mathbf{u}}_\star ||_{\mathcal{H},m} + || \partial_{\tau}\breve{\mathbf{u}}_\star||_{\mathcal{H},m} < \varepsilon $, there exist a unique solution
 $\mathbf{u}=\ring{\mathbf{u}}+ \breve{\mathbf{u}} $ to the  wave equations of Proposition \ref{CEWE} with a minimal existence interval $[0,\wideparen{\tau}_{\bullet}]$ and
 $\mathbf{u} \in C^{m-2}([0,\wideparen{\tau}_{\bullet} \times \mathcal{S}], \mathbb{C}^{N}) $.

\item[(ii)] Given a   sequence $(\mathbf{u}_\star^{(n)},\mathbf{v}^{(n)}_\star) \in
  B_\varepsilon(\mathbf{u}_\star,\mathbf{v}_\star) \cap D_\delta$ such that
\[
\parallel \mathbf{u}_\star^{(n)} -
\mathbf{u}_\star\parallel_{\mathcal{S},m} \rightarrow 0, \qquad \parallel \mathbf{v}_\star^{(n)} -
\mathbf{v}_\star\parallel_{\mathcal{S},m} \rightarrow 0, \qquad
\mbox{as} \qquad n\rightarrow \infty,
\]
then for the solutions $\mathbf{u}^{(n)}$ with
$\mathbf{u}^{(n)}=\mathbf{u}_\star^{(n)}$ and $\partial_t \mathbf{u}^{(n)}(0,\cdot)=\mathbf{v}_\star^{(n)}$, it holds that 
\[
\parallel
\mathbf{u}^{(n)}(t,\cdot)-\mathbf{u}(t,\cdot) \parallel_{\mathcal{S},m}
\rightarrow 0 \qquad \mbox{as} \qquad n\rightarrow \infty 
\]
uniformly in $t\in[0,\tau_{\bullet})$ as $n\rightarrow \infty$. 

\item[(iii)] The solution $\mathbf{u}=\ring{\mathbf{u}}+ \breve{\mathbf{u}}$ is unique in $D^{+}(\mathcal{H})$ and implies, wherever $\Xi\neq 0$, a $C^{m-2}$ solution to the Einstein vacuum equations with vanishing cosmological constant .

\end{itemize}

\end{theorem}

\begin{proof}
Points \emph{(i)} and \emph{(ii)} are a direct application of Theorem
\ref{Theorem:HugKatMar} given in Appendix
\ref{Appendix:PDETheory}. The condition ensuring that
$\ring{\bmg}(x)+\breve{\bmg}{(x;\bmu)}$ is Lorentzian is encoded in
the requirement of the perturbation for the initial data being small
as discussed in Section \ref{perturbedsolution}.

The statement of point \emph{(iii)} follows from the discussion of
Section \ref{PropagationConstraintsSubsidiarySystem} for the
propagation of the constraints and the subsidiary system as summarised
in Propositions \ref{CEWE} and \ref{ReductionLemma}. In particular, in
this section it was shown that a solution to the spinorial wave
equations is a solution to the conformal Einstein field equations if
initial data satisfies the appropriate conditions. As exemplified in
Section \ref{sec:InitialDataWave-Weyl} for the rescaled Weyl spinor,
requiring the zero-quantities to vanish in the initial hypersurface
renders conditions on the initial data.  Finally, recall that a
solution to the CEFE implies a solution to the Einstein field
equations wherever $\Xi \neq 0$ ---see \cite{Fri83}.

\end{proof}


\subsection{Structure of the conformal Boundary}
\label{sec:conformalboundary}

Now, we will complement Theorem \ref{Theorem:ExistenceCauchyStability}
by showing that the conformal boundary $\mathscr{I}$ coincides
with the Cauchy horizon of $\mathcal{H}$. The argument of this section
is based on analogous discussion in \cite{Fri86b}.  Since the Cauchy
horizon $H(\mathcal{H})=\partial(D^+(\mathcal{H}))$ is generated by
null geodesics with endpoints on $\mathcal{Z}$ the null generators of
$H(\mathcal{H})$ ---i.e the null vectors tangent to
$H(\mathcal{H})$--- are given at $\mathcal{Z}$ by
$\Sigma_{\bma}|_{\mathcal{Z}}$ as it follows by the initial
hyperboloidal data \eqref{initialHyperboloidalData}. We then define
two null vectors $(\bmn, \bml)$ on $\mathcal{Z}$ by setting
\begin{equation} 
l_{\bma \star}=\Sigma_{\bma}|_{\mathcal{Z}}, \qquad \bmn_\star
\perp \mathcal{Z}, \qquad  \bmg(\bmn_\star,\bml_\star)=1
\qquad  \text{on $\mathcal{Z}$.}
\label{initialNPtetrad}
\end{equation}
We complement this pair of null vectors $\{\bml_\star,\bmn_\star\}$,
where $\bml_\star$ is tangent to $H(\mathcal{H})$ on $\mathcal{Z}$ and
$\bmn_\star$ is normal to $\mathcal{Z}$, with a pair of complex
conjugate vectors $\bmm_\star$ and $\bar{\bmm}_\star$ tangent to
$\mathcal{Z}$ such that $\bmg(\bmm_\star , \bar{\bmm}_\star)=1$, so as
to obtain the tetrad
$\{\bml_\star,\bmn_\star,\bmm_\star,\bar{\bmm}_\star \}$. In order to
obtain a Newman-Penrose frame $\{ \bml, \bmn, \bmm, \bar{\bmm}\}$ off
$\mathcal{Z}$ along the null generators of $H(\mathcal{H})$ we
propagate them by parallel transport in the direction of $\bml$ by
requiring
\begin{equation}
l^{\bma}\nabla_{\bma}l^{\bmb} =0, \qquad
l^{\bma}\nabla_{\bma}n^{\bmb}=0, \qquad
l^{\bma}\nabla_{\bma}m^{\bmb}=0.
\label{parallelalongl}
\end{equation}
 
Now, suppose that we already have a solution to the conformal wave
equations.  Using the result of Proposition \ref{ReductionLemma}, we know that the
solution will also satisfy the CEFE. In this section we will make use
of the CEFE equations to study the conformal boundary. From the tensorial (frame) version of the CEFE  as given in Appendix
\ref{Appendix:CFE}, one notices the following subset consisting of equations
\eqref{CEFEsecondderivativeCF}, \eqref{CEFEs} and the definition of
$\Sigma_{\bma}$ as the gradient of the conformal factor:
\begin{subequations}
\begin{eqnarray}
&& \nabla_{\bma}\Xi=\Sigma_{\bma},  \label{subsistemCEFE1} \\
&& \nabla_{\bma} \Sigma_{\bmb}=sg_{\bma \bmb}-\Xi L_{\bma \bmb}, \label{subsistemCEFE2} \\
&& \nabla_{\bma}s=-L_{\bma \bmb}\Sigma^{\bmb}. \label{subsistemCEFE3}
\end{eqnarray} 
\end{subequations}

Transvecting the first two equations, respectively, with $l^\bma$, $l^{\bma}l^{\bmb}$ and $l^{\bma}m^{\bmb}$ we get
\begin{eqnarray*}
&& l^{\bma}\nabla_{\bma}\Xi=l^{\bma}\Sigma_{\bma}, \\
&& l^{\bma} \nabla_{\bma} (l^{\bmb}\Sigma_{\bmb})=-\Xi L_{\bma \bmb}l^{\bma}l^{\bmb}, \\
&& l^{\bma} \nabla_{\bma} (m^{\bmb}\Sigma_{\bmb})=-\Xi L_{\bma \bmb}l^{\bma}m^{\bmb}, 
\end{eqnarray*}
where we have used \eqref{parallelalongl} and the fact that $\bml$ is
null and orthogonal to $\bmm$. The latter equations can be read as a
system of homogeneous transport equations along the integral curves of
$\bml$ for a vector-valued variable containing as components $\Xi$ ,
$\Sigma_{\bma}l^{\bma}$ and $\Sigma_{\bma}m^{\bma}$. Written in matricial form
one has
\begin{equation}\label{matrixHomogeneous}
 \nabla_{\bml} \begin{pmatrix} \Xi \\ \Sigma_{\bma}l^{\bma} \\ \Sigma_{\bma}m^{\bma} \end{pmatrix} = \begin{pmatrix} 0 & 1 & 0 \\-L_{\bmc \bmd}l^{\bmc}l^{\bmd} & 0 & 0 \\ -L_{\bmc\bmd}l^{\bmc}m^{\bmd} & 0 & 0  \end{pmatrix} \begin{pmatrix}  \Xi \\ \Sigma_{\bma}l^{\bma} \\ \Sigma_{\bma}m^{\bma} \end{pmatrix} 
\end{equation}
Observe that the column vector shown in the last equation is zero on
$\mathcal{Z}$, since $\Xi|_{\mathcal{Z}}$ =0 ,
$(l^{\bma}{}\Sigma_{\bma})|_{\mathcal{Z}}=(l^{\bma}{}l_{\bma})|_{\mathcal{Z}}=0$
and $(\Sigma_{\bma}m^{\bma})|_{\mathcal{Z}}=
(l_{\bma}{}m^{\bma})|_{\mathcal{Z}}=0$ which follows from
\eqref{initialHyperboloidalData} and 
\eqref{initialNPtetrad}. Since
equation \eqref{matrixHomogeneous} is homogeneous and it has vanishing
initial data on $\mathcal{Z}$ we have that $\Xi$,
$\Sigma_{\bma}l^{\bma}$ and $\Sigma_{\bma}m^{\bma}$ will be zero along
$\bml$ until one reaches a caustic point. Consequently, we conclude
that the conformal $\Xi$ factor vanishes in the portion of
$H(\mathcal{H})$ which is free of caustics. Thus, this portion of
$H(\mathcal{H})$ can be interpreted as the conformal boundary of the
physical spacetime $(\tilde{\mathcal{M}},\tilde{\bmg})$. In
addition, notice that from the vanishing of the column vector of
equation \eqref{matrixHomogeneous} it follows that
$\Sigma_{\bma}l^{\bma}=\Sigma_{\bma}m^{\bma}=0$ on
$H(\mathcal{H})$. Therefore, the only component of $\Sigma_{\bma}$
that can be different from zero is $\Sigma_{\bma}n^{\bma}$. Accordingly,
$\Sigma^{\bma}$ is parallel to $l^{\bma}$ and $\Sigma^{\bma}=
(\Sigma_{\bmc}n^{\bmc})l^{\bma}$. Moreover, since $\bmg(\bmn_\star ,
\bml_\star)=1$ it follows that
$(\Sigma_{\bma}n^{\bma})|_{\mathcal{Z}}=1$ \footnote{This can also be
shown by noticing
that$(n^{\bmb}\Sigma_{\bmb})|_{\mathcal{Z}}=(n^{\bmb}l_{\bmb})|_{\mathcal{Z}}=1$.}. 

Now,
in order to extract the information contained in
$\Sigma_{\bma}n^{\bma}$ one transvects \eqref{subsistemCEFE2} with
$l^{\bma}n_{\bmb}$, to obtain
\[ 
l^{\bma}\nabla_{\bma}(n^{\bmb}\Sigma_{\bmb})= sg_{\bma
  \bmb}l^{\bma}n^{\bmb}-\Xi L_{\bma \bmb}l^{\bma}n^{\bmb}.
\]
Using that $\bmg(\bml,\bmn)=1$ and that $\Xi$ vanishes on
$H(\mathcal{H})$  one concludes that
\[ 
\nabla_{\bml}(\Sigma_{\bma}n^{\bma})=s \qquad \text{on
  $H(\mathcal{H})$}. 
\]
We can obtain a further equation transvecting \eqref{subsistemCEFE3} with $l^{\bma}$ 
\[
l^{\bma}\nabla_{\bma}s=-L_{\bma \bmb}l^{\bma}\Sigma^{\bmb}=-L_{\bma
  \bmb}l^{\bma}\Sigma_{\bmf}n^{\bmf}l^{\bmb} \qquad \text{on
  $H(\mathcal{H})$}.
\]
It follows then that one has the system
\[
\nabla_{\bml}\begin{pmatrix}  \Sigma_{\bma}n^{\bma} \\ s \end{pmatrix} = \begin{pmatrix}  0 & 1 \\ -L_{\bmc \bmd}l^{\bmc}l^{\bmd} & 0 \end{pmatrix} \begin{pmatrix}  \Sigma_{\bma}n^{\bma} \\ s \end{pmatrix}.
 \]
Since $(\Sigma_{\bma}n^{\bma})|_{\mathcal{Z}}=1$ (i.e. non-vanishing),
the solution for the column vector formed by $s$ and
$\Sigma_{\bma}n^{\bma}$ cannot be zero. Accordingly, $s$ and
$\Sigma_{\bma}n^{\bma}$ cannot vanish simultaneously. Finally,
transvecting equation \eqref{subsistemCEFE2} with
$m^{\bma}\bar{m}^{\bmb}$ we get
\[
m^{\bma}\bar{m}^{\bmb}\nabla_{\bma}\Sigma_{\bmb}=-\Xi m^{\bma}\bar{m}^{\bmb}L_{\bma \bmb}+ sg_{\bma \bmb}m^{\bma}\bar{m}^{\bmb}.
 \]
 Using that $\bmg(\bmm,\bar{\bmm})=1$ and restricting to $H(\mathcal{H})$ where $\Xi=0$ we obtain
\[ 
\bar{m}^{\bmb}m^{\bma}\nabla_{\bma}\Sigma_{\bmb}= s  \qquad \text{on  $H(\mathcal{H}).$}
\]
 Using $\bmg(\bar{\bmm},\bml)=0$ it follows that the left hand side of the last equation is equivalent to
\begin{eqnarray*}
&&
m^{\bma}\bar{m}^{\bmb}\nabla_{\bma}\Sigma_{\bmb}=m^{\bma}\bar{m}^{\bmb}\nabla_{\bma}(\Sigma_{\bmc}n^{\bmc}l_{\bmb})\\
&& \phantom{m^{\bma}\bar{m}^{\bmb}\nabla_{\bma}\Sigma_{\bmb}}=m^{\bma}\bar{m}^{\bmb}\Sigma_{\bmc}n^{\bmc}\nabla_{\bma}l_{\bmb}
+
m^{\bma}\bar{m}^{\bmb}l_{\bmb}\nabla_{\bma}\Sigma_{\bmc}n^{\bmc}\\
&& \phantom{m^{\bma}\bar{m}^{\bmb}\nabla_{\bma}\Sigma_{\bmb}}=\Sigma_{\bmc}n^{\bmc}
m^{\bma}\bar{m}^{\bmb}\nabla_{\bma}l_{\bmb}. 
\end{eqnarray*}
Finally, recalling the definition of the expansion
$\rho\equiv-m^{\bma}\bar{m}^{\bmb}\nabla_{\bma}l_{\bmb}$ (in the Newman-Penrose notation \cite{Ste91}) we finally obtain
\[ 
\Sigma_{\bma}n^{\bma}\rho = -s \qquad \text{on
  $H(\mathcal{H}).$}
\]

We already know that the only possible non-zero component of the
gradient of $\Xi$ is $\Sigma_{\bma}n^{\bma}$ and that it cannot vanish
simultaneously with $s$. This means that $\mathbf{d} \Xi =0$ implies
$ \rho \rightarrow \infty$ on $H(\mathcal{H})$. To be able to identify
the point $i^{+}$ where $\mathbf{d}\Xi=0$ with timelike infinity we
need to calculate the Hessian of the conformal factor. Observe that
this information is contained in the conformal field equation
\eqref{CEFEsecondderivativeCF}. Considering this equation at
$H(\mathcal{H})$, where we have already shown that the conformal
factor vanishes, we get
\[
\nabla_{\bma} \nabla_{\bmb}\Xi=sg_{\bma \bmb}. 
\]
Now, as we have shown that $s$ and $\Sigma_{\bma}n^{\bma}$
(or, equivalently, $\mathbf{d} \Xi$) do not vanish simultaneously we
conclude that $s \neq 0$ and that $\nabla_{\bma} \nabla_{\bmb}\Xi$ is
non-degenerate. Thus, we can consider the point $i^{+}$ on $
H(\mathcal{H})$ where both $\Xi$ and $\mathbf{d} \Xi$ vanish as
representing future timelike infinity for the physical spacetime
$(\tilde{\mathcal{M}},\tilde{g})$.

\medskip
\noindent 
\textbf{Remark.} Observe that the construction discussed in the
previous paragraphs crucially assumes that
$\Xi_\star$ is zero on the boundary $\mathcal{Z}$ of the initial
hypersurface $\mathcal{H}$. This construction cannot be repeated
if we were to take another hypersurface $\mathcal{H}'$ with
boundary $\mathcal{Z}'$ where the conformal factor does not
vanish. This is the case of an initial hypersurface that intersects
the cosmological horizon, where for the reference solution the
conformal factor does not vanish ---see Figure \ref{fig:hypdata}. 

\begin{figure}[t]
\centerline{\includegraphics[width=0.4\textwidth]{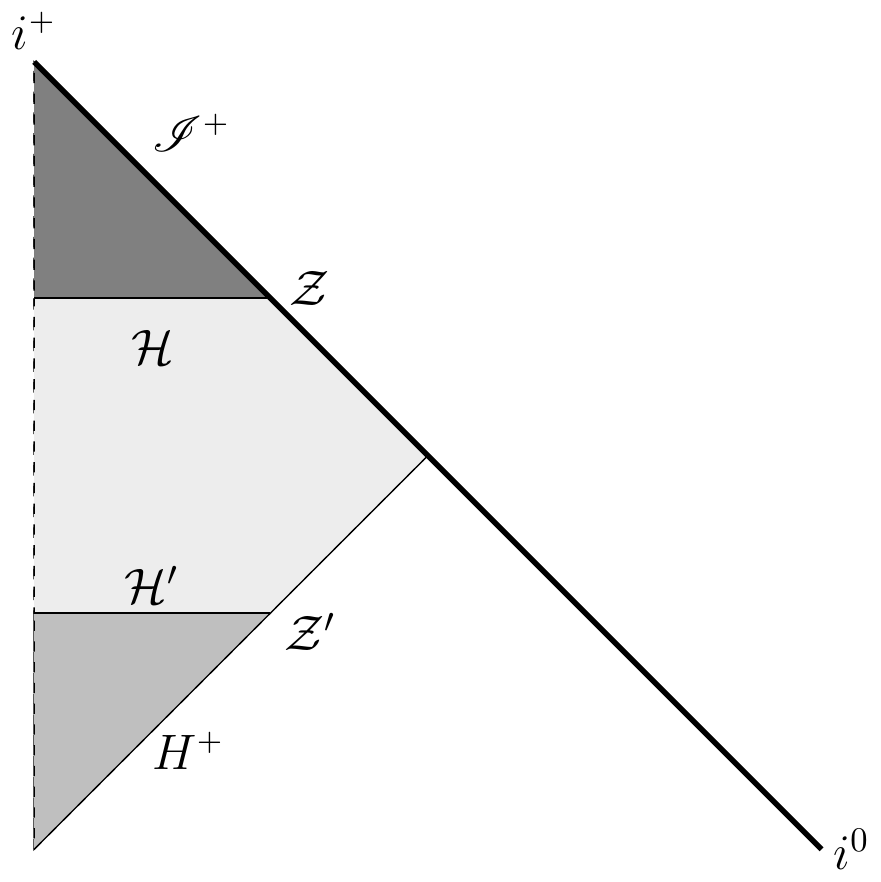}}
\caption{Portion of the Penrose diagram of the Milne
Universe showing the initial hypersurface $\mathcal{H}$ where the the
hyperboloidal data is prescribed. At $\mathcal{Z}$ the conformal
factor vanishes and the argument of Section
\ref{sec:conformalboundary} can be applied. The dark gray area
represents the development of the data on $\mathcal{H}$. Compare with
the case of the hypersurface $\mathcal{H}'$ which intersects the
horizon at ${\mathcal{Z}'}$ where the argument cannot be
applied. Analogous hypersurfaces can be depicted for the lower diamond
of the complete diagram of Figure \ref{fig:Milne}.} 
\label{fig:hypdata}
 \end{figure}

\medskip

The results of the analysis of this section are 
summarised in the following:

\begin{proposition}  \textbf{(Structure of the conformal boundary)}
Let $\mathbf{u}$ denote a solution to the conformal wave equations
equations constructed as described in Theorem
\ref{Theorem:ExistenceCauchyStability}, then, there exists a point
$i^{+} \in H(\mathcal{H})$ where $\Xi|_{i^{+}}=0$ and
$\mathbf{d}\Xi|_{i^{+}}=0$ but the Hessian
$\nabla_{\bma}\nabla_{\bmb}\Xi|_{i^+}$ is non-degenerate. In addition,
$\mathbf{d}\Xi \neq 0$ on $\mathscr{I}^{+} =
H(\mathcal{H})\setminus\{i^+\}$. Moreover
$D^{+}(\mathcal{H})=J^{-}(i^+)$.
\end{proposition}

\begin{proof}

From the conclusions of Theorem  \ref{Theorem:ExistenceCauchyStability} and
 the discussion of Section \ref{sec:conformalboundary} it follows
that if we have a solution to the conformal wave equations which, in
turn implies a solution to the conformal field equations, then there
exists a point $i^{+}$ in $H(\mathcal{H})$ where both the conformal
factor and its gradient vanish but $\nabla_{\bma}\nabla_{\bmb}\Xi$ is
non-degenerate. This means that $i^{+}$ can be regarded as future
timelike infinity for the physical spacetime. In addition, null
infinity $\mathscr{I}^{+}$ will be located at
$H(\mathcal{H})\backslash\{i^{+}\}$ where the conformal factor
vanishes but its gradient does not.
\end{proof}

\section{Conclusions}
\label{Section:Conclusions}

In this article we have shown that the spinorial frame version of the
CEFE implies a system of quasilinear wave equations for the various
conformal fields. The use of spinors allows a systematic and clear
deduction of the equations and the not less important issue of the
propagation of the constraints. The fact that the metric is not part
of the unknowns in the spinorial formulation of the CEFE simplifies
the considerations of hyperbolicity of the operator
$\square\equiv\nabla_{\bma}\nabla^{\bma}$. The application of these
equations to study the semiglobal stability of the Milne Universe
exemplifies how the extraction of a system of quasilinear wave
equations out of the CEFE allows to readily make use of the general theory of
partial differential equations to obtain non-trivial statements about
the global existence of solutions to the Einstein field equations. The
analysis of the present article has been restricted to the vacuum
case. However, a similar procedure can be carried out, in the
non-vacuum case, for some suitable matter models with trace-free
energy-momentum tensor ---see e.g. \cite{Fri91}. 

In addition, the present analysis  has  been restricted to the case of
the so-called standard CEFE. There exists another more
general version of the CEFE, the so-called, extended conformal
Einstein field equations (XCEFE) in which the various equations are
expressed in terms of a Weyl connection ---i.e. a torsion free
connection which is not necessarily metric but, however, respects the
structure of the conformal class, \cite{Fri95}.  The hyperbolic
reduction procedures for the XCEFE available in the literature do not
make use of gauge source functions. Instead, one makes use of
conformal Gaussian systems based on the congruence of privileged
curves known as conformal geodesics to extract a first order symmetric
hyperbolic system. It is an interesting open question to see whether
it is possible to use conformal Gaussian systems to deduce wave
equations for the conformal fields in the XCEFE.

\section*{Acknowledgments}

E. Gasper\'in gratefully acknowledges the support from Consejo
Nacional de Ciencia y Tecnolog\'ia (CONACyT Scholarphip
494039/218141).

\appendix

\section{Appendix: spinorial relations}
\label{Appendix:SpinorialRelations}

In this appendix we recall several relations and identities that are
used repeatedly throughout this article ---see \cite{PenRin84}.
 In addition, using the remarks made in \cite{Pen83}  we give a generalisation for the spinorial Ricci  identities for a connection which is metric but not necessarily torsion free.

\subsection{The Levi-Civita case}
In this subsection we recall some well known relations satisfied the
curvature spinors of a Levi-Civita connection. The discussion of this
subsection follows \cite{PenRin84}. First recall the decomposition of
a general curvature spinor
\[
\grave{R}_{\bmA \bmA' \bmB \bmB' \bmC \bmC' \bmD \bmD'}=
\grave{R}_{\bmA \bmB \bmC \bmC' \bmD \bmD' } \epsilon_{\bmB' \bmA'} +
\grave{\bar{R}}_{\bmA' \bmB' \bmC \bmC' \bmD \bmD'} \epsilon_{\bmB
\bmA}.
\]
We can further decompose the reduced spinor $\grave{R}_{\bmA \bmB \bmC \bmC' \bmD \bmD' }$ as
\[
\grave{R}_{\bmA \bmB \bmC \bmC' \bmD \bmD'}= \grave{X}_{\bmA \bmB \bmC
\bmD}\epsilon_{\bmC' \bmD'} + \grave{Y}_{\bmA \bmB \bmC'
\bmD'}\epsilon_{\bmC \bmD},
\]
where 
\[
\grave{X}_{\bmA \bmB \bmC \bmD} \equiv \tfrac{1}{2}\grave{R}_{\bmA \bmB(\bmC |\bmE'| \bmD)}{}^{\bmE'}   \hspace{1cm}\grave{Y}_{\bmA \bmB \bmC \bmD}\equiv \tfrac{1}{2}\grave{R}_{\bmA \bmB \bmE( \bmC'}{}^{\bmE}{}_{\bmD')}.
\]
In the above expressions the symbol ${\grave{\phantom{X}}}$ over the Kernel
letter indicates that this relation is general ---i.e. the
connection is not necessarily neither metric nor torsion-free. The
spinors $\grave{X}_{\bmA \bmB \bmC \bmD}$ and $\grave{Y}_{\bmA \bmB
\bmC \bmD}$ are not necessarily symmetric in ${}_{\bmA \bmB}$.

\medskip
It is well known that if that the connection is metric, then the
spinors $ \widetriangle{X}_{\bmA \bmB \bmC \bmD} $ and $
\widetriangle{Y}_{\bmA \bmB \bmC' \bmD'}$ have the further symmetries
\begin{equation}
\widetriangle{X}_{\bmA \bmB \bmC \bmD}= \widetriangle{X}_{(\bmA \bmB)
  \bmC \bmD}, \qquad  \widetriangle{Y}_{\bmA \bmB \bmC' \bmD'} =
\widetriangle{Y}_{(\bmA \bmB) \bmC' \bmD'}. 
\label{SymmetriesMetricity}
\end{equation}
We add the symbol $\widetriangle{\phantom{X}}$ over
the Kernel letter to denote that only the metricity of the connection
is being assumed. Now, if the connection is not only metric but, in
addition, is torsion free (i.e. it is a Levi-Civita connection) then
the \emph{first Bianchi identity} $ R_{\bma[\bmb \bmc \bmd]}=0$ can be
written in spinorial terms as
\[
R_{\bmA \bmB}{}^{\bmC \bmD}{}_{\bmA' \bmB'}{}^{\bmD' \bmC'}=0,
\]
which implies that $ X_{\bmS \bmP}{}^{\bmS \bmP} = \bar{X}_{\bmS'
\bmP'}{}^{\bmS' \bmP'}$. Accordingly $X_{\bmS \bmP}{}^{\bmS \bmP}$ is
a real scalar and $Y_{\bmA \bmB \bmA' \bmB'}$ is a Hermitian spinor
which, following the notation of \cite{PenRin84}, will be denoted by
$\Phi_{\bmA \bmB \bmA' \bmD'} = \bar{\Phi}_{\bmA' \bmB' \bmA \bmB}$.
Collecting all this information and decomposing in terms of irreducible
components one obtains the usual decomposition of the curvature
spinors
\[
X_{\bmA \bmB \bmC \bmD} =\Psi_{\bmA \bmB \bmC \bmD} + \Lambda(\epsilon_{\bmD\bmB}\epsilon_{\bmC \bmA}+ \epsilon_{\bmC \bmB}\epsilon_{\bmD \bmA} ), \hspace{1cm}
Y_{\bmA \bmB \bmC' \bmD'}= \Phi_{\bmA \bmB \bmC' \bmD'},
\]
where $\Psi_{\bmA \bmB \bmC \bmD}$ is the Weyl spinor and $ \Phi_{\bmA
\bmB \bmC' \bmD'}$ is the Ricci spinor. The latter is the spinorial
counterpart of a world tensor (because of its Hermiticity which is
consequence of the first Bianchi identity) and $\Lambda$ is a real
scalar (consequence of the first Bianchi identity again).
Additionally, observe that one also has that $ X_{\bmA ( \bmB
\bmC)}{}^{\bmA}=0.$ This is a consequence of the interchange of pairs
symmetry $R_{\bma \bmb \bmc \bmd}= R_{\bmc \bmd \bma \bmb}$ of the
Riemann tensor of a Levi-Civita connection. For a general connection
the right hand side of last equation is not necessarily zero\footnote{
The reason is that the interchange of pairs symmetry is a consequence
of the antisymmetry in the first and second pairs of indices and the
first Bianchi identity which for a general connection will involve the
torsion and its derivatives.}.

\medskip
In the Levi-Civita case, the spinorial Ricci identities are the
spinorial counterpart of 
\[
[\nabla_\bma \nabla_\bmb]v^{\bmc} =R^{\bmc}{}_{\bme \bma
  \bmb}v^{\bme}.
\]
These identities are given in terms of the operator $\square_{\bmA
  \bmB} = \nabla_{\bmQ (\bmA}\nabla_{\bmB){}}{}^{\bmQ}$. The spinorial
Ricci  identities are given in a rather compact form by
\begin{eqnarray}
 \square_{\bmA \bmB}\xi^{\bmC} = X_{\bmA \bmB \bmQ
 }{}^{\bmC}\xi^{\bmQ}, && \square_{\bmA' \bmB'}\xi^{\bmC}=\Phi_{\bmA' \bmB' \bmQ}{}^{\bmC}\xi^{\bmQ},  \\
 \square_{\bmA \bmB}\xi_{\bmC} = -X_{\bmA \bmB
  \bmC}{}^{\bmQ}\xi_{\bmQ},  &&  \square_{\bmA' \bmB'}\xi_{\bmC}=-\Phi_{\bmA' \bmB'\bmC}{}^{\bmQ}{}\xi_{\bmQ}. 
\end{eqnarray}
It is useful to combine these identities with the decomposition of
$X_{\bmA \bmB \bmC \bmD}$ to obtain a more detailed list of
relations. The following expressions are repeatedly used in this
article:
\begin{eqnarray}
&& \square_{\bmA \bmB}\xi_{\bmC}=\Psi_{\bmA \bmB \bmC \bmQ} \xi^{\bmQ}-2\Lambda\xi_{(\bmA}\epsilon_{\bmB)\bmC},  \hspace{0.5cm} \hspace{0.5cm} \square_{(\bmA \bmB}\xi_{\bmC)}= \Psi_{\bmA \bmB \bmC \bmQ}\xi^{\bmQ}, \\
&& \square_{\bmA \bmB}\xi^{\bmB}=-3\Lambda\xi_{\bmA}, \hspace{3.67cm} \square_{\bmA' \bmB'}\xi_{\bmC}=\xi^{\bmQ}\Phi_{\bmQ \bmC \bmA' \bmB'}.
\end{eqnarray}
Using these relations and the Jacobi identity  ($\epsilon$ -identity)
the second Bianchi identity can be expressed in terms of spinors as
\[
\nabla^{\bmA}{}_{\bmB'}X_{\bmA \bmB \bmC \bmD}=\nabla^{\bmA'}{}_{\bmB}\Phi_{\bmC\bmD \bmA'\bmB'}.
\]
For completeness, we recall the relation between $\Phi_{\bmA \bmB \bmA'
\bmB'}$ and the Ricci tensor and between $R$ and $\Lambda$. Namely, that
\begin{equation}
R_{\bma \bmc} \mapsto R_{\bmA \bmA' \bmC \bmC'}= 2\Phi_{\bmA \bmC \bmA' \bmC'}-6\Lambda \epsilon_{\bmA \bmC}\epsilon_{\bmA' \bmC'},
 \qquad  R= -24\Lambda. 
\label{SpinorialCounterpartRicci} 
 \end{equation}
From the above expressions it follows that $2\Phi_{\bmA \bmB \bmA'
\bmB'}$ is the spinorial counterpart of the trace-free Ricci tensor
$R_{\{\bma \bmb\}}\equiv R_{\bma \bmb}-\frac{1}{4}R g_{\bma
\bmb}$. From this last observation, it follows that the spinorial
counterpart of the Schouten tensor can be rewritten in terms of
$\Phi_{\bmA \bmA' \bmB \bmB'}$ and $\Lambda$. Recalling the definition
of the 4-dimensional Schouten tensor $ L_{\bma \bmb}=
\frac{1}{2}R_{\bma \bmb}-\frac{1}{12}Rg_{\bma \bmb} $ and equation
\eqref{SpinorialCounterpartRicci} we get
\begin{equation}
L_{\bmA \bmB \bmA' \bmB'} = \Phi_{\bmA \bmC \bmA' \bmC'}-\Lambda
\epsilon_{\bmA \bmC}\epsilon_{\bmA'
\bmC'}. \label{DecompositionSchouten}
\end{equation}
The second contracted Bianchi identity can be recast in terms of these spinors as
\begin{equation}
\nabla^{\bmC \bmA'}\Phi_{\bmC \bmD \bmA' \bmB'}+ 3\nabla_{\bmD \bmB'}\Lambda =0.
\label{SecondBianchiContracted}
\end{equation}
 
\subsection{Spinorial Ricci identities for a metric connection}
We now consider the case of a connection $\widetriangle{\nabla}$ which
is metric but not torsion free. First, we need to obtain a suitable
generalisation of the operator $\square_{\bmA \bmB}$. In order to
achieve this, observe that the relation
$[\nabla_{\bma},\nabla_{\bmb}]u^{\bmd}=R^{\bmd}{}_{\bmc \bma
\bmb}u^{\bmc}$ valid for a Levi-Civita connection  extends to a
connection with torsion as
 \[
[\widetriangle{\nabla}_{\bma},\widetriangle{\nabla}_{\bmb}]u^{\bmd}=\widetriangle{R}^{\bmd}{}_{\bmc
\bma \bmb}u^{\bmc} +
\Sigma_{\bma}{}^{\bmc}{}_{\bmb}\widetriangle{\nabla}_{\bmc}u^{\bmd} .
\]
Another way to think the last equation is to define a  modified commutator of covariant derivatives through
\[
\llbracket
\widetriangle{\nabla}_{\bma},\widetriangle{\nabla}_{\bmb}\rrbracket u^{\bmd} \equiv \left( [\widetriangle{\nabla}_{\bma},\widetriangle{\nabla}_{\bmb}]- \Sigma_{\bma}{}^{\bmc}{}_{\bmb}\widetriangle{\nabla}_{\bmc} \right) u^{\bmd}.
\]
In this way we can recast the Ricci identities as  
\[
\llbracket \widetriangle{\nabla}_{\bma},\widetriangle{\nabla}_{\bmb}
\rrbracket u^{\bmd}=\widetriangle{R}^{\bmd}{}_{\bmc
  \bma \bmb}u^{\bmc}.
\] 
This observation leads us to an expression for the  generalised operator 
\[
\widetriangle{\square}_{\bmA \bmB} \equiv
\widetriangle{\nabla}_{\bmC'(
  \bmA}\widetriangle{\nabla}_{\bmB)}{}^{\bmC'}. 
\]
The relation between this operator and the commutator of covariant derivatives is
\[
 [\widetriangle{\nabla}_{\bmA \bmA'}, \widetriangle{\nabla}_{\bmB
   \bmB'}] = \epsilon_{\bmA' \bmB'} \widetriangle{\square}_{\bmA \bmB} 
+ \epsilon_{\bmA \bmB}\widetriangle{\square}_{\bmA' \bmB'}.
\]
We cannot directly write down the equivalent spinorial Ricci
identities simply by replacing $X$ and $Y$ by $\widetriangle{X}$ and
$\widetriangle{Y}$ because of appearance of the term $
\Sigma_{\bma}{}^{\bmc}{}_{\bmb}\widetriangle{\nabla}_{\bmc}u^{\bmd}$
in the definition of the curvature tensor. A way to get around this
difficulty is to define a modified operator
$\widetriangle{\boxminus}_{\bmA \bmB}$ formed using the modified
commutator of covariant derivatives instead of the usual
commutator. In this way we can directly translate the previous
formulae simply by replacing $X$ and $Y$ by $\widetriangle{X}$ and
$\widetriangle{Y}$ . Now, the relation between
$\widetriangle{\boxminus}_{\bmA \bmB}$ and
$\widetriangle{\square}_{\bmA \bmB}$ can be obtained by observing that
\begin{eqnarray}
&&\widetriangle{\boxminus}_{\bmC \bmD} = \tfrac{1}{2}\epsilon^{\bmC'
  \bmD'} \llbracket \widetriangle{\nabla}_{\bmC \bmC'},\widetriangle{\nabla}_{\bmD \bmD'} \rrbracket \nonumber
\\
&& \phantom{\widetriangle{\boxminus}_{\bmC \bmD}}=\tfrac{1}{2}\epsilon^{\bmC'\bmD'} \left(
  [\widetriangle{\nabla}_{\bmC \bmC'},\widetriangle{\nabla}_{\bmD
    \bmD'}] -\Sigma_{\bmC \bmC'}{}^{\bmE \bmE'}{}_{\bmD \bmD'}\widetriangle{\nabla}_{\bmE \bmE'} \right)  \nonumber \\
&&\phantom{\widetriangle{\boxminus}_{\bmC \bmD}} = \tfrac{1}{2} \left( \widetriangle{\nabla}_{\bmD'
    \bmC}\widetriangle{\nabla}_{\bmD}{}^{\bmD'} 
+ \widetriangle{\nabla}_{\bmD'\bmD}
\widetriangle{\nabla}_{\bmC}{}^{\bmD'} 
-  \Sigma_{\bmC \bmD'}{}^{ \bmE \bmE'}{}_{\bmD}{}^{\bmD'}\widetriangle{\nabla}_{\bmE \bmE'} \right).
\label{BoxDes1}
\end{eqnarray}
Using the antisymmetry of the torsion spinor we have the decomposition
\begin{equation}
\Sigma_{\bmA \bmA'}{}^{\bmC \bmC'}{}_{\bmB \bmB'}= \epsilon_{\bmA
  \bmB}\Sigma_{\bmA}{}^{\bmE \bmE'}{}_{\bmB} + \epsilon_{\bmA'
  \bmB'}\Sigma_{\bmA'}{}^{\bmE \bmE'}{}_{\bmB'}, 
\label{SigmaDecomposition}
\end{equation}
where the reduced spinor is given by $\Sigma_{\bmA}{}^{\bmE
  \bmE'}{}_{\bmB} = \tfrac{1}{2}\Sigma_{(\bmA | \bmQ'|}{}^{\bmE
  \bmE'}{}_{\bmB)}{}^{\bmQ'} $ . Using this decomposition and
symmetrising  expression \eqref{BoxDes1} in the indices ${}_{\bmC \bmD}$ in one obtains
\[ 
\widetriangle{\boxminus}_{\bmC \bmD}= \widetriangle{\nabla}_{\bmD' (\bmC}\widetriangle{\nabla}_{\bmD)}{}^{\bmD'} - \Sigma_{\bmC}{}^{\bmE \bmE'}{}_{\bmD}\widetriangle{\nabla}_{\bmE \bmE'}  
 = \widetriangle{\square}_{\bmC \bmD} - \Sigma_{\bmC}{}^{\bmE \bmE'}{}_{\bmD} \widetriangle{\nabla}_{\bmE \bmE'}.
\]
Therefore
\begin{equation}
\widetriangle{\square}_{\bmA \bmB}= \widetriangle{\boxminus}_{\bmA
  \bmB}+ \Sigma_{\bmA}{}^{\bmE
  \bmE'}{}_{\bmB}\widetriangle{\nabla}_{\bmE \bmE'}. 
\label{BoxTildeBox}
\end{equation}
 In order to compute explicitly how $\widetriangle{\boxminus}_{\bmA
\bmB}$ acts on spinors we only need to compute the generalised the
spinors $\widetriangle{X}_{\bmA \bmB \bmC \bmD}$ and
$\widetriangle{\Phi}_{\bmA \bmB \bmC' \bmD'}$ .

\medskip
As discussed in previous paragraphs, the fact that the
connection is not torsion free is reflected in the symmetries of the
curvature spinors. We still have that the symmetries in
\eqref{SymmetriesMetricity} hold due to the metricity of
$\widetriangle{\nabla}$. However, the interchange of pairs symmetry of
the Riemann tensor, the reality condition on $\widetriangle{X}_{\bmS
\bmP}{}^{\bmS \bmP}$ and the Hermiticity of $\widetriangle{\Phi}_{\bmA
\bmB \bmC' \bmD'}$ do not longer hold as these properties rely on the
the cyclic identity $R_{\bmd[\bma \bmb \bmc]}=0$. In fact, the first
Bianchi identity is, in general, given  by
\[
\grave{R}^{\bmd}{}_{[\bma \bmb \bmc]} + \grave{\nabla}_{[\bma}\Sigma_{\bmb}{}^{\bmd}{}_{\bmc]}+
\Sigma_{[\bma}{}^{\bme}{}_{\bmb}\Sigma_{\bmc]}{}^{\bmd}{}_{\bme} =0.
\]
It follows that $\widetriangle{X}_{\bmA(\bmB \bmC)}{}^{\bmA}$ does not
necessarily vanish and, generically, it will depend on the torsion and
its derivatives as can be seen from the last equation. Labelling, as
usual, the remaning non-vanishing contractions of
$\widetriangle{X}_{\bmA \bmB \bmC \bmD}$
\[
 \widetriangle{X}_{\bmA \bmB}{}^{\bmA \bmB}=6\widetriangle{\Lambda},
 \qquad 
\widetriangle{X}_{(\bmA \bmB \bmC \bmD)} = \widetriangle{\Psi}_{\bmA \bmB \bmC \bmD}, \qquad
\widetriangle{X}_{\bmA(\bmB \bmC)}{}^{\bmA}=H_{\bmB \bmC},
\]
where $H_{\bmB \bmC}$ is a spinor which, as discussed previously,
depends on the torsion and its derivatives. The explicit form of
$H_{\bmB \bmC}$ will not be needed. Finally, recall the general
decomposition in irreducible terms of a 4-valence spinor $\xi_{\bmA
\bmB \bmC \bmD}$:
\begin{eqnarray*}
&& \xi_{\bmA \bmB \bmC \bmD}= \xi_{(\bmA \bmB \bmC \bmD)}+
\tfrac{1}{2}\xi_{(\bmA \bmB)\bmP}{}^{\bmP} +
\tfrac{1}{2}\xi_{\bmP}{}^{\bmP}{}_{(\bmC \bmD)} \epsilon_{\bmA \bmB} +
\tfrac{1}{4}\xi_{\bmP}{}^{\bmP}{}_{\bmQ}{}^{\bmQ}\epsilon_{\bmA
  \bmB}\epsilon_{\bmC \bmD} \\  
&& \hspace{1.7cm} + \tfrac{1}{2}\epsilon_{\bmA(\bmC}\xi_{\bmD)\bmB} + \tfrac{1}{2}\epsilon_{\bmB(\bmC}\xi_{\bmD)\bmA}-\tfrac{1}{3}\epsilon_{\bmA(\bmC}\epsilon_{\bmD)\bmB}\xi  
\end{eqnarray*}
 where
\[
\xi_{\bmA \bmB}\equiv \xi_{\bmQ(\bmA \bmB)}{}^{\bmQ}, \qquad \xi \equiv \xi_{\bmP \bmQ}{}^{\bmP\bmQ}.
\]
Using the above formula we get the following expressions for the
irreducible decomposition of the curvature spinor $\widetriangle
{X}_{\bmA \bmB \bmC \bmD}$:
\begin{equation}
\widetriangle{X}_{\bmA \bmB \bmC \bmD}= \widetriangle{\Psi}_{\bmA \bmB
  \bmC \bmD} + \widetriangle{\Lambda} \left( \epsilon_{\bmA
    \bmC}\epsilon_{\bmB \bmD} +  \epsilon_{\bmA \bmD}\epsilon_{\bmB
    \bmC} \right) + \tfrac{1}{2}\epsilon_{\bmA(\bmC}H_{\bmD)\bmB} +
\tfrac{1}{2}\epsilon_{\bmB(\bmC}H_{\bmD)\bmA}. 
\label{eqXnew}
\end{equation}
In order to ease the comparisons with the Levi-Civita case let
\begin{equation}
\widetriangle{Y}_{\bmA \bmB \bmC' \bmD'}=\widetriangle{\Phi}_{\bmA
  \bmB \bmC' \bmD'}. 
\label{eqPhinew}
\end{equation}
Observe that, in contrast with the case of a Levi-Civita connection, we have that
$\widetilde{\Lambda}$ is not real and $\widetilde{\Phi}_{\bmA \bmB
\bmC' \bmD'}$ is not Hermitian. In other words, one hast that
\begin{equation}
\widetriangle{\Lambda}-\bar{\widetriangle{\Lambda}} \neq 0, \qquad
\widetriangle{\Phi}_{\bmA \bmB \bmC' \bmD'} -
\bar{\widetriangle{\Phi}}_{\bmA' \bmB' \bmC \bmD } \neq 0, \qquad 
H_{\bmA \bmB} \neq 0. 
\label{NoLevi}
\end{equation}
In fact, the right hand side of the previous equations depends on the
torsion and its derivatives ---see \cite{Pen83}. However, its explicit
expression will not play any role in the discussion of this
article. Having found the curvature spinors, we can derive the
spinorial Ricci identities. As discussed in the previous paragraph, the
modified operator $\widetriangle{\boxminus}_{\bmA \bmB}$ formed from
the modified commutator of covariant derivatives satisfies as version
of the
spinorial Ricci identities which is obtained simply by replacing the curvature spinors
$X_{\bmA \bmB \bmC \bmD}$ and $\Phi_{\bmA \bmB \bmC' \bmD'}$ by the
spinors $\widetriangle{X}_{\bmA \bmB \bmC \bmD}$ and
$\widetriangle{\Phi}_{\bmA \bmB \bmC' \bmD'}$. The Ricci identities with torsion are then given by
\begin{eqnarray*}
 \widetriangle{\square}_{\bmA \bmB}\xi^{\bmC} = \widetriangle{X}_{\bmA
   \bmB \bmQ }{}^{\bmC}\xi^{\bmQ}  + \Sigma_{\bmA}{}^{\bmP
   \bmP'}{}_{\bmB}\widetriangle{\nabla}_{\bmP \bmP'}\xi^{\bmC}, &&
 \widetriangle{\square}_{\bmA'
   \bmB'}\xi^{\bmC}=\bar{\widetriangle{\Phi}}_{\bmA' \bmB'
   \bmQ}{}^{\bmC}\xi^{\bmQ}  + \bar{\Sigma}_{\bmA'}{}^{\bmP
   \bmP'}{}_{\bmB'} \widetriangle{\nabla}_{\bmP \bmP'}\xi^{\bmC},  \\
\widetriangle{\square}_{\bmA \bmB}\xi_{\bmC} = -\widetriangle{X}_{\bmA
  \bmB \bmC}{}^{\bmQ}\xi_{\bmQ} + \Sigma_{\bmA}{}^{\bmP
  \bmP'}{}_{\bmB}\widetriangle{\nabla}_{\bmP \bmP'}\xi_{\bmC}, 
&& \widetriangle{\square}_{\bmA' \bmB'}\xi_{\bmC}=-\bar{\widetriangle{\Phi}}_{\bmA' \bmB'\bmC}{}^{\bmQ}{}\xi_{\bmQ}  +\bar{\Sigma}_{\bmA'}{}^{\bmP \bmP'}{}_{\bmB'}\widetriangle{\nabla}_{\bmP \bmP'}\xi_{\bmC},
\end{eqnarray*}
with $\widetriangle{X}_{\bmA \bmB \bmC \bmD}$ and
$\widetriangle{\Phi}_{\bmA \bmB \bmC' \bmD'}$ given by \eqref{eqXnew}
and \eqref{eqPhinew}. The primed version of the last expressions can
be readily identified. More importantly, the detailed version (in
terms of irreducible components) of the spinorial Ricci identities
become
\begin{subequations}
\begin{eqnarray}\label{TorsionSpinorialRicciIdentities}
&& \widetriangle{\square}_{\bmA
\bmB}\xi_{\bmC}=\widetriangle{\Psi}_{\bmA \bmB \bmC \bmQ}
\xi^{\bmQ}-2\widetriangle{\Lambda}\xi_{(\bmA}\epsilon_{\bmB)\bmC} +
{U}_{\bmA \bmB \bmC \bmQ} \xi^{\bmQ} +\Sigma_{\bmA}{}^{\bmP
\bmP'}{}_{\bmB}\widetriangle{\nabla}_{\bmP \bmP'}\xi_{\bmC}, \\ 
&& \widetriangle{\square}_{(\bmA \bmB}\xi_{\bmC)}=
\widetriangle{\Psi}_{\bmA \bmB \bmC \bmQ}\xi^{\bmQ}
+\Sigma_{(\bmA}{}^{\bmP \bmP'}{}_{\bmB}\widetriangle{\nabla}_{|\bmP
\bmP'|}\xi_{\bmC)}, \\ 
&& \widetriangle{\square}_{\bmA
\bmB}\xi^{\bmB}=-3\widetriangle{\Lambda}\xi_{\bmA} + H_{\bmA
\bmB}\xi^{\bmB} + \Sigma_{\bmA}{}^{\bmP \bmP'}{}_{\bmB}
\widetriangle{\nabla}_{\bmP \bmP'} \xi^{\bmB}, \\ 
&&\widetriangle{\square}_{\bmA'
\bmB'}\xi_{\bmC}=\xi^{\bmQ}\bar{\widetriangle{\Phi}}_{\bmQ \bmC \bmA'
\bmB'} +\bar{\Sigma}_{\bmA'}{}^{\bmP
\bmP'}{}_{\bmB'}\widetriangle{\nabla}_{\bmP \bmP'}\xi_{\bmC}.
\end{eqnarray}
\end{subequations}
The above identities are supplemented by their complex conjugated
version ---keeping in mind the non-Hermiticity of
$\widetilde{\Phi}_{\bmA \bmB \bmC' \bmD'}$ and the non-reality of
$\Lambda$ as stated in \eqref{NoLevi}. In the last list of identities
we have defined in the quantity $U_{\bmA \bmB \bmC \bmD}$ 
\begin{equation} 
U_{\bmA \bmB \bmC \bmD}\equiv
\tfrac{1}{2}\epsilon_{\bmA(\bmC}H_{\bmD)\bmB} +
\tfrac{1}{2}\epsilon_{\bmB(\bmC}H_{\bmD)\bmA}. 
\label {decomU}
\end{equation}
The Levi-Civita case can be readily recovered by setting
$\Sigma_{\bmA}{}^{\bmP \bmP'}{}_{\bmB}=0$ since the spinors $H_{\bmA
\bmB}$ and $U_{\bmA \bmB \bmC \bmD}$ also vanish. Moreover, we also
recover the pair interchange symmetry and the expressions
presented in \eqref{NoLevi} become equalities.

\section{Appendix: the transition tensor and the torsion tensor}
\label{Appendix:Torsion}

In this appendix we briefly derive the transition spinor relating a
Levi-Civita connection $\bmnabla$ with a connection
$\widetriangle{\bmnabla}$ which is metric but not necessarily
torsion-free. The general strategy behind this discussion can be found
in \cite{PenRin84}. Given two general connections  $\grave{\bmnabla}$ and $\acute{\bmnabla}$  we have that
\[
(\grave{\nabla}_{\bma}-\acute{\nabla}_{\bma})\xi^{\bmb}\equiv Q_{\bma}{}^{\bmb}{}_{\bmc}\xi^{\bmc}
\]
where $Q_{\bma}{}^{\bmb}{}_{\bmc}$ is the transition tensor. It is
well known that for the case of a Levi-Civita connection $\bmnabla$
and a metric connection $\widetriangle{\bmnabla}$  one has that
\begin{eqnarray}\label{transitionTorsion}
 \Sigma_{\bma}{}^{\bmc}{}_{\bmb} = -2Q_{[\bma}{}^{\bmc}{}_{\bmb]}  \qquad
 Q_{\bma \bmb \bmc}= Q_{\bma[\bmb \bmc]}.
\end{eqnarray}
Therefore, the spinorial counterpart of the transition tensor can be decomposed as
\begin{equation}
 Q_{\bmA \bmA' \bmB \bmB' \bmC \bmC'}= Q_{\bmA \bmA' \bmB
   \bmC}\epsilon_{\bmB' \bmC'}+ \bar{Q}_{\bmA \bmA' \bmB'
   \bmC'}\epsilon_{\bmB \bmC}
\label{decompositionTransition}
\end{equation}
 where 
\[
Q_{\bmA \bmA' \bmB \bmC}\equiv \tfrac{1}{2}Q_{\bmA \bmA'
  (\bmB|\bmQ'|\bmC)}{}^{\bmQ'}. 
\]
This expression allows to translate expressions containing the
covariant derivative $\widetriangle{\bmnabla}$ to expressions
containing $\bmnabla$ and the transition spinor $Q_{\bmA \bmA' \bmB
\bmC}$ as follows:
\begin{eqnarray}
  \widetriangle{\nabla}_{\bmA \bmA'}\xi^{\bmB}= \nabla_{\bmA \bmA'}{\xi}^{\bmB} + Q_{\bmA \bmA'}{}^{\bmB}{}_{\bmQ}\xi^{\bmQ},  \qquad
\widetriangle{\nabla}_{\bmA \bmA'}\xi_{\bmC}= \nabla_{\bmA
  \bmA'}\xi_{\bmC} - Q_{\bmA \bmA'}{}^{\bmQ}{}_{\bmC}\xi_{\bmQ}. \label{defApplyTransition}
\end{eqnarray}
These expression can be extended in a similar manner to spinors of any
index structure. Now, from the equations in \eqref{transitionTorsion}
it follows that
\begin{equation}
Q_{\bma \bmc \bmb}= -\Sigma_{\bma[\bmc \bmb]}-\tfrac{1}{2}\Sigma_{\bmc
  \bma \bmb}.
\label{transitionTensor}
\end{equation}
Using the above equation along with the decompositions
\eqref{decompositionTransition} and \eqref{SigmaDecomposition} we get
\[
Q_{\bmA \bmA' \bmB \bmC} = -2\Sigma_{(\bmB |\bmA \bmA'|
  \bmC)}-2\Sigma_{\bmA(\bmC|\bmA'|\bmB)}-2\bar{\Sigma}_{\bmA'(\bmC|\bmQ'}{}^{\bmQ'}{}\epsilon_{\bmA|\bmB)}. 
\]

\section{Appendix: explicit expressions for the subsidiary equations}
\label{Appendix:SubsidiarySystem}

In Section \ref{IntroPropSubsidiary} it was shown that the generic
form of the equations in the subsidiary system is
 \[
\widetriangle{\square}\widehat{N}_{\bmA
\bmB\mathcal{K}}=2\widetriangle{\square}_{\bmP \bmA}
\widehat{N}^{\bmP}{}_{\bmB \mathcal{K}} -2\widetriangle{\nabla}_{\bmA
\bmQ'}W^{\bmQ'}{}_{\bmB \mathcal{K}} .
\]
In this section we use the results of Appendices
\ref{Appendix:SpinorialRelations} and \ref{Appendix:Torsion} to
explicitly compute the terms $\widetriangle{\square}_{\bmP \bmA}
\widehat{N}^{\bmP}{}_{\bmB \mathcal{K}}$ and $W^{\bmQ'}{}_{\bmB
\mathcal{K}}$ for every zero-quantity. One has that:
\begin{eqnarray*}
&&\widetriangle{\square}_{\bmP
  \bmA}\widehat{\Sigma}{}^{\bmP}{}_{\bmB}{}^{\bmc}=-3\widetriangle{\Lambda}\widehat{\Sigma}_{\bmA
  \bmB}{}^{\bmc}+ H_{\bmP
  \bmA}\widehat{\Sigma}{}^{\bmP}{}_{\bmB}{}^{\bmc} +
\widetriangle{\Psi}_{\bmP \bmA \bmB \bmG}\widehat{\Sigma}^{\bmP \bmG
  \hspace{1mm}\bmc}-2\widetriangle{\Lambda}\widehat{\Sigma}^{\bmP
}{}_{(\bmP}{}^{\bmc}{}\epsilon_{\bmA)\bmB} \\
&&\hspace{2.5cm} +U_{\bmP\bmA\bmB \bmQ}\widehat{\Sigma}^{\bmP
}{}_{(\bmP}{}^{\bmc}{}\epsilon_{\bmA)\bmB} + 2\Sigma_{\bmP}{}^{\bmQ
\bmQ'}{}_{\bmA} \widetriangle{\nabla}_{\bmQ
\bmQ'}\widehat{\Sigma}^{\bmP}{}_{\bmB}{}^{\bmc}, \\
  && \widetriangle{\square}_{\bmP' \bmC'} \widehat{\Xi}_{\bmA \bmB}{}^{\bmP'}{}_{\bmD'} =
\widehat{\Xi}^{\bmQ}{}_{\bmB}{}^{\bmP'}{}_{\bmD'}\bar{\widetriangle{\Phi}}_{\bmQ
  \bmA \bmP' \bmC'} + \bar{\Sigma}_{\bmP'}{}^{\bmQ
  \bmQ'}{}_{\bmC'}\widetriangle{\nabla}_{\bmQ
  \bmQ'}\widehat{\Xi}_{\bmA \bmB}{}^{\bmP'}{}_{\bmC'}   \\
&& \hspace{3cm}+ \widehat{\Xi}_{\bmA}{}^{\bmQ
  \bmP'}{}_{\bmD'}\bar{\widetriangle{\Phi}}_{\bmA \bmQ \bmP' \bmC'} +
\bar{\Sigma}_{\bmP'}{}^{\bmQ
  \bmQ'}{}_{\bmC'}\widetriangle{\nabla}_{\bmQ
  \bmQ'}\widehat{\Xi}_{\bmA \bmB}{}^{\bmP'}{}_{\bmC'}\\
&&\hspace{3cm}-\bar{\widetriangle{\Lambda}}\widehat{\Xi}_{\bmA \bmB
  \bmC' \bmD'} + \bar{H}_{\bmP' \bmC'}\widehat{\Xi}_{\bmA
  \bmB}{}^{\bmP'}{}_{\bmD'} + \bar{\Sigma}_{\bmP'}{}^{\bmQ
  \bmQ'}{}_{\bmC'}\widetriangle{\nabla}_{\bmQ
  \bmQ'}\widehat{\Xi}_{\bmA \bmB}{}^{\bmP'}{}_{\bmD'}\\
&&\hspace{3cm}+ \bar{\widetriangle{\Psi}}_{\bmP' \bmC' \bmD'
  \bmQ'}\widehat{\Xi}_{\bmA \bmB}{}^{\bmP'
  \bmQ'}-\bar{\widetriangle{\Lambda}}\widehat{\Xi}_{\bmA
  \bmB}{}^{\bmP'}{}_{(\bmP'}\epsilon_{\bmC')\bmD'} + \bar{U}_{\bmP'
  \bmC' \bmD' \bmQ'}\widehat{\Xi}_{\bmA \bmB}{}^{\bmP' \bmQ'} \\
&& \hspace{3cm} + \bar{\Sigma}_{\bmP'}{}^{\bmQ
  \bmQ'}{}_{\bmC'}\widetriangle{\nabla}_{\bmQ
  \bmQ'}\widehat{\Xi}_{\bmA \bmB}{}^{\bmP'}{}_{\bmD'}, \\ 
&&  \widetriangle{\square}_{\bmP \bmC} \widehat{\Delta}^{\bmP}{}_{\bmD
  \bmB \bmB'}  = -3\widetriangle{ \Lambda} \widehat{\Delta}_{\bmC \bmD
  \bmB \bmB'} + H_{\bmP \bmC}\widehat{\Delta}^{\bmP}{}_{\bmD \bmB
  \bmB'} + \Sigma_{\bmP}{}^{\bmQ
  \bmQ'}{}_{\bmC}\widetilde{\nabla}_{\bmQ
  \bmQ'}\widehat{\Delta}^{\bmP}{}_{\bmD \bmB \bmB'}   +
\widetriangle{\Psi}_{\bmP \bmC \bmD \bmQ}\widehat{\Delta}^{\bmP
  \bmQ}{}_{\bmB \bmB'} \\
&& \hspace{2.7cm}
-2\widetriangle{\Lambda}\widehat{\Delta}^{\bmP}{}_{(\bmP|\bmD \bmB
  \bmB'|}\epsilon_{\bmC)\bmD} + U_{\bmP \bmC \bmD
  \bmQ}\widehat{\Delta}^{\bmP \bmQ}{}_{\bmB \bmB'}+
\Sigma_{\bmP}{}^{\bmQ \bmQ'}{}_{\bmC}\widetriangle{\nabla}_{\bmQ
  \bmQ'}\widehat{\Delta}^{\bmP}{}_{\bmD \bmB \bmB'}  \\
&&  \hspace{2.7cm} + \widetriangle{\Psi}_{\bmP \bmC \bmB
  \bmQ}\widehat{\Delta}^{\bmP}{}_{\bmD}{}^{\bmQ}{}_{\bmB'}
-2\widetriangle{\Lambda} \widehat{\Delta}^{\bmP}{}_{\bmD
  (\bmP|\bmB'|}\epsilon_{\bmC)\bmB}+  U_{\bmP \bmC \bmB
  \bmQ}\widehat{\Delta}^{\bmP}{}_{\bmD}{}^{\bmQ}{}_{\bmB \bmB'}  \\ 
&& \hspace{2.7cm} + \Sigma_{\bmP}{}^{\bmQ
  \bmQ'}{}_{\bmC}\widetriangle{\nabla}_{\bmQ
  \bmQ'}\widehat{\Delta}^{\bmP}{}_{\bmD \bmB \bmB'}  +
\widehat{\Delta}^{\bmP}{}_{\bmD
  \bmB}{}^{\bmQ'}{\widetriangle{\Phi}}_{\bmQ' \bmB' \bmP \bmC} +
\Sigma_{\bmP}{}^{\bmQ \bmQ'}{}_{\bmC}\widetriangle{\nabla}_{\bmQ
  \bmQ'}\widehat{\Delta}^{\bmP}{}_{\bmD \bmB \bmB'},\\  
&& \widetriangle{\square}_{\bmP'
  \bmB'}\Lambda_{\bmB}{}^{\bmP'}{}_{\bmA \bmC}   =  \Lambda^{\bmQ
  \bmP'}{}_{\bmA \bmC}\bar{\widetriangle{\Phi}}_{\bmQ \bmB \bmP'
  \bmB'} + \bar{\Sigma}_{\bmP'}{}^{\bmQ
  \bmQ'}{}_{\bmB'}\widetriangle{\nabla}_{\bmQ
  \bmQ'}\Lambda_{\bmB}{}^{\bmP'}{}_{\bmA \bmC}   -
3\widetriangle{\Lambda} \Lambda_{\bmB \bmP' \bmA \bmC}  \\
&&  \hspace{2.7cm} +  \bar{H}_{\bmP'
  \bmB'}\Lambda_{\bmB}{}^{\bmP'}{}_{\bmA \bmC}  +
\bar{\Sigma}_{\bmP'}{}^{\bmQ
  \bmQ'}{}_{\bmB'}\widetriangle{\nabla}_{\bmQ
  \bmQ'}\Lambda_{\bmB}{}^{\bmP'}{}_{\bmA \bmC}  +
\Lambda_{\bmB}{}^{\bmP'\bmQ}{}_{\bmC} \bar{\widetriangle{\Phi}}_{\bmQ
  \bmA \bmP' \bmB'}  \\
&& \hspace{2.7cm} + \bar{\Sigma}_{\bmP'}{}^{\bmQ \bmQ'}{}_{\bmB'}\widetriangle{\nabla}_{\bmQ \bmQ'}\Lambda_{\bmB}{}^{\bmP'}{}_{\bmA \bmC}     + \Lambda_{\bmB}{}^{\bmP'}{}_{\bmA}{}^{\bmQ}\Phi_{\bmQ \bmC \bmP' \bmB'} +   \bar{\Sigma}_{\bmP'}{}^{\bmQ \bmQ'}{}_{\bmB'}\widetriangle{\nabla}_{\bmQ \bmQ'}\Lambda_{\bmB}{}^{\bmP'}{}_{\bmA \bmC}.  
\end{eqnarray*}
Moreover, one has that
\begin{eqnarray*}
&& W[\Sigma]{}^{\bmQ'}{}_{\bmB}{}{}^{\bmc} \equiv
Q^{\bmQ'}{}_{\bmE}{}^{\bmE}{}_{\bmF}\widehat{\Sigma}^{\bmF}{}_{\bmB}{}^{\bmc}-Q^{\bmQ'}{}_{\bmE
\bmB}{}^{\bmF}\widehat{\Sigma}^{\bmE}{}_{\bmF}{}^{\bmc}, \\ 
&& W[\Xi]{}^{\bmQ}{}_{\bmA \bmB \bmD'} \equiv -Q^{\bmQ}{}_{\bmE'
\bmA}{}^{\bmF}\widehat{\Xi}_{\bmF \bmB}{}^{\bmE'}{}_{\bmD}
-Q^{\bmQ}{}_{\bmE'\bmB}{}^{\bmF}\widehat{\Xi}_{\bmA
\bmF}{}^{\bmE'}{}_{\bmD} +
\bar{Q}^{\bmQ}{}_{\bmE'}{}^{\bmE'}{}_{\bmF'}\widehat{\Xi}_{\bmA\bmB}{}^{\bmF'}{}_{\bmD}
-Q^{\bmQ}{}_{\bmE' \bmD}{}^{\bmF}\widehat{\Xi}_{\bmA
\bmB}{}^{\bmE'}{}_{\bmF},\\
&& W[\Delta]{}^{\bmQ'}{}_{\bmD \bmB \bmB'} \equiv
Q^{\bmQ'}{}_{\bmE}\widehat{\Delta}^{\bmF}{}_{\bmD \bmB \bmB'} -
Q^{\bmQ'}{}_{\bmE \bmD}{}^{\bmF}\widehat{\Delta}^{\bmE}{}_{\bmF \bmB
\bmB'} - Q^{\bmQ'}{}_{\bmE
\bmB}{}^{\bmF}\widehat{\Delta}^{\bmE}{}_{\bmD \bmF \bmB'} -
\bar{Q}^{\bmQ'}{}_{\bmE
\bmB'}{}^{\bmF'}\widehat{\Delta}^{\bmE}{}_{\bmD \bmB \bmF'},\\
&& W[\Lambda]{}^{\bmQ}{}_{\bmB \bmA \bmC}\equiv
-Q_{\bmE'}{}^{\bmG}{}_{\bmB}{}^{\bmF}\Lambda_{\bmF}{}^{\bmE'}{}_{\bmA
\bmC} + \bar{Q}_{\bmE'}{}^{\bmG
\bmE'}{}_{\bmF'}\Lambda_{\bmB}{}^{\bmF'}{}_{\bmA \bmC} -
Q_{\bmE'}{}^{\bmG}{}_{\bmA}{}^{\bmF}\Lambda_{\bmB}{}^{\bmE'}{}_{\bmF
\bmC}-Q_{\bmE'}{}^{\bmG}{}_{\bmC}{}^{\bmF}\Lambda_{\bmB}{}^{\bmE}{}_{\bmA
\bmF},\\ 
&& W[Z]{}^{\bmA \bmA'}{}_{\bmA \bmA'}\equiv-Q_{\bmA
\bmA'}{}^{\bmA}{}_{\bmE}Z^{\bmE \bmA'} - \bar{Q}_{\bmA
\bmA'}{}^{\bmA'}{}_{\bmE}Z^{\bmE \bmE'}, \\
\end{eqnarray*}
where the transition spinor is understood to be expressed in terms of the reduced
torsion spinor which is, in itself, a zero-quantity ---see equation \eqref{TransitionSpinor} .

\section{Appendix: the frame conformal Einstein field equations}
\label{Appendix:CFE}

The tensorial (frame) version of the standard vacuum conformal 
Einstein field equations are  given by the following system---see e.g. \cite{Fri81a,Fri81b,Fri82,Fri83}:
\begin{subequations}
\begin{eqnarray}
&& \Sigma_{\bma}{}^{\bmc}{}_{\bmb}e_{\bmc}=0 , \label{CEFEnoTorsion} \\
&&  \nabla_{\bme}d{}^{\bme}{}_{\bma\bmb\bmf}{}=0 , \label{CEFErescaledWeyl} \\
&&  \nabla_{\bmc}L_{\bmd\bmb}-\nabla_{\bmd}L_{\bmb\bmc} -\nabla_{\bma}\Xi
 d{}^{\bma}{}_{\bmb\bmc\bmd} =0 , \label{CEFESchouten}\\
&&  \nabla_{\bma}\nabla_{\bmb}\Xi  +\Xi L_{\bma\bmb} - s g_{\bma\bmb}=0 ,
 \label{CEFEsecondderivativeCF}\\
&&  \nabla_{\bma}s +L_{\bma\bmc} \nabla ^{\bmc}\Xi=0 , \label{CEFEs}\\
&&  R^{\bmc}{}_{\bma\bmb\bmd} -  \rho^{\bmc}{}_{\bma\bmb\bmd}=0,  \label{CEFERiemann}
\end{eqnarray}
\end{subequations}
where $\Sigma_{\bma}{}^{\bmc}{}_{\bmb}$ is the torsion tensor, given in terms of the connection coefficients, as
\[
\Sigma_{\bma}{}^{\bmc}{}_{\bmb}\bme_{\bmc} = [\bme_{\bma},\bme_{\bmb}]-
(\Gamma_{\bma}{}^{\bmc}{}_{\bmb}-\Gamma_{\bmb}{}^{\bmc}{}_{\bma})\bme_{\bmc}, 
\]
$L_{\bma \bmb}$ is the Schouten tensor, $\Xi$ is the conformal factor,
$s$ is the concomitant of the conformal factor defined by 
\[
s\equiv \tfrac{1}{4}\nabla_{\bma}\nabla^{\bma}\Xi + \tfrac{1}{24}R\Xi.
\]
In addition, $\rho^{\bma}{}_{\bmb \bmc \bmd}$ is the algebraic curvature and $R^{\bmc}{}_{\bmd\bma\bmb}$ is the geometric curvature.
\begin{eqnarray*}
&& \rho^{\bma}{}_{\bmb\bmc\bmd}{} =\Xi d^{\bma}{}_{\bmb\bmc\bmd}{} +
2( g^{\bma}{}_{[\bmc}L_{\bmd]\bmb} - g_{\bmb[ \bmc}L_{\bmd]}{}^{\bma}
), \\ && R^{\bmc}{}_{\bmd \bma
\bmb}=\bme_\bma(\Gamma_{\bmb}{}^{\bmc}{}_{\bmd})-\bme_{\bmb}(\Gamma_{\bma}{}^{\bmc}{}_{\bmd})+\Gamma_{\bmf}{}^{\bmc}{}_{\bmd}(\Gamma_{\bmb}{}^{\bmf}{}_{\bma}-\Gamma_{\bma}{}^{\bmf}{}_{\bmb})
+
\Gamma_{\bmb}{}^{\bmf}{}_{\bmd}\Gamma_{\bma}{}^{\bmc}{}_{\bmf}-\Gamma_{\bma}{}^{\bmf}{}_{\bmd}\Gamma_{\bmb}{}^{\bmc}{}_{\bmf}-\Sigma_{\bma}{}^{\bmf}{}_{\bmb}\Gamma_{\bmf}{}^{\bmc}{}_{\bmd}.
\end{eqnarray*}

\section{Appendix: basic existence and stability theory for quasilinear wave equations}
\label{Appendix:PDETheory}

In this appendix we will give an adapted version of a theorem for quasilinear wave equations given in \cite{HugKatMar77}.

\subsubsection{General set up}
In what follows, we will consider open, connected subsets $\mathcal{U} \subset
\mathcal{M}_T\equiv [0,T)\times \mathcal{S}$  for some $T>0$ and
$\mathcal{S}$ an oriented, compact 3-dimensional manifold. Mostly, we
will have $\mathcal{S}\approx \mathbb{S}^3$. On $\mathcal{U}$ one can
introduce local coordinates $x=(x^\mu)=(t,x^\alpha)$. Given a fixed
$N\in \mathbb{N}$, in what follows, let $\mathbf{u}: \mathcal{M}_T
\rightarrow \mathbb{C}^N$ denote a $\mathbb{C}^N$-valued function. The
derivatives of $\mathbf{u}$ will be denoted, collectively, by
$\partial \mathbf{u}$. We will consider quasilinear
wave equations of the form
\begin{equation}
g^{\mu\nu}(x;\mathbf{u}) \partial_\mu \partial_\nu \mathbf{u} =
\mathbf{F}(x;\mathbf{u},\partial \mathbf{u}), 
\label{QuasilinearEquation}
\end{equation}
where $g^{\mu\nu}(x;\mathbf{u})$ denotes the contravariant version of a
Lorentzian metric $g_{\mu\nu}(x;\mathbf{u})$ which depends smoothly on
the unknown $\mathbf{u}$ and the coordinates $x$ and $\mathbf{F}$ is a smooth
$\mathbb{C}^N$-valued function of its arguments. In order to formulate
a Cauchy problem for equation \eqref{QuasilinearEquation} it is
necessary to supplement it with initial data corresponding to the
value of $\mathbf{u}$ and $\partial_t \mathbf{u}$ on the initial
hypersurface $\mathcal{S}$. For simplicity,
choose coordinates such that $\mathcal{S}$ is described by the
condition $t=0$. Given two functions $\mathbf{u}_\star, \,
\mathbf{v}_\star\in H^m(\mathcal{S},\mathbb{C}^N)$, $m\geq 2$, one 
defines the ball of radius $\varepsilon$ centred around
$(\mathbf{u}_\star,\mathbf{v}_\star)$ as the set
\[
B_\varepsilon (\mathbf{u}_\star,\mathbf{v}_\star) \equiv \big\{
(\mathbf{w}_1,\mathbf{w}_2) \in  H^m(\mathcal{S},\mathbb{C}^N)
\times H^m(\mathcal{S},\mathbb{C}^N) \; | \; \parallel
\mathbf{w}_1 - \mathbf{u}_\star\parallel_{\mathcal{S},m}
+ \parallel \mathbf{w}_2 -
\mathbf{v}_\star\parallel_{\mathcal{S},m} \leq \varepsilon\big\}.
\]
Also, given $\delta>0$ define
\[
D_\delta \equiv \big\{ (\mathbf{w}_1,\mathbf{w}_2)\in
H^m(\mathcal{S},\mathbb{C}^N)\times H^m(\mathcal{S},\mathbb{C}^N) \; | \;  \delta < | \det
g_{\mu\nu} (\mathbf{w}_1)|  \big\}.
\]

The basic existence and Cauchy stability theory for equations of the
form \eqref{QuasilinearEquation} has been given in
\cite{HugKatMar77}, from where we adapt the following result:

\begin{theorem}
\label{Theorem:HugKatMar}
Given an orientable, compact, 3-dimensional manifold
$\mathcal{S}$, consider the the Cauchy problem
\begin{eqnarray*}
&& g^{\mu\nu}(x;\mathbf{u}) \partial_\mu \partial_\nu \mathbf{u} =
\mathbf{F}(x;\mathbf{u},\partial \mathbf{u}), \\
&& \mathbf{u}(0,x) = \mathbf{u}_\star(x)\in
H^m(\mathcal{S},\mathbb{C}^N), \\
&& \partial_t \mathbf{u}(0,x) = \mathbf{v}_\star(x)\in
 H^{m}(\mathcal{S},\mathbb{C}^N), \qquad m\geq 4,
\end{eqnarray*}
and assume that $g_{\mu\nu}(x;\mathbf{u}_\star)$ is a Lorentzian metric such
that $(\mathbf{u}_\star,\mathbf{v}_\star )\in D_\delta$
for some $\delta>0$. Then:
\begin{itemize}
\item[(i)] There exists $T>0$ and a unique solution to the Cauchy
  problem defined on $[0,T)\times \mathcal{S}$ such that
\[
\mathbf{u} \in C^{m-2}([0,T)\times \mathcal{S}, \mathbb{C}^N).
\]
Moreover, $(\mathbf{u}(t,\cdot),\partial_t \mathbf{u}(t,\cdot))\in D_\delta$ for $t\in[0,T)$. 

\item[(ii)] There is a $\varepsilon>0$ such that a common existence time
  $T$ can be chosen for all initial data conditions on
  $B_\varepsilon(\mathbf{u}_\star,\mathbf{v}_\star) \cap D_\delta$. 

\item[(iii)] If the solution $\mathbf{u}$ with initial data
  $\mathbf{u}_\star$ exists on $[0,T)$ for some $T>0$, then the
  solutions to all initial conditions in
  $B_\varepsilon(\mathbf{u}_\star,\mathbf{v}_\star)\cap D_\delta$ exist on $[0,T]$ if $\varepsilon>0$
  is sufficiently small.

\item[(iv)] If $\varepsilon$ and $T$ are chosen as in $(i)$ and one has a
  sequence $(\mathbf{u}_\star^{(n)},\mathbf{v}^{(n)}_\star) \in
  B_\varepsilon(\mathbf{u}_\star,\mathbf{v}_\star) \cap D_\delta$ such that
\[
\parallel \mathbf{u}_\star^{(n)} -
\mathbf{u}_\star\parallel_{\mathcal{S},m} \rightarrow 0, \qquad \parallel \mathbf{v}_\star^{(n)} -
\mathbf{v}_\star\parallel_{\mathcal{S},m} \rightarrow 0, \qquad
\mbox{as} \qquad n\rightarrow \infty,
\]
then for the solutions $\mathbf{u}^{(n)}(t,\cdot)$ with
$\mathbf{u}^{(n)}(0,\cdot)=\mathbf{u}_\star^{(n)}$ and $\partial_t \mathbf{u}^{(n)}(0,\cdot)=\mathbf{v}_\star^{(n)}$, it holds that 
\[
\parallel
\mathbf{u}^{(n)}(t,\cdot)-\mathbf{u}(t,\cdot) \parallel_{\mathcal{S},m}
\rightarrow 0 \qquad \mbox{as} \qquad n\rightarrow \infty 
\]
uniformly in $t\in[0,t)$ as $n\rightarrow \infty$. 

\end{itemize}

\end{theorem}



\end{document}